\documentclass[11pt]{article}

\usepackage{geometry}
\usepackage{amsmath,amsthm}
\usepackage{booktabs} 

\usepackage{physics} 
\usepackage{mathtools}
\usepackage{thmtools, thm-restate} 
\usepackage{url} 
\usepackage{bbm} 
\usepackage{enumitem}
\usepackage{listings}
\usepackage{xcolor}
\usepackage[hang,flushmargin]{footmisc} 
\usepackage{hyperref} 
\usepackage{subcaption} 

\usepackage{algorithm}
\usepackage[noend]{algpseudocode}

\let\oldReturn\Return
\renewcommand{\Return}{\State\oldReturn}

\usepackage{tikz}

\newcommand{\etal}{\textit{et al}. }

\DeclareMathOperator*{\E}{\mathbbm{E}} 
\let\S\relax 

\newcommand{\D}{\mathcal{G}} 
\newcommand{\A}{\mathcal{A}} 
\newcommand{\supp}{\text{supp}} 
\newcommand{\values}{\text{values}} 

\newcommand{\S}{\theta} 
\newcommand{\Z}{\Theta} 
\newcommand{\weight}{\omega} 
\newcommand{\M}{\mathcal{M}} 

\newcommand{\R}{\textsf{R}} 
\newcommand{\W}{\textsf{SW}} 
\newcommand{\CS}{\textsf{CS}} 
\newcommand{\opt}{\CS^\ast} 

\newcommand{\area}{\textsf{Area}} 
\newcommand{\bbm}{} 

\let\epsilon\varepsilon 

\newtheorem{theorem}{Theorem}[section]
\newtheorem{definition}{Definition}[section]

\newtheorem{lemma}[theorem]{Lemma}
\newtheorem{observation}[theorem]{Observation}
\newtheorem{example}{Example}[section]

\begin{document}


\title{Optimal Price Discrimination for Randomized Mechanisms\thanks{This work is supported by NSF grant CCF-2113798.}}

\author{Shao-Heng Ko\thanks{Department of Computer Science, Duke University, Durham NC 27708-0129. Email: {\tt shaoheng.ko@duke.edu}. } \and Kamesh Munagala\thanks{Department of Computer Science, Duke University, Durham NC 27708-0129. Email: {\tt kamesh@cs.duke.edu}.}}
\date{}

\maketitle

\begin{abstract}
    We study the power of price discrimination via an intermediary in bilateral trade, when there is a revenue-maximizing seller selling an item to a buyer with a private value drawn from a prior. Between the seller and the buyer, there is an intermediary that can \emph{segment} the market by releasing information about the true values to the seller. This is termed signaling, and enables the seller to price discriminate. In this setting, Bergemann \textit{et al}.~\cite{Bergemann15} showed the existence of a signaling scheme that simultaneously raises the optimal consumer surplus, guarantees the item always sells, and ensures the  seller's revenue does not increase. 
    
    Our work extends the positive result of Bergemann \textit{et al}. to settings where the type space is larger, and where optimal auction is randomized, possibly over a menu that can be exponentially large. In particular, we consider two settings motivated by budgets: The first is  when there is a publicly known budget constraint on the price the seller can charge~\cite{Chawla11} and the second is the FedEx problem~\cite{Fiat16} where the buyer has a private deadline or service level (equivalently, a private budget that is guaranteed to never bind). For both settings, we present a novel signaling scheme and its analysis via a continuous construction process that recreates the optimal consumer surplus guarantee of Bergemann \textit{et al}.  and further subsumes their signaling scheme as a special case.  In effect, our results show settings where even though the optimal auction is randomized over a possibly large menu, there is a market segmentation such that for each segment, the optimal auction is a simple posted price scheme where the item is always sold. 
    
    The settings we consider are special cases of the more general problem where the buyer has a private budget constraint in addition to a private value. We finally show that our positive results do not extend to this more general setting, particularly when the budget can bind in the optimal auction, and when the seller's mechanism allows for all-pay auctions. Here, we show that any efficient signaling scheme necessarily transfers almost all the surplus to the seller instead of the buyer. 
\end{abstract}

\section{Introduction}
\label{section:intro}
A canonical problem in mechanism design is that of {\em bilateral trade} -- a single seller selling an item to a buyer or equivalently, an infinite supply of identical items to a stream of buyers. We assume the item has no value to the seller.  Typically, the buyers directly interact with the seller, who given distributional knowledge of the buyer's private valuation, runs an incentive compatible mechanism in order to maximize its own revenue. This mechanism is termed the {\em optimal auction}, which in this case is just a ``take it or leave it" (or monopoly) price offered to the buyer~\cite{Myerson81}. Such a mechanism could potentially lead to loss in social welfare, since the item is unsold if the value of the buyer falls below the monopoly price. 

\medskip \noindent {\bf Price Discrimination via an Intermediary.} Now imagine a platform or exchange that mediates the interaction between the buyers and the seller. This intermediary {\em observes} the private value of each arriving buyer, and it uses this information to {\em segment} the market of buyers by providing additional information (or a signal) to the seller. The seller uses this signal (or additional information) to {\em price discriminate} between different types of buyers by running separate optimal auctions for each signal. Such intermediaries are motivated by modern platforms such as ad exchanges~\cite{doubleclick,msads,verizon,pubmatic}, which help buyers (in this case, advertisers) interact with sellers (in this case, publishers of content). The ad exchange is usually run by a search engine or social media company that can use its own data to accurately learn values of advertisers for various ad slots, and selectively release this information to the publishers who then set the prices based on this information. 

Such an intermediary clearly benefits the seller's revenue; after all, the seller has more information that enables it to price discriminate. Counter-intuitively, as shown by Bergemann, Brooks, and Morris~\cite{Bergemann15}, it can also lead to more utility for the buyers, and hence larger social welfare!  In fact, the main result of~\cite{Bergemann15} is remarkable -- there is a signaling scheme such that the item always sells (so that the social welfare is as large as possible), while the seller's revenue is the {\em same} as that without signaling. Therefore, the entire extra social surplus due to signaling goes to the buyer as its utility (or its {\em consumer surplus}). This is the best possible outcome buyers can expect given that the seller controls the auction (or the pricing scheme).

Though this result is striking, the underlying setting is the simplest possible -- there is one seller and one buyer (bilateral trade), so that the optimal auction given distributional information about the buyer's valuation (either with or without signaling) is a posted price scheme that can be computed in closed form. Given a prior distribution $\D$ on the valuation of the buyer with a monopoly price $p$, the algorithms  in~\cite{Bergemann15} sequentially construct signals while maintaining the invariant that at any step, the monopoly price of the residual distribution after subtracting the signals constructed so far remains $p$. This strong invariant seems critical to the guarantee on social optimality achieved in~\cite{Bergemann15}. This makes the positive results appear specific to this setting. The question we ask in this paper is: 
\begin{quote}
    Can the positive results in~\cite{Bergemann15} be extended to significantly more general settings where the optimal auction need not be so simple?
\end{quote}

In this paper, we answer this question in the affirmative by extending the positive results in~\cite{Bergemann15} to settings where the optimal auction can be randomized, even with exponential menu complexity.\footnote{In randomized auctions where the outcome for each buyer type is a (payment, allocation) pair, one can equivalently view the collection of all such pairs as a {\em menu}, from which the buyers can choose the best one for them. This encodes incentive compatibility. The menu complexity refers to the size of this set.} 

Concretely, we study the setting where the type space of the buyer is discrete, and includes not only their private value for the item, but also a budget or deadline. Our positive results concern two settings. In the first setting, there is a publicly known upper bound on the price any buyer can be charged; this is termed the {\em public budget} setting in literature~\cite{Chawla11,Laffont96}. In the second setting, the buyer has a private {\em deadline} by which time they need to receive the item; receiving it later than the deadline yields the buyer no value. This can be equivalently viewed as a private service level for the product. The private values and deadlines are assumed to be drawn from an arbitrary two-dimensional discrete prior distribution. The auction thus needs to be incentive compatible in the sense that the buyer should not derive more utility by reporting a tighter deadline. This is termed the {\em FedEx problem} in literature~\cite{Fiat16,Saxena18}.

In both settings, the optimal auction can be randomized. In the public budget case, the randomization is over two possible menu options~\cite{Chawla11}, while for the FedEx case, the randomization can be over a menu that can be exponentially large in the number of deadlines~\cite{Fiat16,Devanur17,Saxena18}.

\subsection{Our Results}
Our main contribution is a novel signaling scheme for price discrimination in bilateral trade for the two settings of public budgets and the FedEx problem mentioned above. We show that this scheme recreates the guarantee in~\cite{Bergemann15}  -- it achieves full social welfare (that is, it always sells the item), while ensuring the seller's revenue is the same as without signaling, thereby transferring all excess social surplus to the buyer. In particular, this shows the following surprising corollary: For both these problems, even though the optimal auction is randomized in general, there is a decomposition of the prior distribution into a collection of signals such that {\em for each signal, the optimal auction is a simple posted price scheme where the item is always sold}. 

The first technical highlight of our paper is a reinterpretation of the signaling schemes for bilateral trade in~\cite{Bergemann15} as a {\em continuous time} process. In this process, an infinitesimal quantity of a signal is continuously removed from the prior distribution, and we maintain two invariants at any time instant: (a) An optimal auction for the signal being removed is efficient; and (b) the  revenue for this signal is exactly equal to the rate of decrease in revenue of the optimal auction on the current prior distribution. We use the Envelope Theorem~\cite{Milgrom} to essentially show that {\em any algorithm} that satisfies invariants (a) and (b) recreates the guarantee in ~\cite{Bergemann15}, regardless of how complex the optimal auction for the setting in consideration is. The advantage of this approach is that it enables us to sidestep both the fine-grained characterization in~\cite{Bergemann15} of how the prior changes when signals are removed from it, as well as proving their invariant that the optimal auction is preserved as signals are removed from the prior.   

The continuous framework provides a unifying method to analyze signaling schemes for both the public budget and the private deadline settings. However, we still need a careful choice of how to run the continuous process so that the two invariants hold. This is particularly challenging for the FedEx problem, since the type space here is two-dimensional, representing the values and deadlines. As we show later, a naive approach that applies the scheme in~\cite{Bergemann15} separately to the marginal induced by each deadline raises too little consumer surplus. We therefore need to develop an approach that carefully hides both the value and deadline information, and our main algorithmic contribution is the development of a novel signaling method in such spaces (Section~\ref{section:deadlines}) that achieves precisely this. This forms our second technical highlight.

Our signaling scheme and analysis require discrete (finite support) priors over valuations. Following~\cite{RubinsteinW,cai2021duality}, such priors are also an arbitrary good approximation for continuous priors via discretization. Our analysis requires a characterization of the optimal auctions in this setting, which we present in Theorems~\ref{theorem:public_optimal_revenue_is_a_distribution_over_posted_price_revenues} and~\ref{theorem:deadlines_optimal_revenue_is_a_distribution_over_posted_price_revenues}. These are the discrete analogs of results in~\cite{Chawla11,Fiat16} for continuous priors, and show that the optimal auction is a distribution over posted prices that satisfy certain nice properties. The characterizations we require are much coarser than those in~\cite{Chawla11,Fiat16} and we present stand-alone alternate proofs of these properties that are tailored to the discrete nature of the priors. In particular, the proof for the deadline setting (Theorem~\ref{theorem:deadlines_optimal_revenue_is_a_distribution_over_posted_price_revenues}) uses convexity in the primal instead of duality, and this technique may be of independent interest.



\paragraph{Impossibility for Private Budgets.} We finally ask how far we can push this positive result. Towards this end, we consider the generalization of the above settings to the {\em private budget} setting~\cite{Che00,Devanur17}. Here, the buyer has a private budget, and the values and budgets are assumed to be drawn from an arbitrary two-dimensional discrete prior distribution. The buyer cannot over-report  her budget, but an incentive compatible auction needs to prevent under-reporting it. We assume interim rationality to allow for all-pay auctions, or equivalently, views the item as infinitely divisible; this is a standard assumption in economics literature~\cite{Che00,Laffont96}. Note that the FedEx problem is a special case where the budgets are larger than all valuations.

For  private budgets, there is {\em no signaling scheme} that satisfies both criteria (a) and (b) above. This leads to a strong lower bound: Even with two values and two budgets, any efficient signaling scheme (that always sells the item) transfers all surplus to the seller, leading to vanishingly small consumer surplus. Therefore, no efficient signaling scheme can reproduce the consumer surplus guarantee in~\cite{Bergemann15} to any approximation. Furthermore, even if we sacrifice efficiency, we cannot hope to achieve better than a constant approximation to the consumer surplus guarantee.

\subsection{Related Work}
Our problem falls in the general framework of {\em information design}~\cite{bergemann2019information} where an information mediator can deliberately provide additional information to impact the behavior of agents in given mechanisms; this is also sometimes termed {\em signaling} or {\em persuasion}~\cite{dughmi2017algorithmic}. The {\em Bayesian Persuasion} model~\cite{Kamenica11} is a special case of information design with only one agent (often called the {\em receiver}) receiving additional information that comes from a {\em sender} with more knowledge of the state of nature. Given the signal, the receiver chooses the actions to maximize her own utility based on her belief of the state of nature (which may be influenced by the signal). Therefore, the sender designs the signals so that the receiver, acting in her own interest, maximizes some utility function the sender cares about. This problem is studied from various theoretical perspectives~\cite{dughmi2016persuasion, dughmi2019algorithmic,Babichenko21} as well as in different application domains~\cite{Bergemann15,chakraborty2014persuasive,xu2015exploring,ParetoIS}.

Starting with the seminal work of Bergemann {\em et al.}~\cite{Bergemann15}, there has been a line of work~\cite{dughmi2016persuasion,shen2018closed,cummings2020algorithmic,cai2020third,ParetoIS} on Bayesian persuasion in the bilateral trade model and its extensions. In this context, the sender is an intermediary and the receiver is the seller, who given the signal, implements an incentive-compatible auction to maximize expected revenue. The sender, on the other hand, is interested in maximizing consumer surplus or social welfare. In the versions we study with budgets or deadlines, the receiver's action space is the set of all randomized pricing rules, instead of just the posted prices in the basic setting~\cite{Bergemann15}. Our main contribution is to show the existence of socially efficient signaling schemes that preserve receiver utility (the revenue) and maximally increases sender utility (consumer surplus) despite this additional complexity. We note that for other non-trivial extensions of bilateral trade, for instance, the multi-buyer auction setting in~\cite{Alijani20} and the multi-item auction setting in~\cite{ParetoIS}, it may in general not be possible to find socially optimal signaling schemes that preserve seller revenue. This makes our positive results all the more surprising.

As mentioned before, our work crucially requires a characterization of optimal randomized auction in the respective settings. For public budgets, Laffont and Roberts~\cite{Laffont96} show that the optimal auction is a posted price scheme assuming regular distributions; for general priors, Chawla {\em et al.}~\cite{Chawla11} show it is a lottery over two options. Che and Gale~\cite{Che00} consider private budgets with a decreasing marginal revenue assumption, and show it is a different {\em price curve} for each budget. Fiat {\em et al.}~\cite{Fiat16} and subsequently Devanur and Weinberg~\cite{Devanur17} use duality to respectively generalize this characterization to private deadlines and private budgets with arbitrary priors; however, the characterization in the latter case is not closed form. Since we use finite support priors, we present stand-alone proofs of the required characterizations, and these may be of independent interest.

\paragraph{Organization.}
In Section \ref{section:model}, we present preliminaries for optimal auction design and signaling. In Section \ref{section:public_continuous}, we present the signaling scheme and analysis for the public budget case. 
In Section \ref{section:deadlines}, we present our main result -- the new signaling scheme for the FedEx problem, where the deadlines are private. In Section \ref{section:private}, we present the impossibility result for the private budget setting with interim rationality. All omitted proofs are in the Appendix. 


\section{Preliminaries}
\label{section:model}

\subsection{Optimal Auctions with Budgets}
\label{subsec:budgets}
We consider a single seller selling an item to a single buyer with private valuation $v$ and private budget $b$ as a hard upper bound of payment. It is known that optimal auctions with budgets require {\em randomized allocations}~\cite{Chawla11,Che00,Devanur17}: The buyer's utility is $(x \cdot v - p)$ if she pays a price of $p \leq b$ to get the item with probability $x \in [0,1]$, and is $-\infty$ if $p > b$. Throughout the paper, we focus on {\em interim IR} auctions where the buyer pays at most $b$ {\it before learning whether or not she receives the item}. 
As mentioned before, this is the standard model for studying budget constrained auctions in economics literature~\cite{Che00,Laffont96}, and allows for all-pay auctions. \footnote{We note that the algorithms in Sections~\ref{section:public_continuous} and~\ref{section:deadlines} construct signals whose optimal auctions are deterministic and satisfy {\em ex-post} IR, meaning that the price charged is always at most the budget. However, our lower bounds in Section~\ref{section:private} assume interim IR, and extending this to ex-post IR is an open question.} Alternatively, it models ex-post rationality assuming the item is infinitely divisible, and $x \in [0,1]$ represents the fraction of the item the buyer obtains at price $p$.

The joint distribution $(v,b) \sim \D$ is common knowledge and supported on a discrete set $\supp(\D) \coloneqq \{v_1, \ldots, v_n\} \times \{b_1, \ldots, b_k\}$, where $0 < v_1 < \cdots < v_n$ and $0 < b_1 < \cdots < b_k$. 
For $j = 1, 2, \ldots, k$, let $\D_j$ represent the marginal distribution of $v$ given $b = b_j$, and define $f_{\D_j}$ as the probability mass function of $\D_j$, i.e., $ f_{\D_j}(v_i) = \Pr_{v \sim \D_j}[v = v_i] = \Pr_{(v,b) \sim \D}[v = v_i \mid b = b_j].$ Let $\underline{F}_{\D_j}(v_i) = \Pr_{v \sim \D_j}[v \leq v_i]$ and $\overline{F}_{\D_j}(v_i) = \Pr_{v \sim \D_j}[v \geq v_i]$. We assume the item holds no value to the seller; therefore, the maximum social welfare is $\W^\ast(\D) = \E_{(v,b) \sim \D}[v]$, and is achieved by any auction that always makes the trade happen (or sells the item). 


\paragraph{Optimal Auctions.}  It is known~\cite{Myerson81, Laffont96, Chawla11, Che00, Devanur17} that the revenue maximizing 
auction for the seller can be described using \textit{lotteries} or randomized allocation rules. Specifically, each buyer type with valuation $v_i$ and budget $b_j$ is associated to a payment $p_{ij} \geq 0$ and an allocation probability $x_{ij} \in [0,1]$ to receive the item. Note that in the interim-IR setting the buyer pays $p_{ij}$ upfront regardless of whether she receives the item.

Following~\cite{Che00,Devanur17}, we assume buyer with type $(v_i, b_j)$ cannot report a budget larger than $b_j$.\footnote{Our positive results in Sections~\ref{section:public_continuous} and~\ref{section:deadlines} either assume the budget is publicly known, or assume deadlines instead of budgets, so this point is moot there. It is only relevant for the lower bounds in Section~\ref{section:private}.} 
We can enforce this by a cash bond that requires the full reported budget.  We note that the setting where the IC constraints are only enforced for smaller budgets is more challenging for designing optimal auctions~\cite{Che00}. By the revelation principle, it is sufficient to consider lotteries that are \textit{incentive compatible}, i.e., for all $i$ and $j$, a buyer of type $(v_i, b_j)$ receives maximum possible utility from reporting her true type $(v_i, b_j)$ and thereby receiving the item with allocation probability $x_{ij}$ at price $p_{ij}$. 

The \textit{revenue optimal} auction can be computed by the following LP from~\cite{Devanur17}. 

{\small \begin{align*}
    \textsf{Budgets}(\D) & \coloneqq \max_{\{p_{ij}\}, \{x_{ij}\}} \quad  \sum_{j=1}^{k} \left( \Pr[b = b_j] \cdot \sum_{i=1}^{n} \big( f_{\D_j}(v_i) \cdot p_{ij} \big) \right) \span\span\span\\
    \text{s.t.} \quad & v_i \cdot x_{ij} - p_{ij} \geq v_i \cdot x_{i'j} - p_{i'j}, &\quad \forall 1 \leq i,i' \leq n, &\, 1 \leq j \leq k, \tag*{(Same-budget IC)}\\
    & v_i \cdot x_{ij} - p_{ij} \geq v_i \cdot x_{i(j-1)} - p_{i(j-1)}, &\quad \forall 1 \leq i \leq n, &\, 2 \leq j \leq k, \tag*{(Inter-budget IC)}\\
    & v_i \cdot x_{ij} - p_{ij} \geq 0, &\quad \forall 1 \leq i \leq n, &\, 1 \leq j \leq k, \tag*{(IR)}\\
    & 0 \leq x_{ij} \leq 1, &\quad \forall 1 \leq i \leq n, &\, 1 \leq j \leq k, \tag*{(Feasibility)}\\
    & p_{ij} \leq b_j, &\quad \forall 1 \leq i \leq n, &\, 1 \leq j \leq k. \tag*{(Budgets)}
\end{align*}}

By transitivity, the same-budget and inter-budget IC constraints imply all necessary IC constraints so that the buyer with valuation $v_i$ and budget $b_j$ does not misreport with some valuation $v_{i'} \neq v_i$ and/or some budget $b_{j'} < b_j$.

\begin{definition} \label{def:three_key_amounts}
For the revenue maximizing auction $(\{p^{\ast}_{ij}\}, \{x^{\ast}_{ij}\})$ that is the optimal solution to \textsf{Budgets}$(\D)$, denote
\begin{align*}
    \R(\D) &= \sum_{i,j} \big( \Pr[{\D} = (v_i,b_j)] \cdot p^{\ast}_{ij} \big),\\
    \W(\D) &= \sum_{i,j} \big( \Pr[{\D} = (v_i,b_j)] \cdot v_i \cdot x^{\ast}_{ij} \big), \ \ \mbox{and}\\
    \CS(\D) &= \sum_{i,j} \big( \Pr[{\D} = (v_i,b_j)] \cdot (v_i \cdot x^{\ast}_{ij} - p^{\ast}_{ij}) \big)
\end{align*}
as the expected revenue (generated by the seller), the expected social welfare, and the expected consumer surplus (generated for the buyer), respectively. Then we have $\CS(\D) + \R(\D) = \W(\D)$.\footnote{If there are multiple optimal auctions maximizing $\R(\D)$, we break ties by defining $(\{p^{\ast}_{ij}\}, \{x^{\ast}_{ij}\})$ to be the auction that maximizes $\W(\D)$ among the optimal solutions. This auction must maximize $\CS(\D)$ as well.} 
\end{definition}

We now specify two special cases of the budgeted problem for which we derive positive results.

\paragraph{Optimal Auctions with Public Budget.}
The first special case we consider is the {\it public} budget setting~\cite{Laffont96,Chawla11} where $k = 1$, the budget $b = b_1$ is public information, and the only marginal distribution is $\D = \D_1$. This setting is motivated by the seller having an upper bound on the price they can charge any buyer, say due to regulation or other considerations.

In this case we omit the subscripts by referring to $f_{\D}(v_i)$, $\underline{F}_{\D}(v_i)$, and $\overline{F}_{\D}(v_i)$, and use $p_i$ and $x_i$ as shorthand for the payment variables $p_{i1}$ and allocation variables $x_{i1}$, respectively. For this case, the optimal auction is captured by the following special case of \textsf{Budgets} with $k = 1$:
\begin{align*}
    \textsf{Public}(\D) \coloneqq \max_{\{p_i\}, \{x_i\}} \quad & \sum_{i=1}^{n} \big( f_{\D}(v_i) \cdot p_{i} \big) &\\
    \text{s.t.} \quad & v_i \cdot x_i - p_i \geq v_i \cdot x_{i'} - p_{i'}, &\quad \forall 1 \leq i,i' \leq n, \tag*{(IC)}\\
    & v_i \cdot x_{i} - p_{i} \geq 0, &\quad \forall 1 \leq i \leq n, \tag*{(IR)}\\
    & 0 \leq x_{i} \leq 1, &\quad \forall 1 \leq i \leq n, \tag*{(Feasibility)}\\
    & p_{i} \leq b, &\quad \forall 1 \leq i \leq n. \tag*{(Budget)}
\end{align*}

We devise price discrimination schemes for the public budget setting in Section~\ref{section:public_continuous}.

\paragraph{Optimal Auctions with Deadlines.}
In this setting~\cite{Fiat16,Saxena18}, we consider a single seller selling an identical item with different levels of service quality to a single buyer. The buyer now has private valuation $v$ (conditioned on getting the item with at least her desired level of quality) and a private desired level of quality  $d$. One can  think of $d$ as either a personal deadline for shipping options, or as a level of service quality for a product. Keeping with previous work, we will refer to $d$ as deadlines throughout.

The buyer's utility is $(x \cdot v - p)$ if she pays a price of $p$ to get the item with a probability of $x$ at some point before or right at her deadline. She incurs utility $-p$ if she gets the item later than her deadline, since in this case, she accrues no value from the item. As observed in~\cite{Fiat16}, it is sufficient to consider auctions that, for each buyer with deadline $d$, only allocates the item right at the $d$-th deadline (if at all). This is because a buyer does not get any additional utility if she receives the item at some point earlier than her own deadline. Furthermore, the buyer weakly prefers getting nothing over getting the item after her own deadline for some price.

The LP for this setting is a special case of the LP for the private budget setting with large budgets, that is, when $v_n < b_j$ holds for all budget types $j = 1, \ldots, k$. As every budget is above the highest possible valuation, by the IR constraint, the optimal auction never sets a price above $b_j$ for any buyer with budget $b_j$, and thus the budget constraint $p_{ij} \leq b_j$ in \textsf{Budgets} can be omitted. 

For this case, we simplify the notations by denoting the joint distribution $(v,d) \sim \D$ supported on $\supp(\D) \coloneqq \{ v_1, \ldots, v_n\} \times \{1, \ldots, k\}$, where $0 < v_1 < \cdots < v_n$. The deadlines can be represented as $\{1, \ldots, k\}$ since their cardinal values do not matter. For $j = 1, 2, \ldots, k$, $\D_j$ now represents the marginal distribution of $v$ given $d = j$, and the corresponding probability mass function of $\D_j$ is
\[ f_{\D_j}(v_i) = \Pr\limits_{v \sim \D_j}[v = v_i] = \Pr\limits_{(v,d) \sim \D}[v = v_i \mid d = j]. \]
Let $\underline{F}_{\D_j}(v_i) = \Pr_{v \sim \D_j}[v \leq v_i]$ and $\overline{F}_{\D_j}(v_i) = \Pr_{v \sim \D_j}[v \geq v_i]$. We again assume the item holds no value to the seller; therefore, the maximum social welfare is ${\W}^\ast(\D) = \E_{(v,d) \sim \D}[v]$, and is achieved by any auction that always allocates the item to each buyer right at her personal deadline.

The \textit{revenue maximizing} randomized incentive compatible auction for the deadlines setting is thus the following:
\begin{align*}
    \textsf{Deadlines}(\D) \coloneqq \max_{\{p_{ij}\}, \{x_{ij}\}} \quad  \sum_{j=1}^{k} \left( \Pr_{(v,d)\sim \D}[d = j] \cdot \sum_{i=1}^{n}  \big( f_{\D_j}(v_i) \cdot p_{ij} \big) \right) \span\span \\
    \text{s.t.} \quad & v_i \cdot x_{ij} - p_{ij} \geq v_i \cdot x_{i'j} - p_{i'j}, &\quad \forall 1 \leq i,i' \leq n, \, 1 \leq j \leq k, \tag*{(Same-deadline IC)}\\
    \quad & v_i \cdot x_{ij} - p_{ij} \geq v_i \cdot x_{i(j-1)} - p_{i(j-1)}, &\quad \forall 1 \leq i \leq n, \, 2 \leq j \leq k, \tag*{(Inter-deadline IC)}\\
    & v_i \cdot x_{ij} - p_{ij} \geq 0, &\quad \forall 1 \leq i \leq n, \, 1 \leq j \leq k, \tag*{(IR)}\\
    & 0 \leq x_{ij} \leq 1, &\quad \forall 1 \leq i \leq n, \, 1 \leq j \leq k. \tag*{(Feasibility)}
\end{align*}

Note that \textsf{Deadlines} is a special case of \textsf{Budgets} where the budget constraint is omitted, and the IC constraints in \textsf{Deadlines} prevent misreporting a lower deadline. 

\paragraph{Remarks.} We first note that though our scheme require discrete priors over valuations, these also serve as arbitrarily good approximations to continuous priors via simple discretization~\cite{cai2021duality,RubinsteinW}. Secondly, note that the optimal auctions with interim IR coincides with that for ex-post IR for both the public budget and the deadline setting; the former follows from Theorem~\ref{theorem:public_optimal_revenue_is_a_distribution_over_posted_price_revenues} (or from~\cite{Chawla11}), while the latter follows because the prices are not really constrained by any budget. Therefore, our positive results in Sections~\ref{section:public_continuous} and~\ref{section:deadlines} extend as is to ex-post IR. Our negative results in Section~\ref{section:private} do require interim IR.

\subsection{Price Discrimination} \label{subsec:intermediary_model}
We next introduce price discrimination via signaling by an information intermediary for the general private budget setting; specializing it to deadlines or public budgets is straightforward. The intermediary knows the type $(v_i, b_j)$ of the buyer, and can propose a \textit{signaling scheme} that maps the buyer's private information (i.e., a \textit{value-budget} pair $(v_i,b_j)$) to a distribution over \textit{signals} that conveys additional information to the seller. This makes the seller \textit{update} her belief of the buyer's information via Bayes' rule. The signaling scheme thus can be seen as \textit{segmenting} the market of buyers, each segment representing the conditional distribution of buyer type given the signal. Therefore, we can overload terminology and simply define a signal $\S$ as the posterior distribution of $(v,b)$ given the signal.

\paragraph{Signaling Scheme.} Formally, a \textit{signaling scheme} $\Z = \{ (\weight_h, \S_h) \}_{h \in [H]}$ is a collection of signals $\S_1, \ldots, \S_H$ and probability weights $\weight_1, \ldots, \weight_H > 0$, where $\sum_{h=1}^{H} \weight_h = 1$. Here $\S_h$ represents the posterior distribution of the type given the $h$-th signal. We also require $\Z$ being \textit{Bayes plausible}~\cite{Kamenica11}, 
\begin{equation}
    \sum_{h=1}^{H} \weight_h \S_h = \D, \label{eq:bayes_plausibility}
\end{equation}
i.e., the \textit{average} signal is just the prior $\D$. The intermediary commits to this signaling scheme before she observes the buyer type, and this scheme is public knowledge to all parties. 

We note that in general, collection of signals above could be a continuous set, and the weights could represent a density over this set. This aspect will not affect our algorithms, since the signaling schemes we construct in Sections~\ref{section:public_continuous} and~\ref{section:deadlines} will have finitely many signals. On the other hand, the lower bound results in Section~\ref{section:private} hold even if the set of signals is uncountable.

Upon observing private information $(v,b)$, the intermediary sends the $h$-th signal with probability $\frac{\weight_h \cdot \Pr[\S_h = (v,b)]}{\Pr[{\D} = (v,b)]}$, and given this signal, if the seller uses Bayes rule to update the prior on the buyer's type, the posterior will be precisely $\S_h$. The seller then implements the revenue maximizing auction based on the updated prior $\S_h$. 

\paragraph{Buyer Optimal Schemes.} Abusing the notation defined before, we let $\R(\Z) = \sum_{h=1}^{H} \big( \weight_h \cdot \R(\S_h) \big)$, $\W(\Z) = \sum_{h=1}^{H} \big( \weight_h \cdot \W(\S_h) \big)$, and $\CS(\Z) = \sum_{h=1}^{H} \big( \weight_h \cdot \CS(\S_h) \big)$ denote the expected revenue, the expected social welfare, and the expected consumer surplus, respectively, achieved by the signaling scheme $\Z$, where the expectation is now taken over all signals. As before, we have $\CS(\Z) + \R(\Z) = \W(\Z)$.

Furthermore, $\R(\Z) \geq \R(\D)$; otherwise, the seller can ignore the signaling scheme $\Z$ and implement the revenue maximizing auction based on $\D$ instead. Hence, for any possible signaling scheme $\Z$, we have 
$$\CS(\Z) = \W(\Z) - \R(\Z) \leq \W^\ast(\D) - \R(\D)$$ 
as an upper bound of the expected consumer surplus. Recall that ${\W}^\ast(\D)$ is the maximum possible social welfare assuming the item always sells. We define this bound on maximum achievable consumer surplus as
\[ \opt(\D) \coloneqq \W^\ast(\D) - \R(\D) = \E\limits_{(v,b) \sim \D}[v] - \R(\D). \]

To achieve $\CS(\Z) = \opt({\D})$, the signaling scheme $\Z$ thereby needs to satisfy (a) the item always sells, and (b) the revenue $\R(\Z)$ generated by $\Z$ is exactly $\R(\D)$, i.e., the expected revenue without signaling. We call a signaling scheme \emph{buyer optimal} if it achieves this upper bound.

\subsection{Buyer Optimal Signaling without Budgets or Deadlines}
We illustrate the concept of signaling via the following example.

\begin{example}
Consider a simple two-point distribution where $v = 1$ and $3$ each with probability $\frac{1}{2}$. Then the maximum social welfare is $\W^\ast(\D) = \frac{1}{2} \cdot 1 + \frac{1}{2} \cdot 3 = 2$. To map to the notation above, we assume the budget is $b=3$, so that the type is $(v,b) = (1,3)$ and $(3,3)$ with probability $1/2$ each; this budget is irrelevant to the computations below. The revenue maximizing auction is a single posted price $p = 3$, which raises a revenue $\R(\D) = 3 \cdot \frac{1}{2} = \frac{3}{2}$, whereas the consumer gets no surplus (i.e., $\CS(\D) = 0$). 

\sloppy Consider a signaling scheme $\Z$ that segments $\D$ into two signals $\S_1$ and $\S_2$ as follows: $\Pr_{(v,b) \sim \S_1}[(v,b) = (1,3)] = \frac{2}{3}$, $\Pr_{(v,b) \sim \S_1}[(v,b) = (3,3)] = \frac{1}{3}$, and $\Pr_{(v,b) \sim \S_2}[(v,b) = (3,3)] = 1$. They are given weights $\weight_1 = \frac{3}{4}$ and $\weight_2 = \frac{1}{4}$.  If the intermediary observes the buyer with type $(v,b) = (1,3)$, $S_1$ is released with probability $\weight_1 \cdot \frac{\Pr[\S_1 = (1,3)]}{\Pr[{\D} = (1,3)]} = \frac{3}{4} \cdot \frac{2/3}{1/2} = 1$; on the other hand, if the buyer has type $(v,b) = (3,3)$, the intermediary releases $S_1$ or $S_2$ with equal probability, since $\weight_1 \cdot \frac{\Pr[\S_1 = (3,3)]}{\Pr[{\D} = (3,3)]} = \frac{3}{4} \cdot \frac{1/3}{1/2} = \frac{1}{2} = \frac{1}{4} \cdot \frac{1}{1/2} = \weight_2 \cdot \frac{\Pr[\S_2 = (3,3)]}{\Pr[{\D} = (3,3)]} $.

Again, since the budget is not relevant, the revenue maximizing auctions in $\S_1$ and $\S_2$ are posted prices $p_1 = 1$ and $p_2 = 3$, respectively. The expected revenue raised by the seller in this signaling scheme is $\R(\Z) = \weight_1 \cdot \R(\S_1) + \weight_2 \cdot \R(\S_2) = \frac{3}{4} \cdot 1 + \frac{1}{4} \cdot 3 = \frac{3}{2} = \R(\D)$. The consumer surplus for the buyer is now $\CS(\Z) = \weight_1 \cdot \CS(\S_1) + \weight_2 \cdot \CS(\S_2) = \frac{3}{4} \cdot \frac{1}{3} \cdot (3-1) + \frac{1}{4} \cdot 0 = \frac{1}{2}$. Finally, since the item always sells, the expected social welfare is exactly at $\W^\ast(\D) = 2$.
\end{example}

In this example, the optimal consumer surplus $\opt(\D)$ is indeed achieved by the signaling scheme $\Z$. This is not a coincidence, but an example of the ``optimal signaling schemes'' given by Bergemann \etal~\cite{Bergemann15}. We restate their main result for the case of no budgets or deadlines:

\begin{restatable}[Bergemann \etal 's signaling schemes~\cite{Bergemann15}]{theorem}{bbmthm}
\label{theorem:bbm}
Suppose $k=1$ and $b_1 \geq v_n$. Then for any arbitrary prior $\D$, there exists a signaling scheme $\Z^{\bbm}_{\D}$ that guarantees:
\begin{enumerate}
    \item {\em efficiency}: $\W(\Z^{\bbm}_{\D}) = \W^\ast(\D)$ (i.e., the item always sells);
    \item {\em minimum revenue}: $\R(\Z^{\bbm}_{\D}) = \R(\D)$ (i.e., the seller's revenue does not increase);
    \item {\em maximum consumer surplus}: $\CS(\Z^{\bbm}_{\D}) = \opt(\D) = \W^\ast(\D) - \R(\D)$ (i.e., the scheme maximizes the expected consumer surplus among all possible signaling schemes.)
\end{enumerate}
\end{restatable}

Note that the third property is implied from the first two. There are multiple constructions of $\Z^{\bbm}_{\D}$ given in~\cite{Bergemann15}, and one of these is equivalent to our scheme for public budgets presented in Section \ref{section:public_continuous}. These schemes proceed via the notion of \emph{Equal Revenue Signals}. We now introduce this notion since it is essential to our signaling scheme as well.

\begin{restatable}[Equal Revenue Signals]{definition}{equirevenue} \label{definition:equi_revenue_distribution}
A valuation distribution $\S$ over its support set $\supp(\S) = \{v_1, \ldots, v_n\}$ is {\em equal revenue} if it satisfies:
\[ \overline{F}_{\S}(v_1) \cdot v_1 = \overline{F}_{\S}(v_2) \cdot v_2 = \cdots = \overline{F}_{\S}(v_n) \cdot v_n = \R(\S). \]
\end{restatable}
In other words, {\em assuming no budgets}, every valuation with nonzero probability mass in $\supp(\S)$ is an optimal monopoly price for $\S$. This distribution is unique and can be obtained as follows:
\begin{align*}
    f_{\S}(v_1) = 1 - \frac{v_1}{v_2}; \quad
    f_{\S}(v_i) = \big( 1 - \underline{F}_{\S}(v_{i-1}) \big) \cdot (1 - \frac{v_i}{v_{i+1}}), \, \forall 2 \leq i \leq n-1; \quad
    f_{\S}(v_n) = 1 - \underline{F}_{\S}(v_{n-1}).
\end{align*}

\section{Warmup: Signaling Scheme for Public Budgets}
\label{section:public_continuous}
In this section, we prove the analog of Theorem~\ref{theorem:bbm} when there is a public budget. We show that there is a signaling scheme that is buyer optimal with a public budget. We will show this via reinterpreting the algorithm in~\cite{Bergemann15} as a continuous time process (Algorithm~\ref{alg:main_algorithm_public} below). The nice aspect of this interpretation is that it leads to a different proof of optimality (than~\cite{Bergemann15}) via invoking the Envelope Theorem~\cite{Milgrom} on the revenue of the residual prior as a function of time. This continuous time interpretation will form the building blocks for our main result for the version with deadlines (the FedEx problem) in Section~\ref{section:deadlines}. 

Interestingly, our signaling scheme for public budgets is {\em the same} as the no-budget signaling scheme in~\cite{Bergemann15}; this is easy to check and we omit the proof. However, our analysis is entirely different and more generalizable to the more complex deadline setting considered later.

\subsection{Signaling Algorithm}
\label{subsec:public_algorithm}
Throughout this section, the buyer's budget $b = b_1$ is public information, and we use $\D$ to denote the prior over the buyer values $\{v_1, v_2, \ldots, v_n\}$. We view the progress of the algorithm as continuously decreasing this prior into a residual prior, and continuously placing the remaining probability mass into the constructed signals.

We use the function $\Vec{\mathbf{f}}(t) = \langle f_1(t), \ldots, f_n(t) \rangle$ to represent the {\it residual prior}, where $f_i(t)$ represents the remaining probability mass on type $v_i$ in the residual prior distribution at time $t$. Strictly speaking, $\Vec{\mathbf{f}}(t)$ is not a distribution since the process we describe only guarantees $\sum_{i=1}^{n} f_i(t) < 1$ for $t > 0$. To make this a valid distribution, we place the remaining probability mass $1 - \sum_{i=1}^{n} f_i(t)$ at a dummy value $v_0 = 0$. We call the resulting distribution $\D(t)$. In the subsequent discussion, the notation $\Vec{\mathbf{f}}(t)$ represents the probability mass of $\D(t)$ at non-zero valuations, and we omit explicitly considering the dummy value $v_0 = 0$ as part of the support of $\D(t)$. We define $\supp(\D(t)) \coloneqq \{ v_i > 0 \mid f_i(t) > 0 \}$ and $v_{\min}(\D(t)) \coloneqq \min \{ v_i > 0 \mid f_i(t) > 0 \}$.

We start with the prior $\D(0) = \D$ and let $f_i(0) = f_{\D}(v_i)$, i.e., $\Vec{\mathbf{f}}(0)$ is just the probability vector associated with $\D$. Our algorithm continuously takes away probability mass from $\Vec{\mathbf{f}}(t)$ and transfers it to the constructed signals,  terminating when $\Vec{\mathbf{f}}(t)$ becomes $\mathbf{0}$; denote the latter time as $T$. 

At any time $t$ such that $\Vec{\mathbf{f}}(t) \neq \mathbf{0}$, we denote  $\Vec{\mathbf{s}}(t)$ as the probability distribution associated with the equal revenue distribution (see Definition~\ref{definition:equi_revenue_distribution}) $\S(t)$ over the set of values in $\supp(\D(t))$. Note that $\Vec{\mathbf{s}}(t)$ depends on $\supp(\D(t))$ but not the $f_i(t)$; therefore, it is fixed as long as $\supp(\D(t))$ does not change.\footnote{Note that there exist other equal revenue distributions over different support sets; for example, any distribution with a support size of one is equal revenue. However, for the purpose of our algorithm, the equal revenue distribution must use all remaining nonzero valuations with nonzero probability mass in the residual prior.} Our algorithm continuously reduces $\Vec{\mathbf{f}}(t)$ at rate  $\Vec{\mathbf{s}}(t)$ until $\Vec{\mathbf{f}}$ becomes $\mathbf{0}$. Formally:

\begin{align}
    \dv{\Vec{\mathbf{f}}}{t}(t) = -\Vec{\mathbf{s}}(t). \label{eq:public_differential_equation}
\end{align}

Since $\sum_{i=1}^{n} s_i(t) = 1$, the rate of decrease of $\sum_{i=1}^{n} f_i(t)$ is also $1$. Since $\sum_{i=1}^{n} f_i(0) = 1$, this means the process terminates at time $T =1$.

\paragraph{Signals constructed.} 
We say the type-$i$ valuation $v_i$ is {\it exhausted at time $t$} if $f_i(t) = 0$ but $f_i(t') > 0$ for all $t' < t$. The algorithm therefore terminates once all types are exhausted. Consider a maximal time interval $t \in [t_1, t_2)$ in which $\supp(\D(t))$ remains fixed; denote the equal revenue signal in this interval by $\Vec{\mathbf{s}}$. Therefore, $\S(t) = \Vec{\mathbf{s}}$ for $t \in [t_1, t_2)$. Then we have:
\begin{equation}
    \Vec{\mathbf{f}}(t_1) - \Vec{\mathbf{f}}(t_2) = -\int_{t = t_1}^{t_2} -\Vec{\mathbf{s}}(t) \dd{t} = (t_2 - t_1) \cdot \Vec{\mathbf{s}} \label{eq:public_signal_construction}
\end{equation}
    
Therefore, the final scheme includes a signal $\Vec{\mathbf{s}}$ with weight $(t_2 - t_1)$. This holds for every such interval $[t_1, t_2)$. Since $\Vec{\mathbf{s}}(t)$ changes only if some element in $\Vec{\mathbf{f}}(t)$ becomes zero, the number of signals constructed is at most $n$. The signaling scheme is now formally described in Algorithm~\ref{alg:main_algorithm_public}.

We have $\sum_{h=1}^{H} \weight^{\ast}_h = \sum_{h=1}^{H} (t_{h} - t_{h-1}) = t_H - t_0 = 1$. Further, we have Bayes plausibility (as defined in Eq.~(\ref{eq:bayes_plausibility}) in Section~\ref{section:model}): 

\begin{restatable}[]{observation}{bayesplausibility}
\label{observation:public_plausibility_of_signals}
    For any arbitrary $\D$, let $\Z^{\ast}_{\D} =  \{ (\weight^{\ast}_h, \S^{\ast}_h) \}_{h \in [H]}$ be the set of signals output by Algorithm~\ref{alg:main_algorithm_public} taking $\D$ as input. Then we have $H \leq |\supp(\D)|$, and $\sum_{h=1}^{H} \weight^{\ast}_h \S^{\ast}_h = \D$.
\end{restatable}
\begin{proof}
By summing up Eq. (\ref{eq:public_signal_construction}) for all pairs $[t_{h-1}, t_{h})$ for $h \in \{ 1, \ldots, H \}$ we have
    \begin{align*}
        \Vec{\mathbf{f}}(0) = \Vec{\mathbf{f}}(0) - \Vec{\mathbf{f}}(1) = \Vec{\mathbf{f}}(t_0) - \Vec{\mathbf{f}}(t_H) = \sum\limits_{h=1}^{H} \big( \Vec{\mathbf{f}}(t_{h-1}) - \Vec{\mathbf{f}}(t_{h}) \big) &= \sum\limits_{h=1}^{H} \big( (t_{h} - t_{h-1}) \cdot \Vec{\mathbf{s}}(t_{h-1}) \big)\\
            &= \sum\limits_{h=1}^{H} \big( \weight^{\ast}_h \cdot \Vec{\mathbf{s}}(t_{h-1}) \big).
    \end{align*}
Therefore the claim follows by observing $\Vec{\mathbf{f}}(0)$ is the probability vector associated with $\D$ and each $\Vec{\mathbf{s}}(t_{h-1})$ is the probability vector associated with $\S^{\ast}_h$.
\end{proof}

\begin{algorithm}[htbp]  
    \caption{Continuous Algorithm for Public Budget Setting}
    \label{alg:main_algorithm_public}
    \begin{algorithmic}[1]
        \Require $\D$
        \Ensure $\Z = \Z^{\ast}_{\D}$
        \State $t_0 \gets 0$; \ $\D(t_0) \gets \D$; \ $\Vec{\mathbf{f}}(t_0) \gets \langle f_{\D}(v_1), \ldots, f_{\D}(v_n) \rangle$; 
        \For {$h \in \{1, \ldots, H\}$}
            \State $t \gets t_{h-1}$
            \State $\S(t_{h-1}) \gets$ Equal revenue distribution on $\supp(\D(t_{h-1}))$ 
            \State Run Equation~(\ref{eq:public_differential_equation}) using $\Vec{\mathbf{s}}(t)$ as density of $\S(t_{h-1})$ till some type's support in $\Vec{\mathbf{f}}(t)$ is exhausted at time $t = t_h$
            \State $\D(t_h) \gets$ distribution induced by $\Vec{\mathbf{f}}(t_h)$
            \State $\weight^{\ast}_h \gets t_{h} - t_{h-1}$; \  $\S^{\ast}_h \gets \S(t_{h-1})$; \ $\Z \gets \Z \cup \{(\weight^{\ast}_h, \S^{\ast}_h)\}$
        \EndFor
        \Return $\Z$
    \end{algorithmic}
\end{algorithm}

\subsection{Optimal Auction For Signals}
\label{subsec:public_welfare}
We start with the easy step. We characterize the revenue-optimal auctions in the signals created by Algorithm~\ref{alg:main_algorithm_public} in  Lemma~\ref{lemma:public_optimal_auction_for_signals} (proved in Appendix~\ref{app:public}). As an easy consequence, $\Z^\ast_{\D}$ always sells the item, and therefore guarantees efficiency. This is the first necessary condition for buyer optimality. 

Let $\S_h \in \Z^{\ast}_{\D}$ denote a signal created by Algorithm~\ref{alg:main_algorithm_public} for the prior $\D$, and let $v_{\min}(\S_h)$ denote the minimum $v_i > 0$ such that $\Pr_{v \sim \S_h}[v = v_i] > 0$.

\begin{restatable}[]{lemma}{publicoptimalauction}
\label{lemma:public_optimal_auction_for_signals}
The optimal auction for $\S_h$ has the following structure:
\begin{itemize}
\item If $b \leq v_{\min}(\S_h)$, there is an optimal auction that posts a price of $b$, and raises a revenue of $b$.
\item If $b > v_{\min}(\S_h)$, there is an optimal auction that posts price $v_{\min}(\S_h)$ and raises a revenue $v_{\min}(\S_h)$. Further, for every $v \in \supp(\S_h)$, we have $v \cdot \overline{F}_{\S_h}(v) = v_{\min}(\S_h) $. 
\end{itemize}
\end{restatable}
\begin{proof}
First consider the case when $b \leq v_{\min}(\S_h)$. In the program $\textsf{Public}(\S_h)$ (defined in Section~\ref{subsec:budgets}), we have the budget constraints $p_i \leq b$ for all $i$. This means the optimal revenue is $\sum_{i=1}^{n} \big( p_i f_{\S_h}(v_i) \big) \le b \cdot \sum_{i=1}^{n} f_{\S_h}(v_i) = b$. However, posting a price of $b$ is also a feasible auction, and raises exactly revenue $b$ since $b \leq v_{\min}(\S_h)$. Therefore, it is an optimal auction for $\textsf{Public}(\S_h)$.

Next consider the case when $b > v_{\min}(\S_h)$. Suppose we remove the constraints $p_i \le b$ from $\textsf{Public}(\S_h)$. This cannot decrease the optimal revenue. The optimal auction without the budget constraint posts the monopoly price~\cite{Myerson81}. Since $\S_h$ is equal revenue, this price is $v_{\min}(\S_h)$. However, this price is also feasible for the budget constraint $b$ since $b > v_{\min}(\S_h)$, which means it must be the optimal auction even with the budget constraint. The second part of the claim directly follows from the fact that $\S_h$ is equal revenue.
\end{proof}

By the characterizations of the optimal auctions in the above two cases, we have the following claim that states the item always sells in $\Z^{\ast}_{\D}$.
\begin{lemma}[Efficiency of $\Z^{\ast}_{\D}$] \label{lemma:public_efficiency_of_scheme}
For each signal $\S_h \in \Z^{\ast}_{\D}$, there exists a revenue optimal auction that always allocates the item. As a consequence, $\W(\Z^\ast_{\D}) = \W^\ast(\D)$.
\end{lemma}

\subsection{Characterization of Optimal Auction for $\D(t)$} 
\label{subsec:public_optimal_auction}
We analyze the revenue of the signals by showing that the rate of decrease in revenue of the optimal auction for $\D(t)$ is equal to the revenue of the signal $\S(t)$. (See Theorem~\ref{theorem:public_revenue_preservation}.) This when integrated over time shows that the optimal revenue of the signals is exactly equal to the optimal revenue for prior $\D$, hence showing buyer optimality.


\paragraph{Continuous Constraints.} For the purpose of analysis, we make the constraints in $\textsf{Public}(\D(t))$ hold not just for $\supp(\D(t))$, but for all {\em continuous} values $v \ge 0$, where the prior possibly has zero probability mass.\footnote{It follows from~\cite{RubinsteinW,cai2021duality} that this formulation is equivalent to $\textsf{Public}(\D(t))$, though we will not need this equivalence.}  Among other things, this formulation allows us to argue that the revenue changes continuously as the prior changes while constructing our signals.

Formally, fix some time $t$, and let $\A = \D(t)$ so that $\supp(\A) \subseteq \supp(\D)$. Recall that the decision variables in $\textsf{Public}(\A)$ are $p_i$ (the payment) and $x_i$ (the allocation probability) for all buyer types with valuation $v_i$.  We augment the variables by extending the domain to $[0,v_n]$; for all $v \in [0, v_n]$, we let $p(v)$ and $x(v)$ denote the expected payment and allocation probability at $v \in [0,v_n]$. This yields the following LP, where the IC and IR constraints are extended to this domain.

\begin{align*}
    \textsf{PublicContinuous}(\A) & \coloneqq \max_{p(\cdot), x(\cdot)} \quad \sum_{i=1}^{n} \big( f_{\A}(v_i) \cdot p(v_i) \big) \span\span\\
    \text{s.t.} \quad & v \cdot x(v) - p(v) \geq v \cdot x(v') - p(v'), &\quad \forall v, v' \in [0, v_n], \tag*{(Cont. IC)}\\
    & v \cdot x(v) - p(v) \geq 0, &\quad \forall v \in [0, v_n], \tag*{(Cont. IR)}\\
    & 0 \leq x(v) \leq 1, &\quad \forall v \in [0, v_n], \tag*{(Feasibility)}\\
    & p(v) \leq b, &\quad \forall v \in [0, v_n]. \tag*{(Budget)}
\end{align*}

From Definition~\ref{def:three_key_amounts}, $\R(\A)$ is the optimal revenue achievable by $\textsf{Public}(\A)$. Denote by $\Tilde{\R}(\A)$ the optimal revenue achievable by $\textsf{PublicContinuous}(\A)$.  Clearly, $\R(\A) \geq \Tilde{\R}(\A)$. 

We now present the main characterization result for the optimal solution to this LP. This can be viewed as a discrete analog of the characterization for continuous priors in~\cite{Chawla11}. 
We present a stand-alone proof for our discrete setting in Appendix~\ref{apdx:proof_of_public_optimal_revenue_for_signals} via convexity of the utility curve.

\begin{restatable}[]{theorem}{publicrevenuedistribution}
\label{theorem:public_optimal_revenue_is_a_distribution_over_posted_price_revenues}
For any prior $\A = \D(t)$ with $b > v_{\min}(\A)$, there exists a set of valuations $\{w'_1, w'_2, \ldots, w'_{m'} \} \subseteq \supp(\A)$ and weights $\delta_1, \delta_2, \ldots, \delta_{m'} \in (0,1]$ such that $\sum_{j=1}^{m'} \delta_j = 1$, and the optimal revenue of  $\textsf{PublicContinuous}(\A)$ is $\Tilde{\R}(\A) = \sum_{j=1}^{m'} \big( \delta_j \cdot w'_j \cdot \overline{F}_{\A}(w'_j) \big)$.
\end{restatable}


\subsection{Revenue Preservation in Algorithm~\ref{alg:main_algorithm_public}}
\label{subsec:public_revenue_preservation}
We are now ready to prove the second necessary criterion for buyer optimality: $\R(\Z^\ast_\D) = \R(\D)$, i.e., Algorithm~\ref{alg:main_algorithm_public} minimizes the expected seller revenue through signaling. 
As mentioned above, the key step (Theorem~\ref{theorem:public_revenue_preservation}) is to argue that {\em the  rate of decrease of revenue of $\textsf{PublicContinuous}(\D(t))$ exactly equals the optimal revenue of the signal $\S(t)$.} 
This when combined with Lemma~\ref{lemma:public_efficiency_of_scheme} gives buyer optimality. As a side effect, this will also show the optimal objectives of $\textsf{Public}(\D)$ and $\textsf{PublicContinuous}(\D)$ are identical.

\paragraph{Convexity of Revenue in $\textsf{PublicContinuous}(\D(t))$.} The next lemmas bound the \textit{continuous flow} of revenue being transferred into the signals. 
At any time $t$, let $\Tilde{\R}(t)$ denote the optimal revenue of $\textsf{PublicContinuous}(\D(t))$. 

Our first step is fairly generic, and shows that $\Tilde{\R}(t)$ is convex\footnote{This is slightly misleading: As we show later, the function $\Tilde{\R}(t)$ is actually linear as long as $\supp(\D(t))$ does not change. The overall function over $t \in [0,1]$ turns out to be piece-wise linear and {\em concave}.}   in any time interval $(t_{h-1},t_{h})$ where $\supp(\D(t))$ (and hence $\S(t)$) does not change for $t \in (t_{h-1},t_{h})$. Let this signal be $\S$ and its corresponding probability vector to be $\Vec{\mathbf{s}} = \{s_i\}$ over $t \in (t_{h-1},t_{h})$. 

Consider the decrease in revenue of any feasible solution $\M \coloneqq (p(\cdot), x(\cdot))$ of \textsf{PublicContinuous}$(\D(t))$. 
The revenue of the solution $\M$ in \textsf{PublicContinuous}$(\D(t))$ is given by $\Tilde{\R}_{\M}(t) \coloneqq \sum\limits_{i=1}^{n} \big(f_i(t) \cdot {p}(v_i)\big).$ Fixing this $\M$, for each $t \in (t_{h-1},t_{h})$, we have:
\begin{align}
\dv{\Tilde{\R}_{\M}(t)}{t} =  \dv{t} \left( \sum_{i=1}^{n} \big(f_i(t) \cdot {p}(v_i) \big) \right) = - \sum_{v_i \in \supp(\D(t))} s_i \cdot {p}(v_i),
\label{eq:public_revenue_decrease_in_residual_prior}
\end{align} 
where the last equality uses Eq. (\ref{eq:public_differential_equation}) and the fact that $s_i = 0$ for any $v_i \notin \supp(\D(t))$. Therefore, for each $\M$, its revenue $\Tilde{\R}_{\M}(t)$ decreases linearly with time. Further, since the constraints of \textsf{PublicContinuous}$(\D(t))$ do not change with time, $\M$ remains feasible at all points in time $t \in [0,1]$ (the duration of the algorithm), and its revenue $\Tilde{\R}_{\M}(t)$ is continuous at all $t \in [0,1]$ since each $f_i(t)$ changes continuously.  Since $\Tilde{\R}(t) = \max_{\M} \Tilde{\R}_{\M}(t)$, we have:

\begin{lemma} \label{lemma:public_revenue_is_convex_for_every_mechanism}
In Algorithm~\ref{alg:main_algorithm_public}, for any interval $(t_{h-1},t_{h})$ where $\supp(\D(t))$ does not change, the function $\Tilde{\R}(t)$ is convex. Further, the function $\Tilde{\R}(t)$ is continuous for all $t \in [0,1]$, that is, the entire duration of the algorithm.
\end{lemma}

\paragraph{Revenue of Signals.}  We are now ready to prove our main theorem quantifying the rate of decrease of $\Tilde{\R}(t)$, and thus bounding the revenue of the signals. 

\begin{theorem} \label{theorem:public_revenue_preservation}
In any interval $(t_{h-1},t_{h})$ where $\supp(\D(t))$ does not change, the function $\Tilde{\R}(t)$ is linear, and $\dv{\Tilde{\R}(t)}{t} = -\R(\S(t))$, where $\R(\S(t))$ is the optimal revenue of \textsf{Public}$(\S(t))$. Furthermore, at the end of Algorithm~\ref{alg:main_algorithm_public}, it holds that $\R(\Z^{\ast}_{\D}) = \R(\D)$.
\end{theorem}
\begin{proof}
Fix some $t \in (t_{h-1},t_{h})$, and assume $\supp(\D(t))$ does not change in $(t_{h-1},t_{h})$ and thus $\S(t) = \S$. Note that $\supp(\D(t)) = \supp(\S)$. 
First consider the case where $v_{\min}(\S) \geq b$. By Lemma~\ref{lemma:public_optimal_auction_for_signals}, every revenue maximizing auction for $\S$ must have $p_i = b$ for all $v_i \in \supp(\S)$, and thus $\R(\S) = b$; the same proof implies every revenue maximizing auction $\M'$ for \textsf{PublicContinuous}$(\D(t))$ must have $p'_i = b$ for all $v_i \in \supp(\D(t))$. Therefore, we have 
\begin{align*}
    \dv{\Tilde{\R}(t)}{t} = \sum\limits_{v_i \in \supp(\D(t))} \big( \dv{f_i(t)}{t} \cdot b \big) = b \cdot \sum\limits_{v_i \in \supp(\S)} -s_i = -b = -\R(\S) = -\R(\S(t)).
\end{align*}

Next consider the case when $v_{\min}(\S) < b$ and let $\supp(\S)\ = \{w_1, \ldots, w_m\}$. Since $\S$ is an equal revenue distribution, by Lemma~\ref{lemma:public_optimal_auction_for_signals}, we have
\begin{align}
    w_1 \cdot \overline{F}_{\S}(w_1) = \cdots = w_m \cdot \overline{F}_{\S}(w_m) = w_1 = v_{\min}(\S) = \R(\S). 
    \label{eq:public_every_price_in_signal_raises_the_optimal_revenue}
\end{align}

Let $\M^\ast$ be an optimal solution to \textsf{PublicContinuous}$(\D(t))$. Note that $\M^\ast$ is not necessarily unique, and further, can change as $t$ changes. Since $\supp(\S) = \supp(\D(t))$, by Theorem~\ref{theorem:public_optimal_revenue_is_a_distribution_over_posted_price_revenues}, the revenue $\Tilde{\R}(t)$ achieved by $\M^\ast$ in \textsf{PublicContinuous}$(\D(t))$ 
is the revenue of a distribution of posted prices $\{w'_1, \ldots, w'_{m'}\} \subseteq \supp(\S)$ with weights $\delta_1, \ldots, \delta_{m'} \in (0,1]$, where $\sum_{j=1}^{m'} \delta_j = 1$.

At any time $t' \in (t_{h-1},t_{h})$,  the function $\Tilde{\R}_{\M^\ast}(t')$ -- the revenue of $\M^\ast$ over $\D(t')$ -- is linearly decreasing. We now calculate the rate of decrease at time $t'$. 
By Theorem~\ref{theorem:public_optimal_revenue_is_a_distribution_over_posted_price_revenues},
\begin{align*}
    \dv{\Tilde{\R}_{\M^\ast}(t')}{t'} &= \dv{t'} \left( \sum\limits_{j=1}^{m'} \big( \delta_j \cdot w'_j \cdot \sum_{i: v_i \geq w'_j} f_i(t') \big) \right)\\ 
    &= \sum\limits_{j=1}^{m'} \left( -\delta_j \big( w'_j \cdot \overline{F}_{\S}(w'_j) \big) \right) \tag*{(by definition that $\dv{\Vec{\mathbf{f}}}{t'} = -\Vec{\mathbf{s}})$}\\
    &= 
    -\sum\limits_{j=1}^{m'} \big( \delta_j \cdot \R({\S}) \big) = -\R({\S}),
\end{align*}
where the second last equality is by Eq. (\ref{eq:public_every_price_in_signal_raises_the_optimal_revenue}), and the last because $\sum_{j=1}^{m'} \delta_j = 1$.

By Lemma~\ref{lemma:public_revenue_is_convex_for_every_mechanism}, $\Tilde{\R}(t')$ is a maximum of linear revenue functions of bounded range, one for each feasible auction.  Now applying the Envelope Theorem (Theorem~2 in~\cite{Milgrom}), we have that $\Tilde{\R}(t')$  is differentiable, and further, $\dv{\Tilde{\R}(t)}{t} = - \R(\S) =  -\R(\S(t))$. Indeed, $\Tilde{\R}(t')$  must be a linear function for $t' \in (t_{h-1},t_{h})$.


Note that $\S(t)$ and thus $-\R(\S(t))$ changes only when some valuation in the residual prior $\D(t)$ is exhausted. This happens finitely many times throughout the process. Hence $\dv{\Tilde{\R}(t)}{t} = -\R(\S(t))$ is a piecewise constant function with finitely many discontinuities, and is thus Riemann integrable. Recall $\Tilde{\R}(0) = \Tilde{\R}(\D)$ and $\Tilde{\R}(1) = 0$. Also recall every signal $\S_h$ in Algorithm~\ref{alg:main_algorithm_public} is associated with weight $\weight_h = t_h - t_{h-1}$. Therefore we have
\begin{align}
    \Tilde{\R}(\D) &= \Tilde{\R}(0) - \Tilde{\R}(1) = \int_{t=0}^{1} \big( -\dv{\Tilde{\R}(t)}{t} \big) \dd{t} = \int_{t=0}^{1} \R(\S(t)) \dd{t} \tag*{}\\
    &= \sum\limits_{h=1}^{H} \int_{t=t_{h-1}}^{t_{h}} \R(\S(t)) \dd{t} = \sum\limits_{h=1}^{H} \big( \weight_h \cdot \R(\S_h) \big) = \R(\Z^{\ast}_{\D}). \label{eq:public_total_revenue_in_signals_is_equal_to_continuous_objective}
\end{align}

We then observe that it is always feasible for the seller to {\em ignore} the signals: Consider any arbitrary optimal auction $\hat{\M} = (\{ \hat{p}_i\},\{ \hat{x}_i\})$ for \textsf{Public}$(\D)$. If the seller implements $\hat{\M}$ as the auction for each signal $\S_h$ created by Algorithm~\ref{alg:main_algorithm_public}, by Bayes plausibility, the resulting revenue is given by
\begin{align*}
    \sum\limits_{h=1}^{H} \Big( \weight_h \cdot \sum\limits_{i=1}^{n} \big(  f_{\S_h}(v_i) \cdot \hat{p}_i \big) \Big) &= \sum\limits_{i=1}^{n} \Big( \sum\limits_{h=1}^{H} \big( \weight_h \cdot f_{\S_h}(v_i) \cdot \hat{p}_i \big) \Big) = \sum\limits_{i=1}^{n} \big( f_{\D}(v_i) \cdot \hat{p}_i \big) = \R(\D). 
\end{align*} 

Since the above is the total revenue raised by some auction $\hat{\M}$ over all signals $\{\S_h\}$, the total revenue $\R(\Z^{\ast}_{\D})$ raised by implementing the optimal auction for each signal $\S_h$ is at least as much. Therefore, we have $ \R(\D) \leq \R(\Z^{\ast}_{\D})$. Combining this with Eq. (\ref{eq:public_total_revenue_in_signals_is_equal_to_continuous_objective}), and observing that \textsf{PublicContinuous}$(\D)$ relaxes \textsf{Public}$(\D)$, we have:
\begin{align*}
    \R(\D) &\geq \Tilde{\R}(\D) = \R(\Z^{\ast}_{\D}) \geq  \R(\D). 
\end{align*}
Hence, all inequalities must be equalities, which proves the theorem.
\end{proof}

In Theorem~\ref{theorem:public_revenue_preservation} above, we have proved that the process in Algorithm~\ref{alg:main_algorithm_public} preserves the seller's expected revenue. By Lemma~\ref{lemma:public_efficiency_of_scheme}, $\Z^{\ast}_{\D}$ also achieves efficiency, and thus maximizes social welfare. Hence $\Z^{\ast}_{\D}$ must maximize the  consumer surplus. This implies that the analog of Theorem~\ref{theorem:bbm} on buyer optimality holds for the public budget case:

\begin{theorem}[Buyer optimality for public budgets] \label{theorem:public_buyer_optimality}
Suppose $k=1$ in \textsf{Budgets}$(\D)$ for some prior $\D$. Then there exists a signaling scheme $\Z^{\ast}_{\D}$ (given by Algorithm~\ref{alg:main_algorithm_public}) that guarantees:\\ (1) $\W(\Z^{\ast}_{\D}) = \W^\ast(\D)$; (2) $\R(\Z^{\ast}_{\D}) = \R(\D)$; and (3) $\CS(\Z^{\ast}_{\D}) = \opt(\D) = \W^\ast(\D) - \R(\D)$.
\end{theorem}

\paragraph{Discussion.}
The nice aspect of our proof approach is twofold. First, we can derive the following corollary showing the existence of a common revenue-optimal auction throughout the process.

\begin{restatable}[]{corollary}{optimalauctionstaysoptimal}
\label{corollary:optimal_mechanism_stays_optimal}
Any optimal auction $\M^\ast(t')$ for \textsf{PublicContinuous}$(\D(t'))$ at time $t'$ in Algorithm~\ref{alg:main_algorithm_public} stays revenue optimal throughout the course of the algorithm (i.e., for all $t > t'$). 
\end{restatable}
\begin{proof}
Recall that we have $\Tilde{\R}(\D) = \R(\D)$, where the former is the optimal objective of \textsf{PublicContinuous}$(\D)$ and the latter is the optimal objective of \textsf{Public}$(\D)$. 

Suppose $\M^\ast(t')$ is optimal in \textsf{PublicContinuous}$(\D(t'))$ at $t = t'$ corresponding to some equal revenue signal $\S(t')$. Then by the above, $\M^\ast(t')$ is also optimal in \textsf{Public}$(\D(t'))$. Also by Theorem~\ref{theorem:public_revenue_preservation} we have $\dv{\Tilde{\R}(t')}{t} = \dv{\Tilde{\R}_{\M^\ast(t')}(t)}{t} = -\R(\S(t'))$ (recall $\Tilde{\R}_{\M^\ast(t')}(t)$ denotes the objective achieved by $\M^\ast(t')$ in \textsf{PublicContinuous}$(\D(t))$). The above means $\M^\ast(t')$ remains optimal (both w.r.t. \textsf{PublicContinuous}$(\D(t))$ and \textsf{Public}$(\D(t))$) till the next time $\S(t)$ changes (say at $t = t''$). Since the revenue of any auction is continuous in $t$, this implies $\M^\ast(t')$ is still optimal at $t = t''$. Take $\M^\ast(t'') = \M^\ast(t')$. Repeating this argument shows $\M^\ast(t')$ remains optimal for all $t > t'$.
\end{proof}

Secondly, our approach is extensible to more complex settings in the following sense: Theorem~\ref{theorem:public_optimal_revenue_is_a_distribution_over_posted_price_revenues} yields a characterization of the revenue optimal auction in the  specific case of public budgets. We use this to prove the first claim in Theorem~\ref{theorem:public_revenue_preservation}, that the rate of decrease of revenue of the optimal auction is equal to the revenue of the signal constructed. The rest of the proof is generic in that it invokes the Envelope Theorem on $\tilde{\R}(t)$.  Our proof for the FedEx case simply reuses the generic portion, along with a specialized characterization of the optimal auction there.

\section{Main Result: Signaling Scheme with Private Deadlines}
\label{section:deadlines}

We now focus on the case with private valuation-deadline pairs $(v,d) \sim \D$, where $\D_j$ denotes the marginal distribution of $v$ given $d = j$. This is the so-called FedEx problem~\cite{Fiat16}. The non-trivial aspect now is the construction of the signals themselves. We first generalize the notion of equal revenue signals in Section~\ref{subsec:signals_private}, and outline the corresponding signaling algorithm (Algorithm~\ref{alg:main_algorithm_deadlines}). We then proceed to show that Algorithm~\ref{alg:main_algorithm_deadlines} is buyer optimal, using the same plan of attack in Section~\ref{section:public_continuous}. 


We begin with some notation. For an arbitrary prior $\A$ with $\supp(\A) \subseteq \supp(\D) = \{ v_1, \ldots, v_n\} \times \{1, \ldots, k\}$, we denote by $\A_j$ the conditional distribution of $\A$ given $d = j$. Further, let $f_{\A_j}(v_i) \coloneqq \Pr_{(v,d) \sim \A}[v = v_i \mid d = j]$, and $\underline{F}_{\A_j}(v_i) \coloneqq \Pr_{(v,d) \sim \A}[v \leq v_i \mid d = j]$. Further, define $\values(\A) \coloneqq \{ v_i > 0 \mid \sum_{j=1}^{k} f_{\A_j}(v_i) > 0\}$ as the set of values with non-zero support in $\A$, and similarly, $\values(\A_j) \coloneqq \{ v_i > 0 \mid f_{\A_j}(v_i) > 0 \}$. Finally, let $v_{\min}(\A) \coloneqq \min \{ v_i > 0 \mid \sum_{j=1}^{k} f_{\A_j}(v_i) > 0 \}$, and $v_{\min}(\A_j) \coloneqq \min \{ v_i > 0 \mid f_{\A_j}(v_i) > 0 \}$.

\subsection{Generalized Equal Revenue Signals} 
\label{subsec:signals_private}
One natural approach to designing a signaling scheme is to apply the algorithm in~\cite{Bergemann15} to the marginal $\D_j$ induced by each deadline separately. However, consider the following example:

\begin{example} \label{example:deadlines_running_example_inputs}
\sloppy Let the initial prior $\D$ be supported on $(v,d) = \{(2,1), (3,1), (1,2), (2,2), (3,4), (4,3)\}$ with uniform probability, as shown in Table~\ref{tab:deadlines_running_example}. Accordingly, the maximum social welfare $\W^\ast_{\D}$ is given by $\frac{2+3+1+2+4+3}{6} = \frac{5}{2}$. The revenue maximizing auction posts a fixed price of $2$ to all types. This auction raises an expected revenue of $\R^\ast_{\D} = 2 \cdot \frac{5}{6} = \frac{5}{3}$ and thus an expected consumer surplus of $\frac{5}{2} - \frac{5}{3} = \frac{5}{6}$. 

On this example, suppose we run the algorithm in~\cite{Bergemann15} separately for each deadline. Then, this results in a price of $2$ when $d=1$, a price of $1$ when $d = 2$, and prices of $4$ and $3$ respectively when $d=3$ and $d=4$. It is easy to check that this raises consumer surplus $\frac{1}{3}$, which is much smaller than the optimal surplus of $\frac{5}{6}$.
\end{example}

\begin{table}[htbp]
\centering
\begin{tabular}{|l|l|l|l|l|}
\hline
$\Pr[\D = (v,d)]$ & $v=1$ & $v=2$ & $v=3$ & $v=4$ \\ \hline
$d=1$  & $0$   & $1/6$ & $1/6$ & $0$   \\ \hline
$d=2$  & $1/6$ & $1/6$ & $0$   & $0$   \\ \hline
$d=3$  & $0$   & $0$   & $0$ & $1/6$ \\ \hline
$d=4$  & $0$   & $0$   & $1/6$ & $0$ \\ \hline
\end{tabular}
\caption{\label{tab:deadlines_running_example} Probability values $\Pr_{(v,d) \sim \D}[(v,d) = (\cdot, \cdot)]$ in the running example.}
\end{table}

The main reason that the scheme that separately develops signals for each deadline does not raise optimal consumer surplus is because it reveals the deadline of the buyer, which provides the seller with too much information. We therefore need a different and novel signaling scheme that can ``blur" the deadline information in addition to the value information. Our key idea is to define signals that continuously pull mass from {\em all marginals} $\D_j$ at once, albeit in an equal revenue fashion. 

\begin{definition}[Lower Envelope] \label{def:lower_envelope}
Given  prior $\A$ with $\supp(\A) \subseteq \supp(\D)$, for all $j \in [1,k]$, let
\[ \hat{i}_j \coloneqq \max\{i : \underline{F}_{\A_{j'}}(v_i) = 0 \quad \forall j' \in [j,k] \}\]
denote the largest $i$ such that no buyer with valuation at most $v_i$ and deadline at least $j$ exists in $\A$. Note that $\hat{i}_j = 0$ if $f_{\A_{j'}}(v_1) > 0$ for some $j' \in [j,k]$. Let $\hat{i}_{k+1} = n$. The lower envelope of $\A$ is defined as
\[ \textsf{LE}(\A) \coloneqq \{ (v_i, j) \mid \big( f_{\A_j}(v_i) > 0 \big) \land \big( \hat{i}_j < i \leq \hat{i}_{j+1} \big) \}. \]
We say a value-deadline pair $(v_i, j)$ is \emph{on the lower envelope of $\A$} if $(v_i, j) \in \textsf{LE}(\A)$. Two pairs $(v_a, r), (v_b, r') \in \textsf{LE}(\A)$ where $a < b$ and $r' \geq r$ are \emph{consecutive points on $\textsf{LE}(\A)$} if there is no $i \in (a,b)$ and some $j$ such that $(v_i, j)$ is on the lower envelope of $\A$.
\end{definition}

An immediate observation is that for any $\A$, its lower envelope $\textsf{LE}(\A)$ does not contain two different valuation-deadline pairs with the same valuation:

\begin{observation} \label{observation:deadlines_le_is_a_function}
For any prior $\A$, if $(v_i, j), (v_i, j') \in \textsf{LE}(\A)$ for some $i$, then $j = j'$.
\end{observation}

\begin{definition}[\textbf{E}qual revenue \textbf{L}ower \textbf{E}nvelope $\S^{ELE}_{\A}$] \label{def:ele}

For an arbitrary prior $\A$ with $\supp(\A) \subseteq \supp(\D)$, let $\textsf{LE}(\A)$ be supported on $\{ (v'_1, j'_1), (v'_2, j'_2), \ldots, (v'_m, j'_m)\}$, where $0 < v_{\min}(\A) = v'_1 < v'_2 < \cdots < v'_m \leq v_n$.\footnote{Note that all $v'_i$'s are unique by Observation~\ref{observation:deadlines_le_is_a_function}.} We define the \emph{Equal Revenue Lower Envelope} signal $\S^{ELE}_{\A}$ for $\A$ to be the equal revenue distribution over $\{(v'_i, j'_i)\}$, i.e.,

\[ \Pr_{(v,d) \sim \S^{ELE}_{\A}}[v \geq v'_1] \cdot v'_1 = \Pr_{(v,d) \sim \S^{ELE}_{\A}}[v \geq v'_2] \cdot v'_2 = \cdots = \Pr_{(v,d) \sim \S^{ELE}_{\A}}[v \geq v'_m] \cdot v'_m = v'_1 = v_{\min}(\A). \]
\end{definition}

In other words, when disregarding the deadlines (and thus treating $\S^{ELE}_{\A}$ as a distribution of $v$), every valuation with nonzero probability mass in its support is an optimal monopoly price. Analogous to Definition~\ref{definition:equi_revenue_distribution}, this distribution is unique given $\textsf{LE}(\A)$.  

\begin{example}
Let $\A = \D$ be the input prior $\D$ in Example~\ref{example:deadlines_running_example_inputs}. Then the lower envelope \textsf{LE}$(\D)$ is given by $\{(1,2), (2,2), (3,4)\}$, and the corresponding equal revenue lower envelope signal $\S^{ELE}_{\D}$ has $\Pr[(1,2)] = 1/2$, $\Pr[(2,2)] = 1/6$, and $\Pr[(3,4)] = 1/3$. We have
\[ \Pr_{(v,d) \sim \S^{ELE}_{\D}}[v \geq 1] \cdot 1 = \Pr_{(v,d) \sim \S^{ELE}_{\D}}[v \geq 2] \cdot 2 = \Pr_{(v,d) \sim \S^{ELE}_{\D}}[v \geq 3] \cdot 3 = 1. \] 
\end{example}

We have the following observation:

\begin{observation} \label{observation:ele_signal_property}
For an arbitrary prior $\A$ with $\supp(\A) \subseteq \supp(\D)$, for any $i,j$ such that $(v_i, j) \in \supp(\S^{ELE}_{\A})$, it holds that $\underline{F}_{(\S^{ELE}_{\A})_{j'}}(v_i) = 0$ for all $j < j' \leq k$.
\end{observation}

\paragraph{Signaling Algorithm.} 
For any time $t \in [0,1]$, we now let $\mathbf{F}(t) = [f_{ij}(t)]$ be an $n \times k$ matrix function representing the residual prior, where $f_{ij}(t)$ represents the remaining probability mass on type $(v_i, j)$ at time $t$. Similar to the public budget case, let $\D(t)$ denote the probability distribution obtained by placing the remaining probability mass $1 - \sum_{j=1}^{k} \sum_{i=1}^{n} f_{ij}(t)$ at $(v, d) = (0, 0)$. We omit considering $(0,0)$ as part of the support of $\D(t)$. Therefore we define $\supp(\D(t)) \coloneqq \{ (v_i, j) \mid v_i > 0, \, f_{ij}(t) > 0\}$. For each deadline $d = j$, we denote the marginal distribution of $\D(t)$ as $\D_j(t)$. We therefore have: $\values(\D(t)) \coloneqq \{ v_i > 0 \mid \sum_j f_{ij}(t) > 0\}$, $\values(\D_j(t)) \coloneqq \{ v_i > 0 \mid f_{ij}(t) > 0 \}$, and $v_{\min}(\D_j(t)) \coloneqq \min \{ v_i > 0 \mid f_{ij}(t) > 0 \}$.

We now start with the prior distribution $\D(0) = \D$ and let $f_{ij}(0) = \Pr[\D = (v_i, j)]$ for every $i,j$. Our algorithm continuously takes away probability mass from $\mathbf{F}(t)$ and transfers it to the constructed signals, terminating when $\mathbf{F}(t)$ becomes $\mathbf{0}$ at time $t = T$. 
At any time $t \in [0,T)$, denote $\S(t) = \S^{ELE}_{\D(t)}$ over $\textsf{LE}(\D(t))$ (Definition~\ref{def:ele}), $s_{ij}(t) = \Pr_{(v,d) \sim \S^{ELE}_{\D(t)}}[(v,d) = (v_i, j)]$, and $\mathbf{S}(t) \coloneqq [s_{ij}(t)]$.  
As $\mathbf{S}(t)$ depends on $\values(\D(t))$ but not $\{f_{ij}(t)\}$, it is fixed as long as $\values(\D(t))$ does not change. Our algorithm continuously reduces $\mathbf{F}(t)$ at rate $\mathbf{S}(t)$ until $\mathbf{F}(t)$ becomes $\mathbf{0}$:

\begin{align}
    \dv{\mathbf{F}}{t} = -\mathbf{S}(t). \label{eq:deadlines_differential_equation}
\end{align}

Since $\sum_{i,j} s_{ij}(t) = 1$, the rate of decrease of $\sum_{i,j} f_{ij}(t)$ is $1$. Since $\sum_{i,j} f_{ij}(0) = 1$, we have $T = 1$.

\paragraph{Signals constructed.} We say the type $(v_i, j)$ is exhausted at time $t$ if $f_{ij}(t) = 0$ but $f_{ij}(t') > 0$ for all $t' < t$. Therefore, $\values(\D(t))$ changes only when some type is exhausted. For each maximal time interval $t \in [t_1, t_2)$ in which $\values(\D(t))$ remains fixed, the final scheme includes a corresponding signal $\S$ with weight $(t_2 - t_1)$ so that $\S(t) = \S$ for $t \in [t_1, t_2)$. Since $\S(t)$ changes only if some element in $\mathbf{F}(t)$ becomes zero, the number of signals constructed is finite. 

The overall signaling scheme is described in Algorithm~\ref{alg:main_algorithm_deadlines}.  

\begin{algorithm}[htbp]  
    \caption{Continuous Algorithm for Deadlines Setting}
    \label{alg:main_algorithm_deadlines}
    \begin{algorithmic}[1]
        \Require $\D$
        \Ensure $\Z = \Z^{\ast}_{\D}$
        \State $t_0 \gets 0$; \ $\D(t_0) \gets \D$; \ $\mathbf{F}(0) \gets [\Pr_{(v,d) \sim \D}[(v,d) = (v_i, j)]]$; 
        \For {$h \in \{1, \ldots, H\}$}
            \State $t \gets t_{h-1}$
            \State $\S(t_{h-1}) \gets$ Equal revenue distribution on $\supp(\D(t_{h-1}))$ 
            \State Run Equation~(\ref{eq:deadlines_differential_equation}) using $\mathbf{S}(t)$ as density of $\S(t_{h-1})$ till some type's support in $\mathbf{F}(t)$ is exhausted at time $t = t_h$
            \State $\D(t_h) \gets$ distribution induced by $\mathbf{F}(t_h)$
            \State $\weight^{\ast}_h \gets t_{h} - t_{h-1}$; \  $\S^{\ast}_h \gets \S(t_{h-1})$; \ $\Z \gets \Z \cup \{(\weight^{\ast}_h, \S^{\ast}_h)\}$
        \EndFor
        \Return $\Z$

    \end{algorithmic}
\end{algorithm}

\begin{example}
Table~\ref{tab:running_example_full_signals} presents the execution of Algorithm~\ref{alg:main_algorithm_deadlines} on the instance in Example~\ref{example:deadlines_running_example_inputs}. Consider the resulting signals $\{S_1, \ldots, \S_6\}$. As all signals are of the lower-envelope fashion, the signaling scheme suggests $v_{\min}(\S_i)$ as the posted price for $\S_i$, which is $\{1,2,2,2,2,2\}$, respectively. Accordingly, the expected revenue raised by the seller under this signaling scheme is given by $\sum_i v_{\min}(\S_i) \cdot \weight_i = 1 \cdot \frac{24}{72} + 2 \cdot \frac{6+12+12+15+3}{72} = \frac{1}{3} + \frac{4}{3} = \frac{5}{3} = \R^\ast_{\D}$. The item always sells, and the maximum expected consumer surplus $\W^\ast_{\D} - \R^\ast_{\D} = \frac{5}{2} - \frac{5}{3} = \frac{5}{6}$ is achieved. Therefore,  Algorithm \ref{alg:main_algorithm_deadlines} is buyer optimal for $\D$.
\end{example}

\begin{table}[htbp]
\centering
\addtolength{\tabcolsep}{-1pt}
\begin{subtable}[c]{0.4\textwidth}
\centering
\begin{tabular}{|l|l|l|l|l|}
\hline
$\Pr[\D]$ & $v=1$ & $v=2$ & $v=3$ & $v=4$ \\ \hline
$d=1$  & -   & $12/72$ & $12/72$ & -   \\ \hline
$d=2$  & $12/72$ & $12/72$ & -   & -   \\ \hline
$d=3$  & -   & -   & - & $12/72$ \\ \hline
$d=4$  & -   & -   & $12/72$ & - \\ \hline
\end{tabular}
\subcaption{\footnotesize The initial prior $\D(0) = \D$ at $t_0 = 0$.}
\end{subtable}
\begin{subtable}[c]{0.4\textwidth}
\centering
\begin{tabular}{|l|l|l|l|l|}
\hline
$\Pr[\S]$ & $v=1$ & $v=2$ & $v=3$ & $v=4$ \\ \hline
$d=1$  & -   & - & - & -   \\ \hline
$d=2$  & $12/72$ & $4/72$ & -   & -   \\ \hline
$d=3$  & -   & -   & - & - \\ \hline
$d=4$  & -   & -   & $8/72$ & - \\ \hline
\end{tabular}
\subcaption{\footnotesize \label{tab:running_example_full_signal_1} Signal $\S_1 = \S(0)$ multiplied by $\weight_1 = \frac{24}{72}$.}
\end{subtable}

\vspace{5pt}

\begin{subtable}[c]{0.4\textwidth}
\centering
\begin{tabular}{|l|l|l|l|l|}
\hline
$\Pr[\D]$ & $v=1$ & $v=2$ & $v=3$ & $v=4$ \\ \hline
$d=1$  & -   & $12/72$ & $12/72$ & -   \\ \hline
$d=2$  & -   & $8/72$ & -   & -   \\ \hline
$d=3$  & -   & -   & - & $12/72$ \\ \hline
$d=4$  & -   & -   & $4/72$ & - \\ \hline
\end{tabular}
\subcaption{\footnotesize The residual prior $\D(\frac{24}{72})$ at $t_1 = \frac{24}{72}$.}
\end{subtable}
\begin{subtable}[c]{0.4\textwidth}
\centering
\begin{tabular}{|l|l|l|l|l|}
\hline
$\Pr[\S]$ & $v=1$ & $v=2$ & $v=3$ & $v=4$ \\ \hline
$d=1$  & -   & - & - & -   \\ \hline
$d=2$  & - & $2/72$ & -   & -   \\ \hline
$d=3$  & -   & -   & - & - \\ \hline
$d=4$  & -   & -   & $4/72$ & - \\ \hline
\end{tabular}
\subcaption{\footnotesize \label{tab:running_example_full_signal_2} Signal $\S_2 = \S(\frac{24}{72})$ multiplied by $\weight_2 = \frac{6}{72}$.}
\end{subtable}

\vspace{5pt}

\begin{subtable}[c]{0.4\textwidth}
\centering
\begin{tabular}{|l|l|l|l|l|}
\hline
$\Pr[\D]$ & $v=1$ & $v=2$ & $v=3$ & $v=4$ \\ \hline
$d=1$  & -   & $12/72$ & $12/72$ & -   \\ \hline
$d=2$  & -   & $6/72$ & -   & -   \\ \hline
$d=3$  & -   & -   & - & $12/72$ \\ \hline
$d=4$  & -   & -   & - & - \\ \hline
\end{tabular}
\subcaption{\footnotesize The residual prior $\D(\frac{30}{72})$ at $t_2 = \frac{30}{72}$.}
\end{subtable}
\begin{subtable}[c]{0.4\textwidth}
\centering
\begin{tabular}{|l|l|l|l|l|}
\hline
$\Pr[\S]$ & $v=1$ & $v=2$ & $v=3$ & $v=4$ \\ \hline
$d=1$  & -   & - & - & -   \\ \hline
$d=2$  & - & $6/72$ & -   & -   \\ \hline
$d=3$  & -   & -   & - & $6/72$ \\ \hline
$d=4$  & -   & -   & - & - \\ \hline
\end{tabular}
\subcaption{\footnotesize \label{tab:running_example_full_signal_3} Signal $\S_3 = \S(\frac{30}{72})$ multiplied by $\weight_3 = \frac{12}{72}$.}
\end{subtable}

\vspace{5pt}

\begin{subtable}[c]{0.4\textwidth}
\centering
\begin{tabular}{|l|l|l|l|l|}
\hline
$\Pr[\D]$ & $v=1$ & $v=2$ & $v=3$ & $v=4$ \\ \hline
$d=1$  & -   & $12/72$ & $12/72$ & -   \\ \hline
$d=2$  & -   & - & -   & -   \\ \hline
$d=3$  & -   & -   & - & $6/72$ \\ \hline
$d=4$  & -   & -   & - & - \\ \hline
\end{tabular}
\subcaption{\footnotesize The residual prior $\D(\frac{42}{72})$ at $t_3 = \frac{42}{72}$.}
\end{subtable}
\begin{subtable}[c]{0.4\textwidth}
\centering
\begin{tabular}{|l|l|l|l|l|}
\hline
$\Pr[\S]$ & $v=1$ & $v=2$ & $v=3$ & $v=4$ \\ \hline
$d=1$  & -   & $4/72$ & $2/72$ & -   \\ \hline
$d=2$  & - & - & -   & -   \\ \hline
$d=3$  & -   & -   & - & $6/72$ \\ \hline
$d=4$  & -   & -   & - & - \\ \hline
\end{tabular}
\subcaption{\footnotesize \label{tab:running_example_full_signal_4} Signal $\S_4 = \S(\frac{42}{72})$ multiplied by $\weight_4 = \frac{12}{72}$.}
\end{subtable}

\vspace{5pt}

\begin{subtable}[c]{0.4\textwidth}
\centering
\begin{tabular}{|l|l|l|l|l|}
\hline
$\Pr[\D]$ & $v=1$ & $v=2$ & $v=3$ & $v=4$ \\ \hline
$d=1$  & -   & $8/72$ & $10/72$ & -   \\ \hline
$d \geq 2$  & -   & - & -   & -   \\ \hline
\end{tabular}
\subcaption{\footnotesize The residual prior $\D(\frac{54}{72})$ at $t_4 = \frac{54}{72}$.}
\end{subtable}
\begin{subtable}[c]{0.4\textwidth}
\centering
\begin{tabular}{|l|l|l|l|l|}
\hline
$\Pr[\S]$ & $v=1$ & $v=2$ & $v=3$ & $v=4$ \\ \hline
$d=1$  & -   & $5/72$ & $10/72$ & -   \\ \hline
$d \geq 2$  & -   & - & -   & -   \\ \hline
\end{tabular}
\subcaption{\footnotesize \label{tab:running_example_full_signal_5} Signal $\S_5 = \S(\frac{54}{72})$ multiplied by $\weight_5 = \frac{15}{72}$.}
\end{subtable}

\vspace{5pt}

\begin{subtable}[c]{0.4\textwidth}
\centering
\begin{tabular}{|l|l|l|l|l|}
\hline
$\Pr[\D]$ & $v=1$ & $v=2$ & $v=3$ & $v=4$ \\ \hline
$d=1$  & -   & $3/72$ & - & -   \\ \hline
$d \geq 2$  & -   & - & -   & -   \\ \hline
\end{tabular}
\subcaption{\footnotesize The residual prior $\D(\frac{69}{72})$ at $t_5 = \frac{69}{72}$.}
\end{subtable}
\begin{subtable}[c]{0.4\textwidth}
\centering
\begin{tabular}{|l|l|l|l|l|}
\hline
$\Pr[\S]$ & $v=1$ & $v=2$ & $v=3$ & $v=4$ \\ \hline
$d=1$  & -   & $3/72$ & - & -   \\ \hline
$d \geq 2$  & -   & - & -   & -   \\ \hline
\end{tabular}
\subcaption{\footnotesize \label{tab:running_example_full_signal_6} Signal $\S_6 = \S(\frac{69}{72})$ multiplied by $\weight_6 = \frac{3}{72}$.}
\end{subtable}

\caption{\label{tab:running_example_full_signals} Timeline of the algorithm applied on Example~\ref{example:deadlines_running_example_inputs}. Values of probabilistic masses in the constructed signals are multiplied by the corresponding weights of the signals. All values are divided by $72$ for simplicity.}
\end{table}

Similar to  Algorithm~\ref{alg:main_algorithm_public}, the signals created by Algorithm~\ref{alg:main_algorithm_deadlines} are Bayes plausible:

\begin{observation} \label{observation:deadlines_plausibility_of_signals}
\sloppy
    For any arbitrary $\D$, let $\Z^{\ast}_{\D} =  \{ (\weight^{\ast}_h, \S^{\ast}_h) \}_{h \in [H]}$ be the set of signals output by Algorithm~\ref{alg:main_algorithm_deadlines} taking $\D$ as input. Then we have  $\sum_{h=1}^{H} \weight^{\ast}_h \S^{\ast}_h = \D$.
\end{observation}

\subsection{Optimal Auction for Signals}
\label{subsec:deadlines_welfare}

In the following, we show the counterparts of Lemmas~\ref{lemma:public_optimal_auction_for_signals} and~\ref{lemma:public_efficiency_of_scheme} in the deadlines context, and that Algorithm~\ref{alg:main_algorithm_deadlines} guarantees efficiency. Let $\S_h \in \Z^{\ast}_{\D}$ denote a signal created by Algorithm~\ref{alg:main_algorithm_deadlines} for the prior $\D$. This theorem is proved in Appendix~\ref{app:deadlines}.

\begin{restatable}[]{lemma}{deadlinesoptimalauction}
\label{lemma:deadlines_optimal_auction}
There is an optimal auction for $\S_h$ that posts a price of $v_{\min}(\S_h)$. Further, for every $v_i \in \values(\S_h)$, we have $v_i \cdot \Pr_{(v,d) \sim \S_h}[v \geq v_i] = v_{\min}(\S_h) $. 
\end{restatable}
\begin{proof}
By Observation~\ref{observation:deadlines_le_is_a_function} and Definition~\ref{def:ele}, we have that for any $i$, if $(v_i, j), (v_i, j') \in \supp(\S_h)$, then $j = j'$; furthermore, if $(v_{i'}, j), (v_{i''}, j) \in \supp(\S_h)$ for some $i' \leq i''$, then for any $i \in [i', i'']$, $(v_i, j') \in \supp(\S_h)$ implies $j' = j$. Therefore, we can denote $\values(\S_h) = \{ w_1, w_2, \ldots, w_m\}$ and $\supp(\S_h) = \{(w_1, d_1), (w_2, d_2), \ldots, (w_m, d_m)\}$, where $w_1 < w_2 < \cdots < w_m$ and $d_1 \leq d_2 \leq \cdots \leq d_m$. Any incentive compatible auction for $\S_h$ can thus be represented by $\{(p_i, x_i)\}$, where $p_i$ and $x_i$ are the expected payment and allocation probability, respectively, for the type $(w_i, d_i)$. For simplicity, let $f_{i} = \Pr_{(v,d) \sim \S_h}[v = w_i]$, and thus $\sum_{i=1}^{m} f_i = 1$; also let $\overline{F}_{i} = \Pr_{(v,d) \sim \S_h}[v \geq w_i]$. We have 
\begin{align}
    \overline{F}_{1} \cdot w_1 = \overline{F}_{2} \cdot w_2 = \cdots = \overline{F}_{m} \cdot w_m =  w_1 = v_{\min}(\S_h),
    \label{eq:deadlines_restatement_of_equal_revenue_property}
\end{align} 
and thus for every $i = 1, \ldots, (m-1)$ we have
\begin{align}
    \overline{F}_{i} \cdot w_{i} - \overline{F}_{i+1} \cdot w_{i+1} &= (\overline{F}_{i+1} + f_{i}) \cdot w_{i} - \overline{F}_{i+1} \cdot w_{i+1} \notag \\
    &= f_{i} \cdot w_{i} - \overline{F}_{i+1} \cdot (w_{i+1} - w_{i}) = 0. \label{eq:deadlines_reorganization_of_equal_revenue_property}
\end{align} 

Consider any feasible auction $\{(p_i, x_i)\}$ for $\S_h$. By individual rationality for the type $(w_1, d_1)$, we have $w_1 \cdot x_1 - p_1 \geq 0$, or $p_1 \leq w_1 \cdot x_1$. Next, for all $i = 2, \ldots, m$, to prevent the type $(w_{i}, d_{i})$ to misreport $(w_{i-1}, d_{i-1} \leq d_{i})$, we have $w_i \cdot x_i - p_i \geq w_i \cdot x_{i-1} - p_{i-1}$, or equivalently, $p_i \leq p_{i-1} + w_i \cdot (x_i - x_{i-1})$. Letting $x_0 = 0$, for all $i \in [1,m]$ we can write
\begin{align*}
    f_i \cdot p_i &\leq f_i \cdot \big( p_{i-1} + w_i \cdot (x_i - x_{i-1}) \big)\\
    &\leq f_i \cdot \Big( \big( p_{i-2} + w_{i-1} \cdot (x_{i-1} - x_{i-2}) \big) + w_i \cdot (x_i - x_{i-1}) \Big)\\
    &\leq \cdots \leq f_i \cdot \sum_{i'=1}^{i} \big( w_i \cdot (x_i - x_{i-1}) \big).
\end{align*}
Summing up the above inequality for all $i \in [1,m]$, we have
\begin{align*}
    \sum_{i=1}^{m} ( f_i \cdot p_i ) &\leq \sum_{i=1}^{m} \Big( f_i \cdot \sum_{i'=1}^{i} \big( w_i \cdot (x_i - x_{i-1}) \big) \Big)\\
    &= \sum_{i=1}^{m-1} \Big( x_i \cdot \big( f_i \cdot w_i - \sum_{i'=i+1}^{m} f_{i'} \cdot ( w_{i} - w_{i-1} ) \big) \Big) + x_m \cdot f_m \cdot w_m \\
    &= \sum_{i=1}^{m-1} \Big( x_i \cdot \big( f_i \cdot w_i - \overline{F}_{i+1} \cdot ( w_{i} - w_{i-1} ) \big) \Big) + x_m \cdot f_m \cdot w_m\\
    &= x_m \cdot f_m \cdot w_m \tag*{(by Eq. (\ref{eq:deadlines_reorganization_of_equal_revenue_property}))}\\
    &\leq f_m \cdot w_m. \tag*{($x_m \leq 1$)}
\end{align*}

Notice that the left-hand side of the above inequality, $\sum_{i=1}^{m} ( f_i \cdot p_i )$, is exactly the revenue raised by $\{(p_i, x_i)\}$ in $\S_h$. Furthermore, by Eq. (\ref{eq:deadlines_restatement_of_equal_revenue_property}) we also have $f_m \cdot w_m = \overline{F}_{m} \cdot w_m = w_1 = v_{\min}(\S_h)$. This shows the maximum revenue is upper bounded by $w_1$. But $w_1 = v_{\min}(\S_h)$ is exactly the revenue raised by the feasible auction that posts a fixed price of $w_1 = v_{\min}(\S_h)$. Therefore, this auction is optimal, and the maximum revenue is $w_1 = v_{\min}(\S_h)$. The second statement in the lemma follows directly from Eq. (\ref{eq:deadlines_restatement_of_equal_revenue_property}).
\end{proof}

The characterization of the optimal auction above implies the item always sells in $\Z^{\ast}_{\D}$:
\begin{lemma}[Efficiency of $\Z^{\ast}_{\D}$] \label{lemma:deadlines_efficiency_of_scheme}
For each signal $\S_h \in \Z^{\ast}_{\D}$, there exists a revenue optimal auction that always sells the item. As a consequence, $\W(\Z^\ast_{\D}) = \W^\ast(\D)$.
\end{lemma}

\subsection{Characterization of Optimal Auction for $\D(t)$} 
\label{subsec:deadlines_optimal_auction}
In the following, we analyze the revenue of the signals using the same technique in Section~\ref{subsec:public_optimal_auction}: We make the constraints in \textsf{Deadlines}$(\D(t))$ hold not just for values in $\values(\D(t))$ but for all continuous values $v > 0$. Fix some time $t$, and let $\A = \D(t)$ so that $\supp(\A) \subseteq \supp(\D)$. We describe the linear program with extended domain and IC/IR constraints:

\begin{align*}
    \textsf{DeadlinesContinuous}(\A) & \coloneqq \max_{\{p_j(\cdot)\}, \{x_j(\cdot)\}} \quad \sum_{j=1}^{k} \left( \Pr_{(v,d) \sim \A}[d = j] \cdot \sum_{i=1}^{n} \big( f_{\A_j}(v_i) \cdot p_j(v_i) \big) \right) \span\span\span \\
    \text{s.t.} \quad & v \cdot x_j(v) - p_j(v) \geq v \cdot x_j(v') - p_j(v'), &\quad \forall v, v' \in [0, v_n], \, 1 \leq j \leq k, \\
    \quad & v \cdot x_j(v) - p_j(v) \geq v \cdot x_{j-1}(v) - p_{j-1}(v), &\quad \forall v \in [0, v_n], \, 2 \leq j \leq k, \\
    \quad & v \cdot x_j(v) - p_j(v) \geq 0, &\quad \forall v \in [0, v_n], \, 1 \leq j \leq k, \\
    & 0 \leq x_j(v) \leq 1, &\quad \forall v \in [0, v_n], \, 1 \leq j \leq k. 
\end{align*}

We denote $\R(\A)$ the optimal revenue achievable by \textsf{Deadlines}$(\A)$ and $\Tilde{\R}(\A)$ the optimal revenue achievable by \textsf{DeadlinesContinuous}$(\A)$. Clearly, $\R(\A) \geq \Tilde{\R}(\A)$.\footnote{It follows from~\cite{RubinsteinW,cai2021duality} that these two revenues are equal; however, we will not need this fact in our proof.}

We now present a characterization result for the optimal auction that is a discrete analog of the characterization for continuous priors in~\cite{Fiat16,Devanur17}.  We present a stand-alone and elementary proof for the discrete setting in Appendix~\ref{app:deadlines}. We note that unlike~\cite{Fiat16,Devanur17}, our proof uses convexity of the utility curve in the primal solution instead of invoking duality, and may be of independent interest. This theorem is proved in Appendix~\ref{app:deadlines}.

\begin{restatable}[]{theorem}{deadlinesrevenuedistribution} \label{theorem:deadlines_optimal_revenue_is_a_distribution_over_posted_price_revenues}
For any prior $\A = \D(t)$ such that $\values(\A) = \{w_1, w_2, \ldots, w_m\}$, where $w_1 < w_2 < \cdots < w_m$, there exists weights $\delta^j_1, \delta^j_2, \ldots, \delta^j_{m} \in [0,1]$ for all $j \in [1,k]$ such that the optimal revenue of $\textsf{DeadlinesContinuous}(\A)$ is 
\[ \Tilde{\R}(\A) = \sum_{j=1}^{k} \left( \Pr_{(v,d) \sim \A}[d = j] \cdot \sum_{i'=1}^{m} \big( \delta^j_{i'} \cdot \overline{F}_{\A_j}(w_{i'}) \cdot w_{i'} \big) \right). \]
Furthermore, we have the following properties about the lower envelope $\textsf{LE}(\A)$:
\begin{itemize}
    \item If $(w_{a},r) \in \textsf{LE}(\A)$, then $\delta^j_a = \delta^r_a$ for all $j \ge r$.
    \item If $(w_a, r)$ and $(w_b,r')$ are consecutive points on $\textsf{LE}(\A)$ where $a < b$ and $r' \ge r$, then $\delta^{j}_i = 0$ for all $i \in (a,b)$ and $j \ge r$.
    \item $\sum_{i=1}^{m} \big( \mathbbm{1}^{\textsf{LE}(\A)}_{i} \cdot \delta^k_{i} \big) = 1$, where $\mathbbm{1}^{\textsf{LE}(\A)}_{i}$ equals $1$ if $(w_i, j) \in \textsf{LE}(\A)$ for some $j$, and $0$ otherwise.
\end{itemize}
\end{restatable}

We illustrate this characterization in Figure~\ref{fig:deadlines_revenue_is_a_distribution}. 
\usetikzlibrary{calc,shapes,matrix,arrows,shapes.misc}
\tikzset{
    vertex/.style={circle,draw,minimum size=1.5em},
    edge/.style={-}
}
\tikzset{cross/.style={cross out, draw=black, minimum size=2*(#1-\pgflinewidth), inner sep=0pt, outer sep=0pt},
cross/.default={3.5pt}}

\begin{figure}[t]
\centering
\begin{tikzpicture}[scale=0.70]

    \draw[dashed, thick] (0,  0) -- (12,  0);
    \draw[dashed, thick] (0, -1) -- (12, -1);
    \draw[dashed, thick] (0, -2) -- (12, -2);
    \draw[dashed, thick] (0, -3) -- (12, -3);
    
    \node[] (type1) at (-2, 0)  {$d = 1$};  
    \node[] (type2) at (-2,-1)  {$d = 2$};  
    \node[] (type3) at (-2,-2)  {$d = 3$};
    \node[] (type4) at (-2,-3)  {$d = 4$};
    \node[] (value1) at (1,1)  {$w_1$};
    \node[] (value2) at (3,1)  {$w_2$};
    \node[] (value3) at (5,1)  {$w_3$};
    \node[] (value4) at (6,1)  {$w_4$};
    \node[] (value5) at (8,1)  {$w_5$};
    \node[] (value6) at (10,1)  {$w_6$};
    \node[] (value7) at (11,1)  {$w_7$};
    
    \draw (1, 0) node[cross=5pt, red] {};
    \draw (3, 0) node[cross] {};
    \draw (5, 0) node[cross] {};
    \draw (6, 0) node[cross] {};
    \draw (8, 0) node[cross] {};
    \draw (10, 0) node[cross] {};
    \draw (3, -1) node[cross=5pt, red] {};
    \draw (8, -1) node[cross] {};
    \draw (11, -1) node[cross] {};
    \draw (6, -2) node[cross=5pt, red] {};
    \draw (10, -2) node[cross] {};
    \draw (11, -2) node[cross] {};
    \draw (8, -3) node[cross=5pt, red] {};
    \draw (10, -3) node[cross=5pt, red] {};
    
    \draw [->, blue, thick] (0.8, 0.7) -- (1, 0.2);
    \draw [->, blue, thick] (5.8, 0.7) -- (6, 0.2);
    \draw [->, blue, thick] (2.8, -0.3) -- (3, -0.8);
    \draw [->, blue, thick] (7.8, -0.3) -- (8, -0.8);
    \draw [->, blue, thick] (5.8, -1.3) -- (6, -1.8);
    \draw [->, blue, thick] (9.8, -2.3) -- (10, -2.8);

    \draw [->, blue, dashed] (1, -0.2) -- (1, -0.9);
    \draw [->, blue, dashed] (1, -1.2) -- (1, -1.9);
    \draw [->, blue, dashed] (1, -2.2) -- (1, -2.9);
    \draw [->, blue, dashed] (3, -1.2) -- (3, -1.9);
    \draw [->, blue, dashed] (3, -2.2) -- (3, -2.9);
    \draw [->, blue, dashed] (6, -2.2) -- (6, -2.9);
    
    
    \draw[->, thick] (0, -9) -- (0, -4);
    \node[] (yaxis) at (-1,-4)  {$x_4(\cdot)$};
    \draw[->, thick] (0, -9) -- (12, -9);
    \node[] (xaxis) at (13, -9)  {$v$};
    
    \draw [very thick] (0, -9) -- (1, -9) -- (1, -7.2) -- (3, -7.2) -- (3, -6.6) -- (6, -6.6) -- (6, -5) -- (10, -5) -- (10, -4.4) -- (11, -4.4);
    \draw [dashed] (3, -7.3) -- (3, -8.9);
    \draw [dashed] (6, -6.7) -- (6, -8.9);
    \draw [dashed] (10, -5.1) -- (10, -8.9);
    \node[] (v1down) at (1, -9.5)  {$w_1$};
    \node[] (v3down) at (3, -9.5)  {$w_2$};
    \node[] (v4down) at (6, -9.5)  {$w_4$};
    \node[] (v6down) at (10, -9.5)  {$w_6$};
    
    \draw [dashed] (9.9, -4.4) -- (0.1, -4.4);
    \node[] (1down) at (-0.5, -4.4)  {$1$};
    \node[] (0down) at (-0.5, -9)  {$0$};
    
\end{tikzpicture}
\caption{Illustration of Theorem~\ref{theorem:deadlines_optimal_revenue_is_a_distribution_over_posted_price_revenues}. Crosses represent $(w_i,j)$ pairs in $\supp(\A)$. Red large crosses represent $(w_i,j)$ pairs on the lower envelope $\textsf{LE}(\A) = \supp(\S)$. Blue solid arrows pointing to the $(w_i,j)$ pairs in a slightly tilted angle represent breakpoints in the allocation curve $x_j(\cdot)$ (i.e., $\delta^j_{i} > 0$). For each $(w_a, r) \in \textsf{LE}(\A)$, the dashed arrows pointing directly downwards illustrate the weight $\delta^r_a$ is preserved up to the latest deadline type, i.e., $\delta^j_a = \delta^r_a$ for all $j \geq r$. The lower portion shows the allocation curve $x_4(\cdot)$ for the latest deadline type $d = 4$ with breakpoints at $v \in \{w_1, w_2, w_4, w_6\}$. Note that $x_{04} = 0$ and $x_{64} = 1$, and we have $\delta^4_1 + \delta^4_2 + \delta^4_4 + \delta^4_6 = 1$, where $\delta^4_i = x_{i4} - x_{(i-1)4}$ is the jump length of $x_4(\cdot)$ at $v = w_i$. Since $(w_2, 2)$ and $(w_4, 3)$ are consecutive points on $\textsf{LE}(\A)$, we have $\delta^j_3 = 0$ for all $j \geq 2$, and thus $x_4(v)$ does not increase at $v = w_3$.}
\label{fig:deadlines_revenue_is_a_distribution}
\end{figure}



\subsection{Revenue Preservation in Algorithm~\ref{alg:main_algorithm_deadlines}}
\label{subsec:deadlines_revenue_preservation}

We now prove that Algorithm~\ref{alg:main_algorithm_deadlines} preserves the expected seller revenue, following the same roadmap as in Section~\ref{subsec:public_revenue_preservation}: We argue that the rate of decrease of revenue of \textsf{DeadlinesContinuous}$(\D(t))$ equals the optimal revenue of the signal $\S(t)$.

\paragraph{Convexity of Revenue in $\textsf{DeadlinesContinuous}(\D(t))$.} 
At any time $t$, let $\Tilde{\R}(t)$ denote the optimal revenue of $\textsf{DeadlinesContinuous}(\D(t))$; since \textsf{DeadlinesContinuous}$(\D(0))$ = \textsf{DeadlinesContinuous}$(\D)$ has more constraints than \textsf{Deadlines}$(\D)$, we have $\Tilde{\R}(0) \leq \R(\D)$. Also, $\Tilde{\R}(1) = 0$. 

Similar to Section~\ref{subsec:public_revenue_preservation}, we consider any time interval $(t_{h-1}, t_h)$ in which $\supp(\D(t))$, and hence $\S(t)$, does not change for $t \in (t_{h-1}, t_h)$, and let the signal be $\S$ and its corresponding probability matrix to be $\mathbf{S} = [s_{ij}]$ over $t \in (t_{h-1}, t_h)$. Then for any feasible solution $\M \coloneqq (\{p_{j}(\cdot)\}, \{x_{j}(\cdot)\})$, its revenue in $\textsf{DeadlinesContinuous}(\D(t))$ is given by $\Tilde{\R}_{\M}(t) \coloneqq \sum_{j=1}^{k} \sum_{i=1}^{n} f_{ij}(t) \cdot p_j(v_i).$ Fixing this $\M$, for each $t \in (t_{h-1}, t_h)$ we have

\begin{align}
\dv{\Tilde{\R}_{\M}(t)}{t} =  \dv{t} \left( \sum\limits_{j=1}^{k} \sum\limits_{i=1}^{n} f_{ij}(t) \cdot p_j(v_i) \right) = - \sum\limits_{j=1}^{k} \sum_{v_i \in \values(\D_j(t))} s_{ij} \cdot p_j(v_i),
\label{eq:deadlines_revenue_decrease_in_residual_prior}
\end{align} 
where the last equality uses Eq. (\ref{eq:deadlines_differential_equation}) and the fact that $s_{ij} = 0$ for any $v_i \notin \values(\D_j(t))$ for all $j \in [1,k]$. Therefore, for each $\M$, its revenue $\Tilde{\R}_{\M}(t)$ decreases linearly with time. Further, since the constraints of \textsf{DeadlinesContinuous}$(\D(t))$ do not change with time, $\M$ remains feasible at all points in time $t \in [0,1]$ (the duration of the algorithm), and its revenue $\Tilde{\R}_{\M}(t)$ is continuous at all $t \in [0,1]$ since each $f_{ij}(t)$ changes continuously.  Since $\Tilde{\R}(t) = \max_{\M} \Tilde{\R}_{\M}(t)$, we have:

\begin{lemma} \label{lemma:deadlines_revenue_is_convex_for_every_mechanism}
For any interval $(t_{h-1},t_{h})$ where $\supp(\D(t))$ does not change, the function $\Tilde{\R}(t)$ is convex. Further, the function $\Tilde{\R}(t)$ is continuous for all $t \in [0,1]$,  the entire duration of the algorithm.
\end{lemma}

\paragraph{Revenue of Signals.}  We now quantify the rate of decrease of $\Tilde{\R}(t)$ in the deadlines setting.

\begin{theorem} \label{theorem:deadlines_revenue_preservation}
In any interval $(t_{h-1},t_{h})$ where $\supp(\D(t))$ does not change, the function $\Tilde{\R}(t)$ is linear, and $\dv{\Tilde{\R}(t)}{t} = -\R(\S(t))$, where $\R(\S(t))$ is the optimal revenue of \textsf{Deadlines}$(\S(t))$. Furthermore, at the end of Algorithm~\ref{alg:main_algorithm_deadlines}, it holds that $\R(\Z^{\ast}_{\D}) = \R(\D)$.
\end{theorem}

\begin{proof}
Fix some $t \in (t_{h-1},t_{h})$, and assume $\supp(\D(t))$ does not change in $(t_{h-1},t_{h})$ and thus $\S(t) = \S$. Let $\values(\D(t)) = \{w_1, w_2, \ldots, w_m\}$, where $w_1 < w_2 < \cdots < w_m$. Recall that $\S$ is an equal revenue lower envelope distribution supported on $\textsf{LE}(\D(t))$ (see Definition~\ref{def:ele}), and thus $\values(\S) \subseteq \values(\D(t))$. By Lemma~\ref{lemma:deadlines_optimal_auction}, for all $w_i \in \values(\S)$ we have: 
\begin{align}
    \Pr_{(v,d) \sim \S}[v \geq w_i] \cdot w_i = v_{\min}(\S) = \R(\S). \label{eq:deadlines_every_price_in_signal_raises_the_optimal_revenue}
\end{align}

Let $\M^\ast$ be the (not necessarily unique) revenue maximizing solution to \textsf{DeadlinesContinuous}$(\D(t))$. At any time $t' \in (t_{h-1},t_{h})$,  the function $\Tilde{\R}_{\M^\ast}(t')$ -- the revenue of $\M^\ast$ over  $\D(t')$ -- is linearly decreasing. We now calculate the rate of decrease at time $t'$. By Theorem~\ref{theorem:deadlines_optimal_revenue_is_a_distribution_over_posted_price_revenues}, for each deadline $d = j$, the revenue $\Tilde{\R}(t)$ achieved by $\M^\ast$ in \textsf{DeadlinesContinuous}$(\D(t))$ is a weighted combination of the revenue for posted prices $\{w_1, \ldots, w_m\}$ with weights $\delta^j_1, \ldots, \delta^j_m \in [0,1]$, where $\sum_{i=1}^{m} \big( \mathbbm{1}^{\textsf{LE}(\D(t))}_{i} \cdot \delta^k_{i} \big) = 1$. Therefore, 
\begin{align*}
    \dv{\Tilde{\R}_{\M^\ast}(t')}{t'} &= \dv{t'} \left( \sum\limits_{j=1}^{k} \Big( \Pr_{(v,d) \sim \D(t')}[d = j] \cdot  \sum\limits_{i'=1}^{m} \big( \delta^j_{i'} \cdot w_{i'} \cdot \overline{F}_{\D_j(t')}(w_{i'}) \big) \Big) \right)\\ 
    &= \dv{t'} \left( \sum\limits_{j=1}^{k} \sum\limits_{i'=1}^{m} \big( \delta^j_{i'} \cdot w_{i'} \cdot \sum_{i: v_i \geq w_{i'}} f_{ij}(t') \big) \right) = \sum\limits_{j=1}^{k} \sum\limits_{i'=1}^{m} \left( -\delta^j_{i'} \cdot w_{i'} \cdot \sum_{i: v_i \geq w_{i'}} s_{ij}(t') \right) \tag*{(using $\dv{\mathbf{F}}{t'} = -\mathbf{S})$}\\
    &= -\sum\limits_{i'=1}^{m} \left( w_{i'} \cdot \sum\limits_{j=1}^{k}  \Big( \delta^j_{i'} \cdot \sum_{i: v_i \geq w_{i'}} s_{ij}(t') \Big) \right) \tag*{($\star$)}.
\end{align*}

Consider some $w_{i'} \notin \values(\S)$. Suppose there exist some $a < i' < b$ and $r \leq r'$ such that $(w_a, r), (w_b, r') \in \textsf{LE}(\D(t)) = \supp(\S)$ are consecutive points in $\textsf{LE}(\D(t))$. By Theorem~\ref{theorem:deadlines_optimal_revenue_is_a_distribution_over_posted_price_revenues}, this implies $\delta^j_{i'} = 0$ for all $j \geq r$. Also, for all $v_i \geq w_{i'}$ we have $s_{ij}(t') = 0$ for all $j < r$. Otherwise, there is no $v_i > w_{i'}$ such that $v_i \in \supp(\S)$; in this case we have $s_{ij}(t') = 0$ for all $v_i \geq w_{i'}$ and all $j \in [1,k]$. In both cases, the summation $\sum_{j=1}^{k}  \Big( \delta^j_{i'} \cdot \sum_{i: v_i \geq w_{i'}} s_{ij}(t') \Big)$ evaluates to 0. 

On the other hand, for each $w_{i'} \in \values(\S)$, there is a unique $r$ such that $(w_{i'}, r) \in \textsf{LE}(\D(t)) = \supp(\S)$, which (by Theorem~\ref{theorem:deadlines_optimal_revenue_is_a_distribution_over_posted_price_revenues}) implies $\delta^j_{i'} = \delta^k_{i'}$ for all $j \in [r,k]$. Since for all $v_i \geq w_{i'}$ it still holds that $s_{ij}(t') = 0$ for all $j < r$, we have
\[ \sum_{j=1}^{k}  \Big( \delta^j_{i'} \cdot \sum_{i: v_i \geq w_{i'}} s_{ij}(t') \Big) = \sum_{j=r}^{k} \Big( \delta^j_{i'} \cdot \sum_{i: v_i \geq w_{i'}} s_{ij}(t') \Big) = \delta^k_{i'} \cdot \sum_{i: v_i \geq w_{i'}} \sum_{j=r}^{k} s_{ij}(t') = \delta^k_{i'} \cdot \Pr_{(v,d) \sim \S}[v \geq w_{i'}].\]
Combining the above with Eq. (\ref{eq:deadlines_every_price_in_signal_raises_the_optimal_revenue}), the expression ($\star$) evaluates to
\begin{align*}
    & \, -\sum\limits_{i: \, w_i \in \values(\S)} \left( \delta^k_i \cdot w_i \cdot \Pr_{(v,d) \sim \S}[v \geq w_i] \right) = \   -\sum\limits_{i: \, w_i \in \values(\S)} \big( \delta^k_i \cdot \R_{\S} \big)  = \   -\R({\S}), 
\end{align*}
where the last equality is by $\sum_{i: \, w_i \in \values(\S)} \delta^k_i = \sum_{i=1}^{m} \big( \mathbbm{1}^{\textsf{LE}(\D(t'))}_{i} \cdot \delta^k_{i} \big) = 1$.

Using Lemma~\ref{lemma:deadlines_revenue_is_convex_for_every_mechanism} and again applying the Envelope Theorem~\cite{Milgrom}, the above implies $\Tilde{\R}(t)$ is linear and hence differentiable in the interval $t \in (t_{h-1},t_{h})$,  and $\dv{\Tilde{\R}(t)}{t} = - \R(\S) =  -\R(\S(t))$.

Note that $\S(t)$ and thus $-\R(\S(t))$ changes only when some $(v,d)$-type in the residual prior $\D(t)$ is exhausted; this happens finitely many times throughout the process. Hence $\dv{\Tilde{\R}(t)}{t} = -\R(\S(t))$ is a piecewise constant function with finitely many discontinuities, and is thus Riemann integrable. Similar to that in the proof of Theorem~\ref{theorem:public_revenue_preservation}, we have $\Tilde{\R}(0) = \Tilde{\R}(\D)$, $\Tilde{\R}(1) = 0$, and every signal $\S_h$ in Algorithm~\ref{alg:main_algorithm_deadlines} is associated with weight $\weight_h = t_h - t_{h-1}$. Therefore Eq. (\ref{eq:public_total_revenue_in_signals_is_equal_to_continuous_objective}) holds exactly as is. 

We observe that it is still feasible for the seller to ignore the signals by implementing some revenue optimal auction for $\textsf{Deadlines}(\D)$ as the auction for each signal $\S_h$ and achieve $\R(\D)$ as the total revenue. Thus, the theorem follows analogously to the public budget case.
\end{proof}

Theorem~\ref{theorem:deadlines_revenue_preservation} above shows that the process in Algorithm~\ref{alg:main_algorithm_deadlines} preserves the seller's expected revenue. By Lemma~\ref{lemma:deadlines_efficiency_of_scheme}, $\Z^{\ast}_{\D}$ also achieves efficiency, and thus maximizes social welfare. Hence $\Z^{\ast}_{\D}$ must maximize the expected consumer surplus, and the analog of Theorem~\ref{theorem:public_buyer_optimality} on buyer optimality holds for the deadlines case as well:

\begin{theorem}[Buyer optimality for deadlines] \label{theorem:deadlines_buyer_optimality}
In the private deadlines setting, there exists a signaling scheme $\Z^{\ast}_{\D}$ for prior $\D$ that guarantees: (1)  $\W(\Z^{\ast}_{\D}) = \W^\ast(\D)$; (2) $\R(\Z^{\ast}_{\D}) = \R(\D)$; and  (3) $\CS(\Z^{\ast}_{\D}) = \opt(\D) = \W^\ast(\D) - \R(\D)$.
\end{theorem}

\section{Impossibility of Optimal Signaling for Private Budgets}
\label{section:private}
We now consider the setting with private budgets. Recall the program \textsf{Budgets}$(\D)$ from Section~\ref{section:model}, where the type space has valuation and budget, with the IR constraint being interim. 
We show that there are instances with just two budget types in which achieving full social welfare via signaling requires sacrificing almost all consumer surplus. 

\begin{restatable}[]{theorem}{privatecounterexampleefficient}
\label{theorem_private_efficiency_sacrifices_consumer_surplus}
For $n = k = 2$, for any given constant $\epsilon > 0$, there exists a prior $\D$ in which any signaling scheme $\Z$ that achieves efficiency (i.e., item always sells) has $\CS(\Z) \leq \epsilon \cdot \opt(\D)$, where $\opt(\D)= \W^\ast(\D) - \R(\D)$ is the maximum achievable consumer surplus with respect to prior $\D$.
\end{restatable}

Furthermore, a similar proof shows a lower bound of $2$ on approximating the consumer surplus even when it is no longer required that the signaling scheme retains full social welfare.

\begin{restatable}[]{theorem}{privatecounterexamplegeneral}
\label{theorem_private_consumer_surplus_constant_gap}
For $n = k = 2$, for any given constant $\epsilon > 0$, there exists a prior $\D$ in which any signaling scheme $\Z$ has $\CS(\Z) \leq (\frac{1}{2} + \epsilon) \cdot \opt(\D)$.
\end{restatable}

\subsection{Proof of Theorems~\ref{theorem_private_efficiency_sacrifices_consumer_surplus} and~\ref{theorem_private_consumer_surplus_constant_gap}}
\label{app:private}
Both Theorems~\ref{theorem_private_efficiency_sacrifices_consumer_surplus} and~\ref{theorem_private_consumer_surplus_constant_gap} use the following family of instances.

\begin{restatable}[]{definition}{privateinstance}
\label{def:privateinstance}
For any  $M > 1$ and $\delta < \frac{1}{M}$, let the prior $\D_{M, \delta}$ be supported on $\{(v_1, b_1), (v_2, b_2)\}$, where $v_1 = 1, \, b_1 = 1-\delta$, and $v_2 = b_2 = M$. Let $f_1 = \Pr_{(v,b) \sim {\D_{M, \delta}}}[(v,b) = (v_1, b_1)] = 1 - \delta$, and $f_2 = \Pr_{(v,b) \sim {\D_{M, \delta}}}[(v,b) = (v_2, b_2)] = \delta$. The social welfare of $\D_{M, \delta}$ is thus given by 
\begin{equation}
    \W^\ast({\D_{M, \delta}}) = f_1 \cdot v_1 + f_2 \cdot v_2 = (1 - \delta) + \delta M. \label{eq:private_instance_full_welfare}
\end{equation}
\end{restatable}

We now characterize the revenue optimal auction for each $\D_{M, \delta}$. Let $(p_1, x_1)$ and $(p_2, x_2)$ denote the (price, allocation probability) pairs for the two types $(v_1, b_1)$ and $(v_2, b_2)$, respectively. 

\begin{lemma} \label{lemma:private_optimal_revenue_in_prior}
The revenue optimal auction for $\D_{M, \delta}$ (denoted by ${\M^\ast_{M, \delta}}$) has $(p_1, x_1) = (1-\delta, 1-\delta)$ and $(p_2, x_2) = (\delta M + (1-\delta), 1)$, and raises a revenue of $\R({\D_{M, \delta}}) = 1 - \delta + \delta^2 M$. Furthermore, the maximum achievable consumer surplus for $\D_{M, \delta}$ is $\opt({\D_{M, \delta}}) = \delta (1 - \delta) M.$
\end{lemma}

\begin{proof}
The revenue optimal auction for ${\D_{M, \delta}}$ is captured by the following LP:
\begin{align*}
    \max_{p_1, x_1, p_2, x_2} \quad & f_1 \cdot p_1 + f_2 \cdot p_2 & &\\
    \text{s.t.} \quad & v_2 \cdot x_2 - p_2 \geq v_2 \cdot x_1 - p_1, &\, \tag*{(IC)} \\
    & v_i \cdot x_i - p_i \geq 0, &\, i \in \{1,2\}, \tag*{(IR)}\\
    & x_i \in [0,1], &\, i \in \{1,2\}, \tag*{(Feasibility)}\\
    & p_i \leq b_i, &\, i \in \{1,2\}. \tag*{(Budgets)}
\end{align*}

To characterize the optimal auction, we first observe that $x_2$ must be $1$. To see this, notice that given fixed values of $x_1$ and $p_1$, if $x_2 < 1$, then $p_2 < M$ by individual rationality, and it is always favorable to simultaneously increase $x_2$ by $\epsilon$ and $p_2$ by $v_2 \cdot \epsilon$ for some small $\epsilon > 0$. Therefore let $x_2 = 1$. Also notice that the IR and budget constraints for $i = 2$ are implied by other constraints and thus are redundant. Thus we can simplify the LP as
\begin{align*}
    \max_{p_1, x_1, p_2} \quad & f_1 \cdot p_1 + f_2 \cdot p_2 & &\\
    \text{s.t.} \quad & M - p_2 \geq M \cdot x_1 - p_1, &\, \tag*{(IC)} \\
    & x_1 - p_1 \geq 0, &\, \tag*{(IR for $i=1$)}\\
    & x_1, x_2 \in [0,1], &\, \tag*{(Feasibility)}\\
    & p_1 \leq 1-\delta. &\,  \tag*{(Budget for $i=1$)}
\end{align*}

Next, observe that given any fixed $p_1$, the optimal solution will set 
\begin{equation}
    x_1 = p_1, \quad p_2 = M \cdot (1 - p_1) + p_1, \label{eq:private_instance_opt_sol_characterization}
\end{equation}
as this does not violate any constraint and maximizes the objective. The objective is then given by $(1 - \delta) \cdot p_1 + \delta \cdot p_2 = \delta M + (1 - \delta M) \cdot p_1$. Since $\delta < \frac{1}{M}$, $1 - \delta M > 0$, and thus the objective is maximized when $p_1$ is maximized at $1 - \delta$. Therefore, $\M^\ast_{M, \delta}$ has $(p_1, x_1) = (1-\delta, 1-\delta)$ and $(p_2, x_2) = (\delta M + (1-\delta), 1)$, and achieves a revenue of $\R({\D_{M, \delta}}) = \delta M + (1 - \delta M) \cdot (1 - \delta) = 1 - \delta + \delta^2 M$.
By Eq. (\ref{eq:private_instance_full_welfare}), we have $\opt({\D_{M, \delta}}) = \W^\ast({\D_{M, \delta}}) - \R({\D_{M, \delta}}) = \delta (1 - \delta) M$.
\end{proof}

\paragraph{Characterizing the revenue optimal auction for signals.}

Next, we  characterize the optimal auctions for any possible signal $\S \in \Z$, and for each category of $\S$, determine whether there exists some optimal auction that is efficient, i.e., sells the item deterministically. Consider some signal $\S$. Let $g_1 = \Pr_{(v,b) \sim \S}[(v,b) = (1, 1-\delta)]$, and $g_2 = 1 - g_1 = \Pr_{(v,b) \sim \S}[(v,b) = (M, M)]$ be the corresponding probabilities of the types $(1, 1-\delta)$ and $(M,M)$ in $\S$, respectively. We have three cases:
\begin{enumerate}
    \item $g_1 = 1$ and $g_2 = 0$. Then $\S$ contains only the type $(v,b) = (1, 1-\delta)$. This is then a signal with a public budget, and by Lemma~\ref{lemma:public_optimal_auction_for_signals}, there is an optimal auction that posts a price of $1-\delta$ and raises a revenue of $1-\delta$ and is efficient. The consumer surplus achieved by this auction in $\S$ is $1 - (1 - \delta) = \delta$.\footnote{Note that there exist other optimal auctions for this case that do not allocate the item deterministically, but any such auction raises a strictly smaller consumer surplus and thus is not suggested by the signaling scheme.}
    \label{item:private_only_low_type}
    
    \item $g_1 = 0$ and $g_2 = 1$. Similar to the above, $\S$ contains only the type $(v,b) = (M, M)$, and the optimal auction posts a price of $M$ and raises a revenue of $M$, and is efficient. Since the price is equal to the valuation, the consumer surplus is zero.
    \label{item:private_only_high_type}
    
    \item $g_1, g_2 \neq 0$, i.e., both types are included in $\S$. Observe that the proof of Lemma~\ref{lemma:private_optimal_revenue_in_prior} does not use the values of $f_1$ and $f_2$ up to Eq. (\ref{eq:private_instance_opt_sol_characterization}), and thus Eq. (\ref{eq:private_instance_opt_sol_characterization}) holds for the optimal auction for $\S$ as well. Hence in any optimal auction for $\S$ we have $x_1 = p_1 < 1$, i.e., the auction is not efficient.
    \label{item:private_mixed_signals_are_inefficient}
    Fix $p_1$, the objective for this case is then given by $g_1 \cdot p_1 + g_2 \cdot p_2 = g_2 M + \big(g_1 - (M-1)g_2 \big) p_1$. Therefore, we have three cases:
    \begin{enumerate}
        \item \label{item:private_mixed_signals_good_signal}
        If $g_1 \geq (M-1) \cdot g_2$, the above objective is maximized\footnote{If $g_1 = \frac{M-1}{M}$, then the objective is $g_2 M$ for any $p_1 \in [0,1-\delta]$. Hence any auction in this family is optimal. For any $p_1$, the corresponding consumer surplus in this case is $g_2 \cdot (M - p_2) = g_2 \cdot \big( M - (M \cdot (1 - p_1) + p_1) \big) = p_1 (M-1) g_2$. Hence for any $M > 1$ it is maximized by the auction with $p_1 = (1 - \delta)$, i.e., ${\M^\ast_{M, \delta}}$, and the signaling scheme suggests ${\M^\ast_{M, \delta}}$ over all other optimal auctions.} when $p_1$ is maximized at $p_1 = 1 - \delta$, i.e., the auction is exactly ${\M^\ast_{M, \delta}}$. The consumer surplus in this case is given by 
        \begin{equation}
            g_2 \cdot \big( M - (\delta M + (1 - \delta)) \big) = (1-\delta)(M-1)g_2. \label{eq:private_consumer_surplus_in_mixed_signal}
        \end{equation}
        
        \item If $g_1 < (M-1) \cdot g_2$, and the optimal auction for $\S$ sets $p_1 = x_1 = 0$, and accordingly $p_2 = M$. The consumer surplus in this case is $0$. 
        \label{item:private_mixed_signals_bad_signal}
    \end{enumerate}
\end{enumerate}

We are now ready to prove Theorems~\ref{theorem_private_efficiency_sacrifices_consumer_surplus} and~\ref{theorem_private_consumer_surplus_constant_gap} based on the above characterizations.

\paragraph{Proof of Theorem~\ref{theorem_private_efficiency_sacrifices_consumer_surplus}.}

In case~\ref{item:private_mixed_signals_are_inefficient} in the above characterization, we see that no signals in which $g_1, g_2 \neq 0$ can retain efficiency. Hence any signaling scheme that achieves efficiency does not include any signal of this type, and thus must completely separate the two types in all signals created. Accordingly, the lower type $(v_1,b_1) = (1, 1-\delta)$ always pays a price of $1 - \delta$, and the higher type $(v_2,b_2) = (M,M)$ always pays a price of $M$. Thus the total consumer surplus is given by
\[ f_1 \cdot \big(v_1 - (1 - \delta) \big) + f_2 \cdot (v_2 - M) = (1 - \delta) \delta = \frac{\opt({\D_{M, \delta}})}{M},\]
and the theorem follows by taking any instance $\D_{M, \delta}$ with $M = \frac{1}{\epsilon}$ and $\delta < \epsilon$ for any given $\epsilon > 0$. 

\paragraph{Proof of Theorem~\ref{theorem_private_consumer_surplus_constant_gap}.}
Consider any signaling scheme $\Z$. Based on the previous characterizations, we can categorize all possible signals $\S \in \Z$ by the consumer surplus-maximizing auction among all revenue-maximizing auctions for $\S$. We first observe that any signal $\S \in \Z$ that contain both types $(v_1,b_1)$ and $(v_2,b_2)$ but do not admit ${\M^\ast_{M, \delta}}$ as an optimal auction (case~\ref{item:private_mixed_signals_bad_signal}) raises zero consumer surplus. Since the case~\ref{item:private_only_low_type} signal that contains only type $(v_1,b_1)$ raises positive consumer surplus, it is always better (in terms of total consumer surplus) to further separate this signal $\S$ into one that contains only type $(v_1,b_1)$ and another signal that only contains type $(v_2,b_2)$. Hence, any signaling scheme that maximizes consumer surplus does not include any case~\ref{item:private_mixed_signals_bad_signal} signal. Therefore:
\begin{itemize}
    \item (Case~\ref{item:private_only_low_type}) let $g_{11}$ be the total weight of all signals that contains only type $(v_1,b_1)$. 
    \item (Case~\ref{item:private_only_high_type}) let $g_{22}$ be the total weight of all signals that contains only type $(v_2,b_2)$. 
    \item (Case~\ref{item:private_mixed_signals_good_signal}) let $g_{31}$ and $g_{32}$ be the total probability masses of types $(v_1,b_1)$ and $(v_2,b_2)$ put into signals that contain both types and admit ${\M^\ast_{M, \delta}}$ as an optimal auction. Then we have $g_{31} \geq (M-1) \cdot g_{32}$. 
\end{itemize}

Let $M = 2$. The problem of optimizing consumer surplus is captured by the following LP:

\begin{align*}
    \max \quad & \delta \cdot g_{11} + (1-\delta) \cdot g_{32}\\
    \text{s.t.} \quad & g_{11} + g_{31} = 1 - \delta, \\
    & g_{22} + g_{32} = \delta, \\
    & g_{31} \geq g_{32}, \\
    & g_{11}, g_{22}, g_{31}, g_{32} \geq 0.
\end{align*}

Since $M = 2$, we have $\delta < \frac{1}{M} = \frac{1}{2}$, and thus $(1 - \delta) > \delta$. For any value $g_{32} \in [0, \delta]$, it is feasible and optimal to set $g_{31} = g_{32} \leq \delta < (1 - \delta)$ and $g_{11} = 1 - \delta - g_{32}$. This achieves an objective of 
\[ \delta \cdot (1 - \delta - g_{32}) + (1 - \delta) \cdot g_{32} = \delta(1 - \delta) + (1 - 2\delta) \cdot g_{32}.\]
Since $1 - 2\delta > 0$, the above is maximized at $g_{32} = \delta$ for an objective of $\delta (2 - 3\delta)$. Recall that for $M = 2$ the benchmark consumer surplus is $\opt({\D_{M, \delta}}) = 2 \delta (1 - \delta)$. Take $\delta = \frac{1}{2} - \frac{\epsilon}{2}$. Then for all $\Z$ and $\epsilon > 0$ we have
\begin{align*} \CS(\Z) \leq \delta (2 - 3\delta) &= \frac{\delta(1+3\epsilon)}{2} < \frac{\delta(1+3\epsilon+2\epsilon^2)}{2}
=(\frac{1}{2}+\epsilon)\delta(2-2\delta) = (\frac{1}{2}+\epsilon)\opt({\D_{M, \delta}}), 
\end{align*}
which proves the theorem.

\section{Conclusion}
\label{sec:conclusions}

Observe that our positive results hold for two budgeted settings where the optimal auctions with interim and ex-post IR constraints coincide, while our negative result holds for the most general budgeted setting where imposing ex-post IR constraints does reduce optimal revenue.  In effect, our work points to a separation between auctions with ex-post IR constraints, where optimal signaling is possible (public budget or deadlines), and interim IR constraints, where it is not possible (private budget setting). The main open question is whether this separation can be formalized. We conjecture that there is indeed an optimal signaling scheme for the general private budget setting, when the mechanism is required to be ex-post IR instead of interim IR. The key stumbling block is the development of a characterization analogous to Theorem~\ref{theorem:deadlines_optimal_revenue_is_a_distribution_over_posted_price_revenues} for this setting, and we leave this as an interesting open question.

Several other open questions arise from our work. For instance, for private budgets with interim IR,  is there an inefficient signaling scheme that extracts a constant factor of the optimal consumer surplus, thereby providing a positive counterpart to Theorem~\ref{theorem_private_consumer_surplus_constant_gap}? Finally, can our results be generalized to larger type spaces, for instance, spaces with three dimensions such as value, deadline, and amount required?

\bibliographystyle{abbrv}
\bibliography{ref.bib}

\begin{thebibliography}{10}

\bibitem{doubleclick}
{DoubleClick Ad Exchange}.
\newblock \url
  {https://static.googleusercontent.com/media/www.google.com/en//adexchange/AdExchangeOverview.pdf}.

\bibitem{msads}
{Microsoft Ad Exchange}.
\newblock \url {https://ads.microsoft.com/}.

\bibitem{pubmatic}
{Pubmatic Ad Exchange}.
\newblock \url {https://pubmatic.com/}.

\bibitem{verizon}
{Yahoo Ad Exchange}.
\newblock \url
  {https://policies.yahoo.com/us/en/yahoo/terms/yahooadexchange/index.htm}.

\bibitem{Alijani20}
R.~Alijani, S.~Banerjee, K.~Munagala, and K.~Wang.
\newblock The limits of an information intermediary in auction design.
\newblock {\em CoRR}, abs/2009.11841, 2020.

\bibitem{Babichenko21}
Y.~Babichenko, I.~Talgam{-}Cohen, H.~Xu, and K.~Zabarnyi.
\newblock Regret-minimizing bayesian persuasion.
\newblock In {\em Proceedings of the 22$^\text{nd}$ {ACM} Conference on
  Economics and Computation}, page 128, 2021.

\bibitem{Bergemann15}
D.~Bergemann, B.~Brooks, and S.~Morris.
\newblock The limits of price discrimination.
\newblock {\em American Economic Review}, 105(3):921--57, 2015.

\bibitem{bergemann2019information}
D.~Bergemann and S.~Morris.
\newblock Information design: {A} unified perspective.
\newblock {\em Journal of Economic Literature}, 57(1):44--95, 2019.

\bibitem{cai2021duality}
Y.~Cai, N.~R. Devanur, and S.~M. Weinberg.
\newblock A duality-based unified approach to bayesian mechanism design.
\newblock {\em {SIAM} J. Comput.}, 50(3), 2021.

\bibitem{cai2020third}
Y.~Cai, F.~Echenique, H.~Fu, K.~Ligett, A.~Wierman, and J.~Ziani.
\newblock Third-party data providers ruin simple mechanisms.
\newblock {\em Proceedings of the {ACM} on Measurement and Analysis of
  Computing Systems}, 4(1):1--31, 2020.

\bibitem{chakraborty2014persuasive}
A.~Chakraborty and R.~Harbaugh.
\newblock Persuasive puffery.
\newblock {\em Marketing Science}, 33(3):382--400, 2014.

\bibitem{Chawla11}
S.~Chawla, D.~L. Malec, and A.~Malekian.
\newblock Bayesian mechanism design for budget-constrained agents.
\newblock In {\em Proceedings of the 12$^\text{th}$ {ACM} Conference on
  Electronic Commerce}, pages 253--262, 2011.

\bibitem{Che00}
Y.~Che and I.~L. Gale.
\newblock The optimal mechanism for selling to a budget-constrained buyer.
\newblock {\em Journal of Economic Theory}, 92(2):198--233, 2000.

\bibitem{cummings2020algorithmic}
R.~Cummings, N.~R. Devanur, Z.~Huang, and X.~Wang.
\newblock Algorithmic price discrimination.
\newblock In {\em Proceedings of the 2020 {ACM-SIAM} Symposium on Discrete
  Algorithms}, pages 2432--2451, 2020.

\bibitem{Devanur17}
N.~R. Devanur and S.~M. Weinberg.
\newblock The optimal mechanism for selling to a budget constrained buyer: The
  general case.
\newblock In {\em Proceedings of the 2017 {ACM} Conference on Economics and
  Computation}, pages 39--40, 2017.

\bibitem{dughmi2017algorithmic}
S.~Dughmi.
\newblock Algorithmic information structure design: {A} survey.
\newblock {\em {ACM SIGecom Exchanges}}, 15(2):2--24, 2017.

\bibitem{dughmi2016persuasion}
S.~Dughmi, D.~Kempe, and R.~Qiang.
\newblock Persuasion with limited communication.
\newblock In {\em Proceedings of the 2016 {ACM} Conference on Economics and
  Computation}, pages 663--680, 2016.

\bibitem{dughmi2019algorithmic}
S.~Dughmi and H.~Xu.
\newblock Algorithmic bayesian persuasion.
\newblock {\em {SIAM} J. Comput.}, 50(3), 2021.

\bibitem{Fiat16}
A.~Fiat, K.~Goldner, A.~R. Karlin, and E.~Koutsoupias.
\newblock The {FedEx} problem.
\newblock In {\em Proceedings of the 2016 {ACM} Conference on Economics and
  Computation}, pages 21--22, 2016.

\bibitem{ParetoIS}
N.~Haghpanah and R.~Siegel.
\newblock Pareto improving segmentation of multi-product markets, 2020.

\bibitem{Kamenica11}
E.~Kamenica and M.~Gentzkow.
\newblock Bayesian persuasion.
\newblock {\em American Economic Review}, 101(6):2590--2615, 2011.

\bibitem{Laffont96}
J.-J. Laffont and J.~Robert.
\newblock Optimal auction with financially constrained buyers.
\newblock {\em Economics Letters}, 52(2):181--186, 1996.

\bibitem{Milgrom}
P.~Milgrom and I.~Segal.
\newblock Envelope theorems for arbitrary choice sets.
\newblock {\em Econometrica}, 70(2):583--601, 2002.

\bibitem{Myerson81}
R.~B. Myerson.
\newblock Optimal auction design.
\newblock {\em Mathematics of Operational Research}, 6(1):58--73, 1981.

\bibitem{RubinsteinW}
A.~Rubinstein and S.~M. Weinberg.
\newblock Simple mechanisms for a subadditive buyer and applications to revenue
  monotonicity.
\newblock {\em {ACM} Trans. Economics and Comput.}, 6(3-4):19:1--19:25, 2018.

\bibitem{Saxena18}
R.~R. Saxena, A.~Schvartzman, and S.~M. Weinberg.
\newblock The menu complexity of ``one-and-a-half-dimensional'' mechanism
  design.
\newblock In {\em Proceedings of the Twenty-Ninth Annual {ACM-SIAM} Symposium
  on Discrete Algorithms}, pages 2026--2035, 2018.

\bibitem{shen2018closed}
W.~Shen, P.~Tang, and Y.~Zeng.
\newblock A closed-form characterization of buyer signaling schemes in monopoly
  pricing.
\newblock In {\em Proceedings of the 17$^\text{th}$ International Conference on
  Autonomous Agents and MultiAgent Systems}, pages 1531--1539, 2018.

\bibitem{xu2015exploring}
H.~Xu, Z.~Rabinovich, S.~Dughmi, and M.~Tambe.
\newblock Exploring information asymmetry in two-stage security games.
\newblock In {\em Proceedings of the Twenty-Ninth AAAI Conference on Artificial
  Intelligence}, pages 1057--1063, 2015.

\end{thebibliography}

\newpage

\appendix
\section{Proof of Theorem~\ref{theorem:public_optimal_revenue_is_a_distribution_over_posted_price_revenues}}
\label{app:public}
\label{apdx:proof_of_public_optimal_revenue_for_signals}
\publicrevenuedistribution*
Recall that $\A = \D(t)$ is the residual prior with support being a subset of $\{v_1, v_2, \ldots v_n\}$. By Myerson's characterization~\cite{Myerson81}, every feasible solution $(p(\cdot), x(\cdot))$ to $\textsf{PublicContinuous}(\A)$ must satisfy
\begin{align}
    p(v) = v \cdot x(v) - \int_{w=0}^{v} x(w) \dd{w} \label{eq:payment_identity}
\end{align}
for all $v \in [0, v_n]$; furthermore, $p(\cdot)$ and $x(\cdot)$ must be monotone non-decreasing. Define $\area_{x}(v) = \int_{w=0}^{v} x(w) \dd{w}$. Then $\area_{x}(v)$ is convex in $[0, v_n]$, and $\textsf{PublicContinuous}(\A)$ can be simplified to
\begin{align*}
    \text{B3}(\A) \coloneqq \max_{x(\cdot)} \quad & \sum_{j=1}^{n} \Big( f_{\A}(v_j) \cdot \big( v_j \cdot x(v_j) - \area_{x}(v_j) \big) \Big) &\\
    \text{s.t.} \quad & v_n \cdot x(v_n) - \area_{x}(v_n) \leq b, \tag*{(Budget)}\\
    & x(0) \geq 0, \, x(v_n) \leq 1, \, x(v) \text{ is monotone non-decreasing in } [0, v_n]. &
\end{align*}

Let $\supp(\A) = \{w_1, w_2, \ldots, w_m\}$, where $0 < w_1 < \cdots < w_m \leq v_n$, and let $w_0 = 0$ and $w_{m+1} = v_n$, which can possibly be equal to $w_m$. We next simplify B3$(\A)$ one step further, based on the observation that there always exists an optimal solution to B3$(\A)$ in which $x(v)$ is a constant over each interval $[w_i, w_{i+1})$ and thus $\area_{x}(v)$ is piecewise linear. 

\begin{lemma} \label{lemma:public_linear_program_piecewise_constant}
Given an optimal solution $x^\ast(\cdot)$ to B3$(\A)$, there exists another optimal solution $\hat{x}(\cdot)$ in which $\hat{x}(w_{m+1}) = \hat{x}(v_n) = x^\ast(v_n)$, and $\hat{x}(v) = \frac{\int_{w = w_i}^{w_{i+1}} x^\ast(w) \dd{w}}{w_{i+1} - w_i}$ for all $v \in [w_i, w_{i+1})$, and $w_i < v_n$. 
\end{lemma}
\begin{proof}
Let $x^\ast(\cdot)$ be an optimal solution to B3$(\A)$, and $\hat{x}(\cdot)$ as stated above. This implies
\begin{align}
    \area_{\hat{x}}(w_i) = \sum\limits_{i'=0}^{i-1} \left( (w_{i'+1} - w_{i'}) \cdot \frac{\int_{w=w_{i'}}^{w_{i'+1}} x^\ast(w) \dd{w}}{w_{i'+1} - w_{i'}} \right) \tag*{}
    = \sum\limits_{i'=0}^{i-1} \int\limits_{w=w_{i'}}^{w_{i'+1}} x^\ast(w) \dd{w} = \area_{x^\ast}(w_i)
    \label{eq:public_area_preservation_for_values}
\end{align}
for all $i \in \{0,1,\ldots,m+1\}$. Note that $v_n = w_{m+1}$ and $\hat{x}(v_n) = x^\ast(v_n)$. This means the budget constraint is exactly preserved by $\hat{x}(\cdot)$. It is easy to see $\hat{x}(\cdot)$ remains bounded in $[0,1]$ and is non-decreasing; hence it is feasible.

By monotonicity of $x^\ast(\cdot)$, we also have $\hat{x}(w_i) \geq x^\ast(w_i)$ for all $i \in \{0,1,\ldots,m+1\}$. Hence, for all $j \in [1,n]$, either $v_j \notin \supp(\A)$ (so $f_{\A}(v_j) = 0$), or $v_j = w_i$ for some $i$, for which we have
\begin{align*}
    \hat{p}(v_j) = \hat{p}(w_i) = w_i \cdot \hat{x}(w_i) - \area_{\hat{x}}(w_i) \geq w_i \cdot x^\ast(w_i) - \area_{x^\ast}(w_i) = p^\ast(w_i) = p^\ast(v_j), 
\end{align*}
where $\hat{p}(v)$ and $p^\ast(v)$ are the prices associated with value $v$ given by Eq. (\ref{eq:payment_identity}) when the allocation function is $\hat{x}(\cdot)$ and $x^\ast(\cdot)$, respectively. Summing this up for all $v_j \in \supp(\A)$ implies $\hat{x}(\cdot)$ achieves (in B3$(\A)$) an objective value of at least that given by $x^\ast(\cdot)$. Hence $\hat{x}(\cdot)$ is also optimal for B3$(\A)$.
\end{proof}

By Lemma~\ref{lemma:public_linear_program_piecewise_constant}, we only need to consider solutions of the form $\hat{x}(\cdot)$. Specifically, we now have the extra constraint that $x(v) = x(w_i)$ for all $v \in [w_i, w_{i+1})$ (again using $w_0 = 0$ and $w_{m+1} = v_n$). Therefore, for all $i \in [1,m]$ we now have
\begin{align*}  
    \area_{x}(w_i) = \sum\limits_{i'=0}^{i-1} \big( (w_{i'+1} - w_{i'}) \cdot x(w_{i'}) \big).
\end{align*}

Thus, the decision variables are now $\{x_i \coloneqq x(w_i)\}$ for $i \in [0,m]$, and $\text{B3}(\A)$ is equivalent to
\begin{align}
    \text{B4}(\A) \coloneqq \max_{\{x_i\}} \quad & \sum_{i=1}^{m} \Big( f_{\A}(w_i) \cdot \big( w_i \cdot x_i - \area_{x}(w_i) \big) \Big) \tag*{}\\
    \text{s.t.} \quad & w_m \cdot x_m - \area_{x}(w_m) \leq b, \label{B4_constraint_budget}\\
    & 0 \leq x_{0} \leq x_{1} \leq \cdots \leq x_{m} \leq 1. & \tag*{}
\end{align}

Next, we show that assuming $w_1 < b$, there always exists an optimal solution to B4$(\A)$ in which (i) the buyer with the highest valuation $w_m$ of nonzero probability mass in $\A$ gets the item deterministically, and (ii) the buyer with the dummy valuation $w_0$ is never allocated the item.

\begin{lemma} \label{lemma:public_highest_type_gets_item_deterministically}
If $w_1 < b$, there exists an optimal solution $\{ \hat{x}_i \}$ to B4$(\A)$ such that $\hat{x}_{m} = 1$ and $\hat{x}_{0} = 0$.
\end{lemma}
\begin{proof}
Consider any optimal solution $\{ x^\ast_i \}$ to B4$(\A)$ with $x^\ast_{m} < 1$. Let $\epsilon = 1 - x^\ast_{m}$. Consider another solution $\{\Tilde{x}_i\}$ such that $\Tilde{x}_i = x^\ast_i + \epsilon$ for all $i \in [0,m]$. This implies (for all $i \in [0,m]$)
\begin{align*}
    \area_{\Tilde{x}}(w_i) = \area_{x^\ast}(w_i) + \epsilon \cdot w_i
\end{align*}
and thus
\begin{align*}
    w_{i} \cdot \Tilde{x}_{i} - \area_{\Tilde{x}}(w_i) &= w_{i} \cdot \big(x^\ast_{i} + \epsilon \big) - \big( \area_{x^\ast}(w_i) + \epsilon \cdot w_i \big) = w_{i} \cdot x^\ast_{i} - \area_{x^\ast}(w_i).
\end{align*}
Therefore, both the budget constraint (Eq. (\ref{B4_constraint_budget})) and the objective value are exactly preserved. Since $\{ \Tilde{x}_i \} \in [0,1]$ and is monotone non-decreasing, it is feasible and also optimal in B4$(\A)$.

Note that $\Tilde{x}_0 > 0$ and thus $\area_{\Tilde{x}}(w_1) > 0$, i.e., the point $(w_1, 0)$ lies strictly below $\area_{\Tilde{x}}(\cdot)$. Next, pick the smallest $y$ such that
\begin{align}
    \area_{\Tilde{x}}(w_i) = (w_i - w_1) \cdot y
\end{align}
is satisfied for some $i \in [1,m]$, or equivalently, the straight line $\ell_y(\cdot)$ passing $(w_1, 0)$ with slope $y$ is tangent to the curve $\area_{\Tilde{x}}(\cdot)$ at some $v = v'$ (see Figure~\ref{fig:public_align}). Note that the slope of $\area_{\Tilde{x}}(\cdot)$ between $v \in (w_i, w_{i+1})$ is $\Tilde{x}_i$. By convexity of $\area_{\Tilde{x}}(\cdot)$, we have $y \geq \Tilde{x}_1$ (otherwise $\ell_y(\cdot)$ stays strictly below $\area_{\Tilde{x}}(\cdot)$ for all $v \geq w_1$). Since $\area_{\Tilde{x}}(\cdot)$ is piecewise linear in each $v \in [w_i, w_{i+1}]$, $v'$ must lie in $\{w_1, \ldots, w_m\}$.

\begin{figure}[t]
        \centering
        \begin{tikzpicture}[scale=0.5]
        
            \draw[thick,->] (0,0) -- (18,0) node[anchor=north west] {$v$};
            \draw[thick,->] (0,0) -- (0,8.5) node[anchor=south east] {};
            \foreach \x in {1,2,3,4}
                \draw (4 * \x,1pt) -- (4 * \x, -1pt) node[anchor=north] {$w_\x$};
            \filldraw[black] (4, 1) circle (2pt) node[]{};
            \filldraw[black] (8, 2.5) circle (2pt) node[anchor=west]{};
            \filldraw[black] (12, 4.5) circle (2pt) node[anchor=west]{};
            \filldraw[black] (16, 7.8) circle (2pt) node[anchor=west]{};
                
            \draw[black] (0, 0) -- (4, 1) -- (8, 2.5) -- (12, 4.5) -- (16, 7.8) {};
            \node at (17.4,8.2) {$\area_{\Tilde{x}}(v)$};
            \node[red] at (12,3.7) {$\ell_y(\cdot)$};
            \draw[dashed, red] (4, 0) -- (12, 4.5); 

        \end{tikzpicture}
    \caption{$\ell_y(\cdot)$ is tangent to $\area_{\Tilde{x}}(\cdot)$ at $v' = w_3$.}
    \label{fig:public_align}
\end{figure}
    
Let $y' = \min\{y, 1\}$, and consider another solution $\{\hat{x}_i\}$ such that $\hat{x}_0 = 0$, and $\hat{x}_i = \max\{y', \Tilde{x}_i\} \geq \Tilde{x}_i$ for all $1 \leq i \leq m$. Then $\hat{x}_i \in [0,1]$ and is monotone non-decreasing, and we have 
\begin{align}
    \area_{\hat{x}}(v) \leq 
    \begin{cases}
    (v - w_1) \cdot y = \ell_y(v), & v \in [w_1, v'],\\
    \area_{\Tilde{x}}(v), & v \in [v', w_m],
    \end{cases}
    \label{eq:public_aligned_curve}
\end{align}
where the equality holds if $y' = y$. This implies $\area_{\hat{x}}(v) \leq \area_{\Tilde{x}}(v)$ for all $v \in [0, w_m]$, and thus
\begin{align}
    w_{i} \cdot \hat{x}_{i} - \area_{\hat{x}}(w_i) \geq w_{i} \cdot \Tilde{x}_i - \area_{\hat{x}}(w_i) \geq w_{i} \cdot \Tilde{x}_i - \area_{\Tilde{x}}(w_i)
    \label{eq:public_prices_only_go_up}
\end{align}
holds for all $i = 1, 2, \ldots, m$. 

Recall that $\Tilde{x}_m = 1$. Suppose $y \leq 1$ and thus $y' = y$. Then $\hat{x}_m = \Tilde{m} = 1$, and by Eq. (\ref{eq:public_aligned_curve}) we have $\area_{\hat{x}}(w_m) = \area_{\Tilde{x}}(w_m)$. Hence all equalities in Eq. (\ref{eq:public_prices_only_go_up}) must hold for $i = m$. Otherwise, $y > 1$, and $\hat{x}_i = 1$ for all $i \in [1,m]$. Then we have
$w_{m} \cdot \hat{x}_{m} - \area_{\hat{x}}(w_m) = w_m - (w_m - w_1) = w_1 < b$. This implies $\{\hat{x}_i\}$ satisfies the budget constraint in both cases.

Finally, observe that by Eq. (\ref{eq:public_prices_only_go_up}), each term in the objective of B4$(\A)$ is weakly higher when the allocation variables are $\{\hat{x}_i\}$ then when they are $\{\Tilde{x}_i\}$. Hence $\{\hat{x}_i\}$ is optimal to B4$(\A)$.
\end{proof}

Finally, putting the previous lemmas together, we show that the optimal objective of B4$(\A)$ equals the revenue achieved by a distribution over posted-price schemes when $w_1 < b$. 

\paragraph{Completing proof of Theorem~\ref{theorem:public_optimal_revenue_is_a_distribution_over_posted_price_revenues}.}
Assume $w_1 < b$. Then by Lemma~\ref{lemma:public_highest_type_gets_item_deterministically}, there is an optimal solution $\{\hat{x}_i\}$ to B4$(\A)$ such that $\hat{x}_0 = 0$ and $\hat{x}_m = 1$. Let $\hat{x}(w) = \hat{x}_i$ for all $w \in [w_i, w_{i+1})$. Let $\{ w'_1, w'_2, \ldots, w'_{m'} \} \coloneqq \{ w_i \mid \hat{x}_i > \hat{x}_{i-1} \}$ be the set of vertices of $\area_{\hat{x}}(\cdot)$ (equivalently, valuations whose corresponding allocation probability is strictly higher than the previous type of buyer). Let $w'_0 = 0$. Then for $j \in \{1, \ldots, m'\}$ denote allocation function $x_j(\cdot)$ such that 
\[ x_j(w_i) =
\begin{cases}
1, \quad & w_i \geq w'_j;\\
0, \quad & w_i < w'_j.
\end{cases}
\]
This corresponds to posting a price of $w'_j$ and selling to only the buyers with valuation $w_i \geq w'_j$ with probability $1$. Therefore, the revenue generated by $\{x_j(w_i)\}$ is 
\begin{align}
    \sum_{i=1}^{m} \big( f_{\A}(w_i) \cdot p_j(w_i) \big) = \overline{F}_{\A}(w'_j) \cdot w'_j,
    \label{eq:public_decomposition}
\end{align}
where $p_j(w_i) \coloneqq w_i \cdot x_j(w_i) - \int_{w=0}^{w_i} x_j(w) \dd{w}$ is the price associated with buyer with valuation $w_i$. 

Let $\delta_j = \hat{x}(w'_j) - \hat{x}(w'_{j-1})$. Then for all $i = 1, 2, \ldots, m$ we can write
\[ \hat{x}(w_i) = \sum\limits_{j=1}^{m'} \big( \delta_j \cdot x_j(w_i) \big), \quad \hat{p}(w_i) = \sum\limits_{j=1}^{m'} \big( \delta_j \cdot p_j(w_i) \big), \]
where $\hat{p}(w_i) \coloneqq w_i \cdot \hat{x}(w_i) - \int_{w=0}^{w_i} \hat{x}(w) \dd{w}$. Therefore, the optimal objective for \text{B4}$(\A)$ is 
\begin{align*}
  \sum\limits_{i=1}^{m} \big( f_{\A}(w_i) \cdot \hat{p}(w_i) \big) \Big) &=
  \sum\limits_{i=1}^{m} \Big( f_{\A}(w_i) \cdot \sum\limits_{j=1}^{m'} \big( \delta_j \cdot p_j(w_i) \big) \Big)\\
  &= \sum\limits_{j=1}^{m'} \Big( \delta_j \cdot \sum\limits_{i=1}^{m} \big( f_{\A}(w_i) \cdot p_j(w_i) \big) \Big)
  = \sum\limits_{j=1}^{m'} \big( \delta_j \cdot \overline{F}_{\A}(w'_j) \cdot w'_j \big),
\end{align*}
where the last equality follows from Eq. (\ref{eq:public_decomposition}). Finally, notice that since $\hat{x}(w'_{m'}) = \hat{x}(w_m) = 1$, we have $\sum_{j=1}^{m'} \delta_j = \hat{x}(w_m) - \hat{x}(w'_0)  = 1$. This proves the theorem. \hfill \qedsymbol

\section{Proof of Theorem~\ref{theorem:deadlines_optimal_revenue_is_a_distribution_over_posted_price_revenues}}
\label{app:deadlines}
\label{apdx:proof_of_deadlines_optimal_revenue_for_signals}

\deadlinesrevenuedistribution*

\subsubsection{Simplification of $\textsf{DeadlinesContinuous}(\A)$ into discrete set of decision variables}

\sloppy Recall $\A = \D(t)$ is the residual prior with support $\supp(\A) \subseteq \supp(\D) = \{v_1, v_2, \ldots, v_n\} \times \{1, 2, \ldots, k\}$. By Myerson's characterization~\cite{Myerson81}, every feasible solution $\big(\{p_j(\cdot)\}, \{x_j(\cdot)\}\big)$ for $\textsf{DeadlinesContinuous}(\A)$ must satisfy
\begin{align}
    p_j(v) \coloneqq v \cdot x_j(v) - \int_{w=0}^{v} x_j(w) \dd{w} \label{eq:fedex_payment_identity}
\end{align}
for all $v \in [0, v_n]$ and $j \in [1,k]$, and all $\{x_j(\cdot)\}$ are monotone non-decreasing. Therefore, for each $j \in [1,k]$, the area below the allocation curve $x_j(\cdot)$, denoted $\area_{x}(v,j) \coloneqq \int_{w=0}^{v} x_j(w) \dd{w}$, is convex in $[0, v_n]$. Therefore, $\textsf{DeadlinesContinuous}(\A)$ can be simplified to
\begin{align*}
    \text{D3}(\A) \coloneqq \max_{\{x_j(\cdot)\}} \quad & \sum_{j=1}^{k} \left( \Pr_{(v,d) \sim \A}[d = j] \cdot \sum_{i=1}^{n} \Big( f_{\A_j}(v_i) \cdot \big( v_i \cdot x_j(v_i) - \area_{x}(v_i, j) \big) \Big) \right)&\\
    \text{s.t.} \quad & \area_{x}(v, j) \geq \area_{x}(v, j-1), \quad \forall v \in [0, v_n], \, 2 \leq j \leq k, \tag*{(Inter-deadline IC)}\\
    & x_j(0) \geq 0, \, x_j(v_n) \leq 1, \, x_j(v) \text{ is monotone non-decreasing in } [0, v_n] \text{ for all } j. &
\end{align*}

Let $\values(\A) = \{w_1, w_2, \ldots, w_m\}$, where $0 < w_1 < \cdots < w_m \leq v_n$, and let $w_0 = 0$ and $w_{m+1} = v_n$, which can possibly be equal to $w_m$. Analogous to Lemma~\ref{lemma:public_linear_program_piecewise_constant}, we next simplify D3$(\A)$ such that for all $j$, $x_j(v)$ is piecewise constant, and $\area_{x}(v,j)$ is piecewise linear, in each interval $v \in [w_i, w_{i+1})$.

\begin{lemma} \label{lemma:deadlines_linear_program_piecewise_constant}
Given an optimal solution $\{x^\ast_j(\cdot)\}$ to D3$(\A)$, there exists another optimal solution $\{\hat{x}_{j}(\cdot)\}$ in which for all $j = 1, \ldots, k$ we have $\hat{x}_j(w_{m+1}) = \hat{x}_j(v_n) = x^\ast_j(v_n)$, and $\hat{x}_j(v) = \frac{\int_{w = w_i}^{w_{i+1}} x^\ast_j(w) \dd{w}}{w_{i+1} - w_i}$ for all $v \in [w_i, w_{i+1})$, and $w_i < v_n$. 
\end{lemma}

\begin{proof}
Let $\{x^\ast_j(\cdot)\}$ be an optimal solution to D3$(\A)$, and let $\{\hat{x}_j(\cdot)\}$ as stated above. By the same arguments in the proof of Lemma~\ref{lemma:public_linear_program_piecewise_constant}, this implies
\begin{align}
    \area_{\hat{x}}(w_i, j) = \area_{x^\ast}(w_i, j)
    \label{eq:deadlines_piecewise_constant_allocation_probabilities_preserves_utility}
\end{align}
holds for all $i \in [0,m+1]$ and $j \in [1,k]$. By Eq. (\ref{eq:deadlines_piecewise_constant_allocation_probabilities_preserves_utility}) and the inter-deadline IC constraints in D3$(\A)$, we have $\area_{\hat{x}}(w_i, j) = \area_{x^\ast}(w_i, j) \geq \area_{x^\ast}(w_i, j-1) = \area_{\hat{x}}(w_i, j-1)$ for all $i \in [0,m+1]$ and $j \in [2,k]$. Since both $\area_{\hat{x}}(v, j)$ and $\area_{\hat{x}}(v, j-1)$ are continuous and piecewise linear in each $v \in [w_i, w_{i+1})$, this implies $\area_{\hat{x}}(v, j)$ must be weakly above $\area_{x^\ast}(v, j-1)$ at all $v \in [w_0, w_{m+1}] = [0, v_n]$, i.e., $\{\hat{x}_j(\cdot)\}$ satisfies all inter-deadline IC constraints. Similar to the case in the proof of Lemma~\ref{lemma:public_linear_program_piecewise_constant}, it can be verified that each $\{\hat{x}_j(\cdot)\}$ is monotone non-decreasing and in $[0,1]$. Hence $\{\hat{x}_j(\cdot)\}$ is feasible in D3$(\A)$. 

For all $j$, by the monotonicity of $x^\ast_j(\cdot)$ and the construction of $\hat{x}_j(\cdot)$, we have $\hat{x}_j(w_i) \geq x^\ast_j(w_i)$ for all $i \in [0, m]$. Hence, for all $i' \in [1,n]$, either $v_{i'} \notin \values(\A)$ (so $f_{\A_j}(v_{i'}) = 0$), or $v_{i'} = w_i$ for some $i'$. Analogous to Eq. (\ref{eq:public_prices_only_go_up}), for all such $v_{i'} \in \values(\A)$, the price $\hat{p}_j(v_{i'})$ (given by the allocation variables $\{\hat{x}_j(\cdot)\}$) is at least the old price $p^\ast_j(v_{i'})$ (given by $\{x^\ast_j(\cdot)\}$). Since this holds for all $j$, $\{\hat{x}_j(\cdot)\}$ gives in D3$(\A)$ an objective at least that given by $\{x^\ast_j(\cdot)\}$, and thus is optimal.
\end{proof}

By Lemma~\ref{lemma:deadlines_linear_program_piecewise_constant}, we only need to consider solutions in the form of $\{\hat{x}_j(\cdot)\}$, i.e., for all $j \in [1,k]$ we require that $x_j(v) = x_j(w_i)$ for all $v \in [w_i, w_{i+1})$ (again using $w_0 = 0$ and $w_{m+1} = v_n$). Therefore, the decision variables are now $\{x_{ij} \coloneqq x_j(w_i)\}$ for $i \in [0,m]$ and $j \in [1,k]$. 

Throughout the rest of the proof, given the set of decision variables $\{x_{ij}\}$, for all $i_1 \leq i_2 \in [0,m]$ and $j \in [1,k]$, we denote the area below the curve $x_j(v)$ from $v = w_{i_1}$ to $v = w_{i_2}$ as
\[ \area_{x} (w_{i_1},w_{i_2},j) \coloneqq \sum_{i'=i_1}^{i_2-1} \big( (w_{i'+1} - w_{i'}) \cdot x_{i'j} \big). \]
For convenience we also let $\area_{x}(w_i,j) \coloneqq \area_{x}(w_0,w_i,j)$. 
Therefore, D3$(\A)$ is equivalent to
\begin{align*}
    \text{D4}(\A) \coloneqq \max_{x_{ij}} \quad & \sum_{j=1}^{k} \left( \Pr_{(v,d) \sim \A}[d = j] \cdot \sum_{i=1}^{m} \Big( f_{\A_j}(w_i) \cdot \big( w_i \cdot x_{ij} - \area_{x}(w_i, j) \big) \Big) \right) \span\span\\
    \text{s.t.} \quad & \area_{x}(w_i,j) \geq \area_{x}(w_i,j-1), \quad & i \in [1,m], \, j \in [2,k];\\
    & 0 \leq x_{0j} \leq x_{1j} \leq \cdots \leq x_{mj} \leq 1, \quad & j \in [1,k].
\end{align*}

\subsubsection{Aligning the Allocation Curves}

In the following, we aim to further characterize the optimal solution to D4$(\A)$. We will show the following theorem that shows the allocation curves are aligned below the lower envelope:
\begin{theorem} \label{theorem:deadlines_align_process}
There exists a set of optimal allocation variables $\{\hat{x}_{ij}\}$ to D4$(\A)$ in which for all $(i,j)$ such that $\underline{F}_{\A_{j'}}(w_i) = 0$ for all $j < j' \leq k$, we have $\hat{x}_{i'j'} = \hat{x}_{i'j}$ for all $i' \leq i$ and $j' \in [j+1,k]$.
\end{theorem}

We will prove the above theorem by showing a process that gradually converts any optimal solution $\{x^\ast_{ij}\}$ to such a solution $\{\hat{x}_{ij}\}$ so that the feasibility is preserved and the objective is weakly improved. Consider the following process, which we refer to as $\mathsf{Align}(\A, \{x_{ij}\}, \hat{i}, \hat{j})$, that takes the prior $\A$, the allocation variables $\{x_{ij}\}$, a valuation type $\hat{i} \in [0,m]$, and a deadline $\hat{j} \in [1, k-1]$ as inputs, and modifies $\{x_{ij}\}$ into another solution $\{x'_{ij}\}$ as follows:

\begin{enumerate}
    \item Let $x'_{ij} = x_{ij}$ for all $i,j$, except for $i \geq \hat{i}$ and $j = \hat{j}+1$. 
    \item Initialize $x'_{\hat{i}(\hat{j}+1)} = x_{\hat{i}\hat{j}}$.
    \item If $\hat{i} \neq m$, let $y_{\hat{i}\hat{j}}$ be the smallest $y \geq 0$ such that 
        \begin{align}
            \area_{x}(w_{\hat{i}+1}, \hat{j}) + (w_i - w_{\hat{i}+1}) \cdot y = \area_{x}(w_{i'}, \hat{j}+1) \label{eq:deadlines_align_process}
        \end{align}
         is satisfied for some $i' > \hat{i}$; equivalently, $y_{\hat{i}\hat{j}}$ is the smallest $y$ such that the straight line $\ell_y(\cdot)$ passing $\big(w_{\hat{i}+1}, \area_{x}(w_{\hat{i}+1}, \hat{j})\big)$ with slope $y$ is tangent to $\area_{x}(v, \hat{j}+1)$ at some $v = v' = w_{i'}$ (see Figure~\ref{fig:deadlines_align}). 
    \item Let $y'_{\hat{i}\hat{j}} = \min \{ y_{\hat{i}\hat{j}}, 1 \}$, and let $x'_{i(\hat{j}+1)} = \max \{ y'_{\hat{i}\hat{j}}, x_{i(\hat{j}+1)} \} \geq x_{i(\hat{j}+1)}$ for all $i \in [\hat{i}+1,m]$.
    \item Return $\{x'_{ij}\}$ as the modified solution.
\end{enumerate}

\begin{figure}[t]
        \centering
        \begin{tikzpicture}[scale=0.5]
        
            \draw[thick,->] (0,0) -- (20,0) node[anchor=north west] {$v$};
            \draw[thick,->] (0,0) -- (0,11) node[anchor=south east] {};
            \foreach \x in {2,3,4}
                \draw (4 * \x,1pt) -- (4 * \x, -1pt) node[anchor=north] {$w_\x$};
            \draw (4, 1pt) -- (4, -1pt) node[anchor=north] {$w_1 = w_{\hat{i}}$};
            \filldraw[black] (4, 1.5) circle (2pt) node[]{};
            \filldraw[black] (8, 3.5) circle (2pt) node[anchor=west]{};
            \filldraw[black] (12, 6) circle (2pt) node[anchor=west]{};
            \filldraw[black] (16, 9) circle (2pt) node[anchor=west]{};
            
            \filldraw[blue] (4, 0.5) circle (2pt) node[]{};
            \filldraw[blue] (8, 2) circle (2pt) node[anchor=west]{};
            \filldraw[blue] (12, 4) circle (2pt) node[anchor=west]{};
            \filldraw[blue] (16, 6.5) circle (2pt) node[anchor=west]{};
                
            \draw[black, thick] (0, 0) -- (4, 1.5) -- (8, 3.5) -- (12, 6) -- (16, 9) -- (18, 11) {};
            \node at (6, 4) {$\area_{x}(v,\hat{j}+1)$};
            \draw[blue] (0, 0) -- (4, 0.5) -- (8, 2) -- (12, 4) -- (16, 6.5) -- (18, 7.8) {};
            \node[blue] at (19,5.5) {$\area_{x}(v,\hat{j}) = \area_{x'}(v,\hat{j})$};
            
            \draw[dashed, thick, red] (0, 0) -- (4, 0.5) -- (8, 2) -- (16, 9) -- (18, 11) {};
            \draw[dashed] (8, 2) -- (16, 9); 
            \node[red] at (14.4, 6.6) {$\ell_{y_{\hat{i}\hat{j}}}$};
            \node[red] at (18.4, 8.6) {$\area_{x'}(v,\hat{j}+1)$};
            \draw[->, red] (12,2) -- (10.2, 3.8);
            \node[red] at (13.7,1.6) {slope $ = y_{\hat{i}\hat{j}}$};

        \end{tikzpicture}
    \caption{Illustration of $\mathsf{Align}(\A, \{x_{ij}\}, \hat{i}, \hat{j})$ with $\hat{i} = 1$. In step (3), $y_{\hat{i}\hat{j}}$ is the smallest $y$ such that the straight line $\ell_y(\cdot)$ passing $\big(w_2, \area_{x}(w_2, \hat{j})\big)$ with slope $y$ is tangent to $\area_{x}(v, \hat{j}+1)$. The tangent point is $v' = w_4$, and thus $\mathsf{Align}(\A, \{x_{ij}\}, \hat{i}, \hat{j})$ sets $x'_{2(\hat{j}+1)} = x'_{3(\hat{j}+1)} = y_{\hat{i}\hat{j}}$, and $x'_{i(\hat{j}+1)} = x_{i(\hat{j}+1)}$ for $i \geq 4$. The resulting new area curve $\area_{x'}(v, \hat{j}+1)$ shown in red remains convex, and lies between $\area_{x}(v, \hat{j})$ and $\area_{x}(v, \hat{j}+1)$ at all $v$.}
    \label{fig:deadlines_align}
\end{figure}
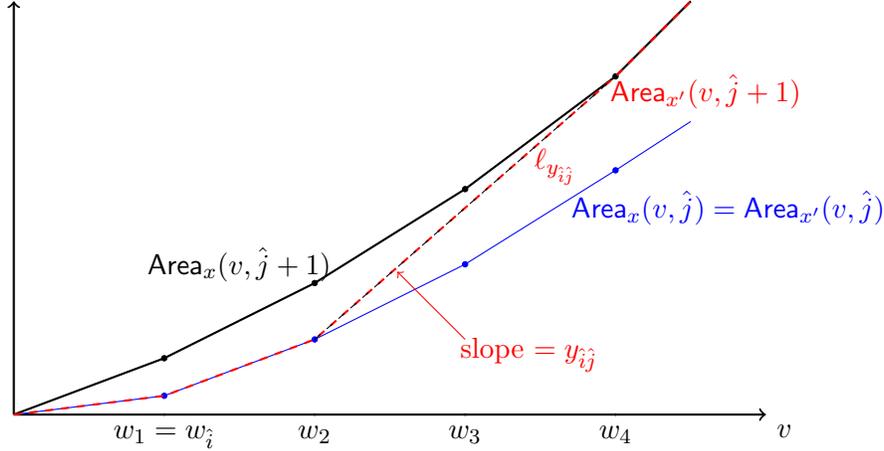

Recall from Definition~\ref{def:lower_envelope} that $\hat{i}_j \coloneqq \max\{i : \underline{F}_{\A_{j'}}(v_i) = 0 \quad \forall j' \in [j,k] \}$ is the highest type $i$ such that no buyer with valuation at most $v_i$ and deadline at least $j$ exists in $\A$. 

Define $\mathsf{Align}(\A, \{x_{ij}\})$ as follows: Apply $\mathsf{Align}(\A, \{x_{ij}\}, \hat{i}, \hat{j})$ in increasing order of $\hat{j}$ from $\hat{j} = 1$ to $k-1$, and for each $\hat{j}$, in increasing order of $\hat{i}$ from $\hat{i} = 0$ to $\hat{i}_{\hat{j}}$. We next show that (i) If $\{x_{ij}\}$ input to $\mathsf{Align}(\A, \{x_{ij}\}, \hat{i}, \hat{j})$ is feasible, so is the constructed $\{x'_{ij}\}$ (via Lemmas~\ref{lemma:deadlines_align_preserves_feasibility_at_each_step}-~\ref{lemma:deadlines_align_preserves_feasibility}); and (ii) If the input $\{x_{ij}\}$ is optimal for D4$(\A)$, then so is $\{\hat{x}_{ij}\}$ (via Lemma~\ref{lemma:deadlines_aligned_solution_is_optimal}). 

Consider the execution of $\mathsf{Align}(\A, \{x_{ij}\})$, and the instant just before the execution of  $\mathsf{Align}(\A, \{x_{ij}\}, \hat{i}, \hat{j})$ for some $\hat{i} \in [0,m]$ and $\hat{j} \in [1, k-1]$. The next lemma assumes the following inductive scenario:

\noindent {\bf Inductive Scenario $P(\hat{i}, \hat{j})$:} Assume the input  $\{x_{ij}\}$ to $\mathsf{Align}(\A, \{x_{ij}\}, \hat{i}, \hat{j})$ is feasible for D4$(\A)$. Further, by the execution of $\mathsf{Align}(\A, \{x_{ij}\})$, we have $x_{i(\hat{j}+1)} = x_{i\hat{j}}$ holds for all $i \in [0, \hat{i}-1]$.

\begin{lemma} \label{lemma:deadlines_align_preserves_feasibility_at_each_step}
Assuming $P(\hat{i}, \hat{j})$, the procedure $\mathsf{Align}(\A, \{x_{ij}\}, \hat{i}, \hat{j})$ is well-defined, and guarantees that the output $\{x'_{ij}\}$ of $\mathsf{Align}(\A, \{x_{ij}\}, \hat{i}, \hat{j})$ is feasible.
\end{lemma}
\begin{proof}
Assume $P(\hat{i}, \hat{j})$ holds. Then $(w_{\hat{i}}, \area_{x}(w_{\hat{i}}, \hat{j}))$ lies weakly below the curve $\area_{x}(v, \hat{j}+1)$. By the convexity of $\area_{x}(v, \hat{j}+1)$, $y_{\hat{i}\hat{j}}$ uniquely exists, i.e., $\mathsf{Align}(\A, \{x_{ij}\}, \hat{i}, \hat{j})$ is well-defined. Furthermore, we have $y_{\hat{i}\hat{j}} \geq x_{(\hat{i}+1)\hat{j}}$, since the latter is the slope of $\area_{x}(v, \hat{j})$ at $v = (w_{\hat{i}+1})^+$. Therefore we have $x'_{(\hat{i}+1)(\hat{j}+1)} = \min\{ y_{\hat{i}\hat{j}}, 1\} \geq x_{(\hat{i}+1)\hat{j}} \geq x_{\hat{i}\hat{j}} = x'_{\hat{i}(\hat{j}+1)}.$
By $P(\hat{i}, \hat{j})$, we also have $x'_{i(\hat{j}+1)} = x_{i\hat{j}}$ for $i \in [0, \hat{i}]$, which implies $\{x'_{i(\hat{j}+1)}\}$ is non-decreasing for $i \in [0, \hat{i}]$. Finally, notice that $\{x'_{i(\hat{j}+1)}\}$ is the max of two non-decreasing functions for $i \in [\hat{i}+1, m]$. Combining the above, $\{x'_{i(\hat{j}+1)}\}$ is non-decreasing for $i \in [0,m]$. 

Suppose $y_{\hat{i}\hat{j}} \leq 1$, and thus $y'_{\hat{i}\hat{j}} = y_{\hat{i}\hat{j}}$. Recall that $\ell_{y_{\hat{i}\hat{j}}}$ is tangent to $\area_{x}(v,\hat{j}+1)$ at $v = w_{i'}$. By convexity of $\area_{x}(v,\hat{j}+1)$, if $i' < m$, we must have $y_{\hat{i}\hat{j}} \leq x_{i' (\hat{j}+1)}$, 
so we have
\[ x'_{i(\hat{j}+1)} =
\begin{cases}
x_{i \hat{j}}, & i \in [0, \hat{i}], \\
y_{\hat{i}\hat{j}}, & i \in [\hat{i}+1, i'-1], \\
x_{i (\hat{j}+1)}, & i \in [i', m-1],
\end{cases}
\quad
\area_{x'} (v, \hat{j}+1) =
\begin{cases}
\area_{x}(v, \hat{j}), & v \in [0, w_{\hat{i}+1}), \\
\ell_{y_{\hat{i}\hat{j}}}(v), & v \in [w_{\hat{i}+1}, w_{i'}), \\
\area_{x}(v, \hat{j}+1), & v \in [w_{i'}, w_m],
\end{cases}
\]
where the interval $[i', m-1]$ may be empty when $i' = m$. Otherwise, if $y_{\hat{i}\hat{j}} > 1$, then we have
\[ x'_{i(\hat{j}+1)} =
\begin{cases}
x_{i \hat{j}}, & i \in [0, \hat{i}], \\
1, & i \in [\hat{i}+1, m-1],
\end{cases}
\quad
\area_{x'} (v, \hat{j}+1) =
\begin{cases}
\area_{x}(v, \hat{j}), & v \in [0, w_{\hat{i}+1}), \\
\ell_{1}(v) < \area_{x}(v, \hat{j}+1), & v \in [w_{\hat{i}+1}, w_m].
\end{cases}
\]

Recall that $y_{\hat{i}\hat{j}} \geq x_{(\hat{i}+1)\hat{j}}$ and both $\ell_{y_{\hat{i}\hat{j}}}(v)$ and $\ell_1(v)$ goes through $\big(w_{\hat{i}+1)}, \area_{x}(w_{\hat{i}+1)}, \hat{j})\big)$. Therefore, in every interval in both cases above, $\area_{x'} (v, \hat{j}+1)$ is equal to some function upper bounded by $\area_{x} (v, \hat{j}+1)$ and lower bounded by $\area_{x}(v, \hat{j}) = \area_{x'}(v, \hat{j})$. This implies
\begin{align}
    \area_{x'} (w_i, \hat{j}) \leq \area_{x} (w_i, \hat{j}+1) \leq \area_{x} (w_i, \hat{j}+1) \leq \area_{x} (w_i, \hat{j}+2) = \area_{x'} (w_i, \hat{j}+2)
\end{align}
holds for all $i \in [0,m]$, i.e., $\{x'_{ij}\}$ satisfies all inter-deadline IC constraints (as $\mathsf{Align}(\A, \{x_{ij}\}, \hat{i}, \hat{j})$ only modifies the $(\hat{j}+1)$-th curve). Finally, since $y'_{\hat{i}\hat{j}} \leq 1$, we have $\{x'_{i(\hat{j}+1)}\} \in [0,1]$, and the lemma follows.
\end{proof}

The next lemma now follows by induction on the execution of $\mathsf{Align}(\A, \{x_{ij}\})$ using the observation from the proof of Lemma~\ref{lemma:deadlines_align_preserves_feasibility_at_each_step} that any $\area$ curve cannot increase in an iteration.

\begin{lemma} \label{lemma:deadlines_align_preserves_feasibility}
Suppose the input $\{x_{ij}\}$ to $\mathsf{Align}(\A, \{x_{ij}\})$ is feasible to \textsf{D4}$(\A)$. Then $\{\hat{x}_{ij}\}$ (the output of $\mathsf{Align}(\A, \{x_{ij}\})$ is also feasible in D4$(\A)$. Furthermore, for all $i \in [0,m], j \in [1,k]$ it holds that
\begin{align}
    \area_{\hat{x}} (w_i, j) \leq \area_{x} (w_i, j).
\end{align}
\end{lemma}

The next lemma shows the optimality of $\{\hat{x}_{ij}\}$ if the input to $\mathsf{Align}(\A, \{x_{ij}\})$ is optimal. 

\begin{lemma} \label{lemma:deadlines_aligned_solution_is_optimal}
If the input $\{x_{ij}\}$ to $\mathsf{Align}(\A, \{x_{ij}\})$ is optimal for D4$(\A)$, then the output $\{\hat{x}_{ij}\}$ is also optimal for D4$(\A)$.
\end{lemma}
\begin{proof}
Recall that $\hat{p}_{ij} \coloneqq w_i \cdot \hat{x}_{ij} - \area_{\hat{x}}(i,j)$ and $p_{ij} \coloneqq w_i \cdot x_{ij} - \area_{x}(i,j)$
are the prices associated with the buyer with valuation $w_i$ and deadline $j$ when the allocation variables are $\{\hat{x}_{ij}\}$ and $\{x_{ij}\}$, respectively. 
Since $\{\hat{x}_{i1}\} = \{x_{i1}\}$, we have $\hat{p}_{i1} = p_{i1}$ for all $i$. 

Fix some $j \in [2,k]$. For all $i \in [0, \hat{i}_{j-1}]$ we have $f_{\A_j}(w_i) = 0$ and thus $f_{\A_j}(w_i) \cdot \hat{p}_{ij} = f_{\A_j}(w_i) \cdot p_{ij} = 0$. For $i \in [\hat{i}_{j-1}+1, m]$, we have $\hat{x}_{ij} \geq x_{ij}$ by the execution of $\mathsf{Align}(\A, \{x_{ij}\})$, and thus
\begin{align*}
    \hat{p}_{ij} &= w_i \cdot \hat{x}_{ij} - \area_{\hat{x}}(i,j) \geq w_i \cdot x_{ij} - \area_{\hat{x}}(i,j) \geq w_i \cdot x_{ij} - \area_{x}(i,j) = p^\ast_{ij},
\end{align*}
where the last inequality follows from Lemma~\ref{lemma:deadlines_align_preserves_feasibility}. Therefore, for all $i \in [1,m]$ we have $f_{\A_j}(w_i) \cdot \hat{p}_{ij} \geq f_{\A_j}(w_i) \cdot p_{ij}$. Since this holds for all $j \in [1,k]$, it implies that $\{\hat{x}_{ij}\}$ achieves (in D4$(\A)$) at least the objective achieved by $\{x_{ij}\}$. By Lemma~\ref{lemma:deadlines_align_preserves_feasibility}, it is also feasible, hence optimal.
\end{proof}

\paragraph{Proof of Theorem~\ref{theorem:deadlines_align_process}.} Consider the output $\{\hat{x}_{ij}\}$ by applying $\mathsf{Align}(\A, \{x^\ast_{ij}\})$ to any arbitrary optimal solution $\{x^\ast_{ij}\}$. For any $i,j$ such that $\underline{F}_{\A_{j'}}(w_i) = 0$ for all $j' \in [j+1,k]$, we have $i \leq \hat{i}_{j} \leq \hat{i}_{j+1} \leq \cdots \leq \hat{i}_{k-1}$, and thus $\hat{x}_{ij} = \hat{x}_{i(j+1)} = \cdots = \hat{x}_{ik}$. \hfill \qedsymbol

\subsubsection{Matching the Allocation Curves to the Lower Envelope}
Next, we show that there exists some optimal solution to D4$(\A)$ that satisfies an additional set of properties required for proving Theorem~\ref{theorem:deadlines_optimal_revenue_is_a_distribution_over_posted_price_revenues}. We will prove the following theorem that shows the allocation curve only has breakpoints at values with non-zero mass in the lower envelope. Furthermore, the allocation for the highest type is deterministic.

\begin{theorem} \label{theorem:deadlines_complete_characterization_of_optimal_revenue_for_signals}
For any prior $\A = \D(t)$, there is an optimal solution $\{\hat{x}_{ij}\}$ to D4$(\A)$ that satisfies the following:
\begin{enumerate}
    \item $\hat{x}_{0j} = 0$ for all $j \in [1,k]$;
    \item If $(w_a, r) \in \textsf{LE}(\A)$, then $\hat{x}_{ij} = \hat{x}_{ir}$ for all $i \in [0,a]$ and $j \ge r$;
    \item If $(w_a, r)$ and $(w_b, r')$ are consecutive points on $\textsf{LE}(\A)$ where $a < b$ and $r' \ge r$, then $\hat{x}_{ij} = \hat{x}_{aj}$ for all $i \in [a, b)$ and $j \ge r$;
    \item $\hat{x}_{mk} = 1$.
\end{enumerate}
\end{theorem}
\begin{proof}

\begin{figure}[t]
        \centering
        \begin{tikzpicture}[scale=0.5]
        
            \draw[thick,->] (0,0) -- (20,0) node[anchor=north west] {$v$};
            \draw[thick,->] (0,0) -- (0,11) node[anchor=south east] {};
            \foreach \x in {1,2,3,4}
                \draw (4 * \x,1pt) -- (4 * \x, -1pt) node[anchor=north] {$w_\x$};
            \filldraw[black] (4, 0) circle (2pt) node[]{};
            \filldraw[black] (4, 0.5) circle (2pt) node[]{};
            \filldraw[black] (8, 2.5) circle (2pt) node[anchor=west]{};
            \filldraw[black] (12, 5) circle (2pt) node[anchor=west]{};
            \filldraw[black] (16, 9) circle (2pt) node[anchor=west]{};
            
            \draw[thick, dashed, black] (0, 0) -- (4, 0.5) -- (8, 2.5) -- (12, 5) -- (16, 9) -- (18, 10.75) {};
            \node[black] at (17.5, 8) {$\area_{x^\ast}(v,k)$};
            
            \draw[thick, dashed, red] (0, 0) -- (4, 0) -- (8, 2.5) -- (16, 9) -- (18, 11) {};
            \node[red] at (14, 9) {$\area_{\hat{x}}(v,k)$};

        \end{tikzpicture}
    \caption{Illustration of the modification process in proof of Theorem~\ref{theorem:deadlines_complete_characterization_of_optimal_revenue_for_signals}. The black and red curves illustrate $\area_{x^\ast}(v,k)$ and $\area_{\hat{x}}(v,k)$ (corresponding to some optimal solution $\{x^\ast_{ij}\}$ and the final solution $\{\hat{x}_{ij}\}$, respectively) for the last deadline $k$. $(w_4,k)$ is the only value-deadline pair in the support of $\A$, and no value-deadline pair $(w_3, j)$ is on the lower envelope of $\A$. This implies (i) $\hat{x}_{0j} = 0$, and thus $\area_{\hat{x}}(w_1,j) = 0$; (ii) $\area_{\hat{x}}(v,k)$ is linear in $v \in [w_2, w_4]$ without a breakpoint at $v = w_3$; and (iii) $\hat{x}_{ik} = 1$ for all $i \geq 4$, i.e., $\area_{\hat{x}}(v,k)$ has a slope of $+1$ in $[w_4, w_m]$.}
    \label{fig:deadlines_fine_tuning}
\end{figure}
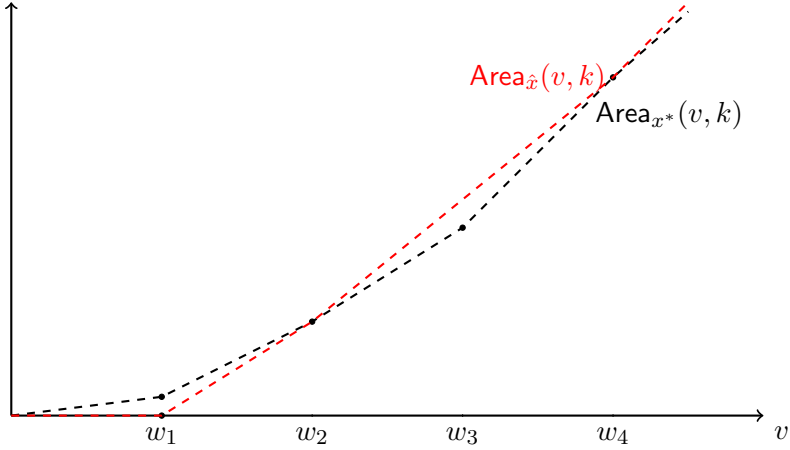
Similar to the proof of Theorem~\ref{theorem:deadlines_align_process}, we show a process that gradually modifies an optimal solution until it satisfies all the properties above.

Consider any optimal solution $\{ x^\ast_{ij} \}$ to D4$(\A)$ in which $x^\ast_{01} > 0$. Define another solution $\{x'_{ij}\}$ such that $x'_{01} = 0$ and $x'_{ij} = x^\ast_{ij}$ for all $(i,j) \neq (0,1)$. Since the curve $\area_{x'}(v,1)$ is upper bounded by $\area_{x^\ast}(v,1)$, it is also upper bounded by $\area_{x^\ast}(v,2) = \area_{x'}(v,2)$. Hence all inter-deadline IC constraints are preserved by $\{x'_{ij}\}$, so it is feasible. Since the objective of D4$(\A)$ does not depend on $x_{01}$, $\{x'_{ij}\}$ achieves (in D4$(\A)$) exactly the objective achieved by $\{ x^\ast_{ij} \}$, and thus is optimal. Take $\{\Bar{x}_{ij}\}$ to be the output of $\mathsf{Align}(\A, \{x'_{ij}\})$. Since $\{\Bar{x}_{ij}\}$ satisfies Theorem~\ref{theorem:deadlines_align_process}, it satisfies properties (1) and (2). 

Now consider another solution $\{\Tilde{x}_{ij}\}$ obtained by the following:
\begin{itemize}
    \item Initialize $\Tilde{x}_{ij} = \Bar{x}_{ij}$ for all $i,j$.
    \item For all consecutive $(w_a, r), (w_b, r') \in \textsf{LE}(\A)$ such that $b > a + 1$ and $r' \ge r$, set 
    \[ \Tilde{x}_{ij} = \frac{\area_{\Bar{x}}(w_a, w_b, r)}{w_b - w_a} \]
    for all $i \in [a, b)$ and $j \in [r,k]$.
\end{itemize}

Clearly, $\{\Tilde{x}_{ij}\}$ satisfies property (3) while preserving properties (1) and (2) from $\{\Bar{x}_{ij}\}$; for property (1), observe that $w_1 \in \values(\A)$, and thus $\Tilde{x}_{0j} = \Bar{x}_{0j}$ for all $j$. Hence it suffices to prove that $\{\Tilde{x}_{ij}\}$ is feasible and optimal to $\textsf{D4}(\A)$.

Consider any consecutive $(w_a, r), (w_b, r') \in \textsf{LE}(\A)$ such that $b > a + 1$ and $r' \ge r$. For all $j \ge r$ we have
\begin{align}
\area_{\Tilde{x}}(v,j) =
\begin{cases}
    \area_{\Bar{x}}(v,j), & v \notin (w_a, w_b),\\
    \frac{w_b - w_i}{w_b - w_a} \cdot \area_{\Bar{x}}(w_a, j) + \frac{w_i - w_a}{w_b - w_a} \cdot \area_{\Bar{x}}(w_b, j), & v \in (w_a, w_b).
\end{cases}
\label{eq:deadlines_flattening}
\end{align}

By convexity of $\area_{\Bar{x}}(v,j)$, this implies
\begin{align}
    \area_{\Tilde{x}}(v,j) \geq \area_{\Bar{x}}(v,j)  \label{eq:deadlines_flattening_increases_area}
\end{align}
for all $v \in [0, v_n]$ and $j \ge r$. By taking $j = r$ in Eq. (\ref{eq:deadlines_flattening_increases_area}) above, all inter-deadline IC constraints between deadlines $(r-1)$ and $r$ are preserved in $\{\Tilde{x}_{ij}\}$. Furthermore, since $(w_a, r), (w_b, r') \in \textsf{LE}(\A)$, we have $\area_{\Bar{x}}(v,j) = \area_{\Bar{x}}(v,r)$ for all $v \in [0, b]$, which implies all inter-deadline IC constraints for $i \in [0,b]$ and $j \geq r$ are satisfied by $\{\Tilde{x}_{ij}\}$ as well. Repeating this argument for all $a, b$, and $r$ shows all inter-deadline IC constraints in $\textsf{D4}(\A)$ are preserved in $\{\Tilde{x}_{ij}\}$. Since $\{\Tilde{x}_{ij}\} \in [0,1]$ and is monotone non-decreasing for all $j$, it is feasible.

We next show $\{\Tilde{x}_{ij}\}$ is optimal. Fix any deadline $r$, and recall that $\Tilde{p}_{ir} = w_i \cdot \Tilde{x}_{ir} - \area_{\Tilde{x}}(w_i,r)$ and $\Bar{p}_{ir} = w_i \cdot \Bar{x}_{ir} - \area_{\Bar{x}}(w_i,r)$ are the prices associated with the buyer $(v,d) = (w_i, r)$ given the allocation variables $\{\Tilde{x}_{ij}\}$ and $\{\Bar{x}_{ij}\}$, respectively. By Eq. (\ref{eq:deadlines_flattening}) and the construction of $\{\Tilde{x}_{ij}\}$, for all $i$ such that $f_{\A_r}(w_i) > 0$, it holds that (i) $\Tilde{x}_{ir} = \Bar{x}_{ir}$, (ii) $\area_{\Tilde{x}}(w_i,r) = \area_{\Bar{x}}(w_i,r)$, and thus (iii) $\Tilde{p}_{ir} = \Bar{p}_{ir}$. Thus the revenue raised by $\{\Tilde{x}_{ij}\}$ for deadline $r$ equals that of $\{\Bar{x}_{ij}\}$. Since this is true for all $r$, $\{\Tilde{x}_{ij}\}$ raises the same revenue as $\{\Bar{x}_{ij}\}$, and thus is optimal.

As the last step, we modify $\{\Tilde{x}_{ij}\}$ such that it satisfies property (4) and remains optimal. Let $k' \leq k$ be the largest $j$ such that some $(w_i, j) \in \textsf{LE}(\A)$, and let $m'$ be the largest such $i$. Consider the following solution $\{\hat{x}_{ij}\}$:
\begin{align*}
    \hat{x}_{ij} \coloneqq
    \begin{cases}
        1, \quad & i \in [m', m], j \in [k', k];\\
        \Tilde{x}_{ij}, \quad & \text{otherwise.}
    \end{cases}
\end{align*}

Therefore, $\hat{x}_{mk} = 1$, i.e., property (4) is satisfied. It is easy to see properties (1) - (3) are preserved by $\{\hat{x}_{ij}\}$. Hence it suffices to show $\{\hat{x}_{ij}\}$ is feasible and optimal to $\textsf{D4}(\A)$.

The curve $\area_{\hat{x}}(v,k')$ is weakly above $\area_{\Tilde{x}}(v,k')$, which by the feasibility of $\{\Tilde{x}_{ij}\}$, is weakly above $\area_{\Tilde{x}}(v,k'-1) = \area_{\hat{x}}(v,k'-1)$. Hence $\{\hat{x}_{ij}\}$ satisfies all inter-deadline IC constraints (note that for all $j > k'$ we have $\area_{\hat{x}}(v,j) = \area_{\hat{x}}(v,k')$ for all $v \in [0, v_n]$). Since $\{\hat{x}_{ij}\} \in [0,1]$ and remains monotone non-decreasing for all $j$, it is feasible.

Again let $\Tilde{p}_{ij} = w_i \cdot \Tilde{x}_{ij} - \area_{\Tilde{x}}(w_i,j)$ and $\hat{p}_{ij} = w_i \cdot \hat{x}_{ij} - \area_{\hat{x}}(w_i,j)$ denote the price associated with the buyer $(v,d) = (w_i, j)$ given the allocation variables $\{\Tilde{x}_{ij}\}$ and $\{\hat{x}_{ij}\}$, respectively. We have $\hat{p}_{ij} = \Tilde{p}_{ij}$ for all $j < k'$ or $i < m'$, and the only other $(w_i, j)$ such that $f_{\A_j}(w_i) > 0$ is $(w_{m'}, k')$. Note that the two curves $\area_{\hat{x}}(v,k')$ and $\area_{\Tilde{x}}(v,k')$ are identical for $v \in [0, w_{m'}]$. Thus, we have
\begin{align*}
    \hat{p}_{m'k'} &= w_{m'} \cdot \hat{x}_{m'k'} - \area_{\hat{x}}(w_{m'},k') = w_{m'} - \area_{\hat{x}}(w_{m'},k') \tag*{(since $\hat{x}_{m'k'} = 1$)}\\
    &= w_{m'} - \area_{\Tilde{x}}(w_{m'},k') \geq w_{m'} \cdot \Tilde{x}_{m'k'} - \area_{\Tilde{x}}(w_{m'},k') = \Tilde{p}_{m'k'},
\end{align*}
where the inequality follows from $\Tilde{x}_{m'k'} \leq 1$. Therefore, we have $\hat{p}_{ij} \geq \Tilde{p}_{ij}$ for all $(w_i, j)$ such that $f_{\A_j}(w_i) > 0$, which implies $\{\hat{x}_{ij}\}$ raises at least the revenue achieved by $\{\Tilde{x}_{ij}\}$. Hence $\{\hat{x}_{ij}\}$ is optimal, and the theorem follows. 
\end{proof}

\subsubsection{Completing proof of Theorem~\ref{theorem:deadlines_optimal_revenue_is_a_distribution_over_posted_price_revenues}}
Consider an optimal solution $\{\hat{x}_{ij}\}$ to D4$(\A)$ satisfying Theorem~\ref{theorem:deadlines_complete_characterization_of_optimal_revenue_for_signals}. For all $j \in [1,k]$, let $\hat{x}_j(w) = \hat{x}_{ij}$ for all $w \in [w_i, w_{i+1})$. Let $w_0 = 0$. Then for $i' \in \{1, \ldots, m\}$ denote allocation function $x^j_{i'}(\cdot)$ such that 
\[ x^j_{i'}(w_i) =
\begin{cases}
1, \quad & w_i \geq w_{i'};\\
0, \quad & w_i < w_{i'}.
\end{cases}
\]
This corresponds to posting a price $w_{i'}$ and selling to only the buyers with valuation $w_i \geq w_{i'}$ with probability $1$. Therefore, we have
\begin{align}
    \sum_{i=1}^{m} \big( f_{\A_j}(w_i) \cdot p^j_{i'}(w_i) \big) &= \overline{F}_{\A_j}(w_{i'}) \cdot w_{i'}, \label{eq:deadlines_decomposition}
\end{align}
where $p^j_{i'} \coloneqq w_i \cdot x^j_{i'}(w_i) - \int_{w=0}^{w_i} x^j_{i'}(w) \dd{w}$ is the price associated with the buyer with valuation $w_i$ and deadline $j$ given the allocation curve $x^j_{i'}(\cdot)$.

Let $\delta^j_{i'} = \hat{x}_j(w_{i'}) - \hat{x}_j(w_{i'-1}) = \hat{x}_{i'j} - \hat{x}_{(i'-1)j}$ for all $j \in [1,k]$ and $i' \in [1,m]$. Then for all $i \in [1,m]$ we can write
\[ \hat{x}_j(w_i) = \sum\limits_{i'=1}^{m} \big( \delta^j_{i'} \cdot x^j_{i'}(w_i) \big), \quad \hat{p}_j(w_i) = \sum\limits_{i'=1}^{m} \big( \delta^j_{i'} \cdot p^j_{i'}(w_i) \big), \]
where $\hat{p}_j(w_i) \coloneqq w_i \cdot \hat{x}_j(w_i) - \int_{w=0}^{w_i} \hat{x}(w) \dd{w}$. Therefore, any optimal auction for \text{D4}$(\A)$ raises a revenue of
\begin{align*}
   \sum\limits_{i=1}^{m} \big( f_{\A_j}(w_i) \cdot \hat{p}_j(w_i) \big) &=  \sum\limits_{i=1}^{m} \Big( f_{\A_j}(w_i) \cdot \sum\limits_{i'=1}^{m} \big( \delta^j_{i'} \cdot p^j_{i'}(w_i) \big) \Big)\\
   &= \sum\limits_{i'=1}^{m} \Big( \delta^j_{i'} \cdot \sum\limits_{i=1}^{m}  \big( f_{\A_j}(w_i) \cdot p^j_{i'}(w_i) \big) \Big) = \sum\limits_{i'=1}^{m} \Big( \delta^j_{i'} \cdot \overline{F}_{\A_j}(w_{i'}) \cdot w_{i'} \Big)
\end{align*}
for the buyer with deadline $j$, where the last equality follows from Eq. (\ref{eq:deadlines_decomposition}). Summing this up for all $j \in [1,k]$ (and weighting each term by $\Pr_{(v,d) \sim \A}[d = j]$) gives the expression in the theorem.

Consider any arbitrary $(w_a, r) \in \textsf{LE}(\A)$. By property 2 in Theorem~\ref{theorem:deadlines_complete_characterization_of_optimal_revenue_for_signals}, we have $\hat{x}_{ij} = \hat{x}_{ir}$ for all $i \in [0, a]$ and $j \in [r, k]$. Therefore for all $j \in [r,k]$ we have
$\delta^{j}_{a} = \hat{x}_{aj} - \hat{x}_{(a-1)j} = \hat{x}_{ar} - \hat{x}_{(a-1)r} = \delta^r_a$.

Next, consider all consecutive $(w_a, r), (w_b, r') \in \textsf{LE}(\A)$ such that $a < b$, $r' \ge r$. Since $f_{\A_j}(w_i) = 0$ for all $i \in (a,b)$ and $j \ge r$, by property 3 in Theorem~\ref{theorem:deadlines_complete_characterization_of_optimal_revenue_for_signals}, we have $\hat{x}_{ij} = \hat{x}_{aj}$ for all $i \in [a,b)$ and $j \in [r,k]$. Hence for all $i \in (a,b)$ and $j \in [r,k]$ we have 
$\delta^{j}_{i} = \hat{x}_{ij} - \hat{x}_{(i-1)j} = \hat{x}_{aj} - \hat{x}_{aj} = 0$.

Finally, notice that $\hat{x}_{0k} = 0$ and $\hat{x}_{mk} = 1$. Therefore we have 
\begin{align*}
    \sum\limits_{i'=1}^{m} \delta^{k}_{i'} &= \sum\limits_{i'=1}^{m} \big( \hat{x}_{i'k} - \hat{x}_{(i'-1)k} \big) = \hat{x}_{mk} - \hat{x}_{0k}  = 1.
\end{align*}

The last property then follows from observing that $\hat{x}_{ik} = 0$ for every $i$ such that $\mathbbm{1}^{\textsf{LE}(\A)}_{i} = 0$. \hfill \qedsymbol

\end{document}